\documentclass[aps,prd, notitlepage, onecolumn,superscriptaddress,floatfix,letterpaper,nofootinbib, longbibliography]{revtex4-1}

\usepackage{amsmath, amssymb, graphics, setspace}
\usepackage{amsmath, amsthm, amssymb,slashed}
\usepackage{pifont}

\newcommand{\xmark}{\ding{55}}

\usepackage[usenames, dvipsnames]{color}
\usepackage[svgnames]{xcolor}
\usepackage[colorlinks,citecolor=RoyalBlue, urlcolor=RoyalBlue, linkcolor=RoyalBlue ]{hyperref} 

\usepackage[normalem]{ulem}

\newcommand{\cred}[1]{\textcolor{black}{#1}}

\usepackage{booktabs}

\definecolor{mygray}{gray}{0.6}

\usepackage{upgreek}
\usepackage{bbm}

\newenvironment{myfont}[2][]{\csname#2\endcsname[#1]}{}

\usepackage{slashed}
\usepackage[makeroom]{cancel}
\usepackage[normalem]{ulem}
\usepackage{soul}
\newcommand{\stkout}[1]{\ifmmode\text{\sout{\ensuremath{#1}}}\else\sout{#1}\fi}

\usepackage{sseq}
\usepackage[all,cmtip]{xy}
\usepackage{tikz-cd}
\usepackage{tikz}
\usetikzlibrary{fit,matrix}
\usetikzlibrary{decorations.markings}
\usetikzlibrary{tikzmark,decorations.pathreplacing,positioning}

\usepackage{amsfonts}

\newcommand{\bea}{\begin{eqnarray}}
\newcommand{\eea}{\end{eqnarray}}
\def\be{\begin{equation}}
\def\ee{\end{equation}}

\newcommand{\e}{\hspace{1pt}\mathrm{e}}

\newcommand{\ii}{\hspace{1pt}\mathrm{i}\hspace{1pt}}

\definecolor{red}{rgb}{1,0,0}
\definecolor{blue}{rgb}{0,0,1}
\definecolor{dblue}{rgb}{0,0,0.4}
\definecolor{green}{rgb}{0,1,0}
\definecolor{black}{rgb}{0,0,0}
\definecolor{white}{rgb}{1,1,1}

\definecolor{brn}{rgb}{.8,.4,.0}
\definecolor{redo}{rgb}{1,.5,.0}
\definecolor{ddgrn}{rgb}{0,0.4,0}
\definecolor{dgrn}{rgb}{0,0.55,0}
\definecolor{dbl}{rgb}{0,0,0.5}

\usepackage[bbgreekl]{mathbbol}
\usepackage{amscd}

\newcommand{\Z}{\mathbb{Z}}
\newcommand{\C}{\mathbb{C}}
\newcommand{\R}{\mathbb{R}}
\newcommand{\M}{\mathbb{M}}
\newcommand{\bH}{\mathbb{H}}

\newcommand{\eq}[1]{eq.~(\ref{#1})}

\newcommand{\prt}{\partial}

\newcommand{\bpm}{\begin{pmatrix}}
\newcommand{\epm}{\end{pmatrix}}
\newcommand{\bmm}{\begin{matrix}}
\newcommand{\emm}{\end{matrix}}

\newcommand{\cL}{ {\cal L} }

\def\Z{{\mathbb{Z}}}

\def\R{{\mathbb{R}}}
\def\C{{\mathbb{C}}}

\def \Hom{\operatorname{Hom}}

\def \H{\operatorname{H}}

\def \Z{\mathbb{Z}}

\newcommand {\emptycomment}[1]{}

\def\B{\mathrm{B}}

\newcommand{\SO}{{\rm SO}}
\newcommand{\Spin}{{\rm Spin}}
\newcommand{\U}{{\rm U}}
\newcommand{\SU}{{\rm SU}}

\newcommand{\Pin}{{\rm Pin}}

\newcommand{\rO}{{\rm O}}

\usepackage{centernot}
\newcommand{\nn}{{\nonumber}}
\newcommand{\Sec}[1]{Sec.~\ref{#1}} 
\newcommand{\App}[1]{App.~\ref{#1}}

\newcommand{\diag}{{\rm diag}}
\usepackage{enumitem} 
\usepackage{mathtools,varwidth}

\newcommand{\Fig}[1]{Fig.~\ref{#1}} 
\newcommand{\Table}[1]{Table \ref{#1}}

\usepackage{datetime}

\newtheorem{theorem}{Theorem}[section]

\newtheorem{lemma}[theorem]{Lemma}

\newtheorem{definition}[theorem]{Definition}
\newtheorem{proposition}[theorem]{Proposition}

\def\bD{{\mathbb{D}}}

\newcommand{\rC}{{\rm C}}
\newcommand{\rP}{{\rm P}}
\newcommand{\rT}{{\rm T}}
\newcommand{\rR}{{\rm R}}
\newcommand{\rF}{{\rm F}}

\newcommand{\GL}{{\rm GL}}

\newcommand{\Cl}{\mathrm{Cl}}

\newcommand{\CCl}{\mathbb{C}\mathrm{l}}

\newcommand{\Ad}{\mathrm{Ad}}

\newcommand{\Aut}{\mathrm{Aut}}

\newcommand{\jw}[1]{\textcolor{black}{#1}}

\begin{document}


\title{C-R-T Fractionalization, 
Fermions, 
and Mod 8 Periodicity
} 
%

\author{Zheyan Wan}
\email{wanzheyan@bimsa.cn}
\affiliation{Beijing Institute of Mathematical Sciences and Applications, Beijing 101408, China}

\author{Juven Wang 
}
\email{jw@cmsa.fas.harvard.edu}
\homepage{http://sns.ias.edu/~juven/}
\affiliation{Center of Mathematical Sciences and Applications, Harvard University, MA 02138, USA}
\affiliation{London Institute for Mathematical Sciences, Royal Institution, W1S 4BS, UK}

\author{Shing-Tung Yau}
\email{styau@tsinghua.edu.cn}

\affiliation{Yau Mathematical Sciences Center, Tsinghua University, Beijing 100084, China}
\affiliation{Beijing Institute of Mathematical Sciences and Applications, Beijing 101408, China}
\author{Yi-Zhuang You}
\email{yzyou@physics.ucsd.edu}
\affiliation{Department of Physics, University of California San Diego, CA 92093, USA}
\begin{abstract} 

Charge conjugation (C), mirror reflection (R), time reversal (T), and fermion parity $(-1)^{\rm F}$ are basic discrete spacetime and internal symmetries of the Dirac fermions. In this article, we determine the group, called the C-R-T fractionalization, which is a group extension of $\mathbb{Z}_2^{\rm C}\times\mathbb{Z}_2^{\rm R}\times\mathbb{Z}_2^{\rm T}$ by the fermion parity $\mathbb{Z}_2^{\rm F}$, and its extension class in all spacetime dimensions $d$, for a single-particle fermion theory. For Dirac fermions, with the canonical CRT symmetry $\mathbb{Z}_2^{\rm CRT}$, the C-R-T fractionalization has two possibilities that only depend on spacetime dimensions $d$ modulo 8, which are order-16 nonabelian groups, including the famous Pauli group. For Majorana fermions, we determine the R-T fractionalization in all spacetime dimensions $d=0,1,2,3,4\mod8$, an order-8 abelian or nonabelian group. For Weyl fermions, we determine the C or T fractionalization in all even spacetime dimensions $d$, which is an order-4 abelian group. We only have an order-2 $\mathbb{Z}_2^{\rm F}$ group for Majorana-Weyl fermions.
We determine the maximal number of linearly independent Dirac and Majorana mass terms and construct them explicitly. We also discuss how the conventional Dirac and Majorana mass terms break the symmetries C, R, or T. 
We study the domain wall dimensional reduction of the fermions and their C-R-T fractionalization: from $d$-dim Dirac to $(d-1)$-dim Dirac or Weyl and from $d$-dim Majorana to $(d-1)$-dim Majorana or Majorana-Weyl.

\end{abstract}


\maketitle


 \tableofcontents

\section{Introduction and summary}

The study of symmetries and their breaking patterns is a cornerstone of modern theoretical physics, providing deep insights into the laws that govern the behavior of fundamental particles and fields. 
Among the symmetries of particular interest are charge conjugation (C), mirror reflection (R), and time reversal (T), each representing a fundamental operation that can be applied to physical systems
\cite{SchwingerPhysRev.82.914, Pauli1955, Pauli1957, Luders1954, LUDERS19571, StreaterWightman1989}. 
Charge conjugation inverts the charge of particles, mirror reflection flips one of the spatial coordinates, and time reversal switches the direction of time's flow. These discrete symmetries, and their conservation or violation, have profound implications for our understanding of the universe at its most fundamental level.

The time-reversal T and mirror reflection R are discrete spacetime \jw{symmetry transformations} that are not part of the proper orthochronous continuous Lorentz symmetry group $\SO^+(d-1,1)$. As passive transformations on the spacetime coordinates $(t,x)=(t,x_1,\dots,x_{d-1})$,
\bea
\rT(t,x)\rT^{-1}&=&(-t,x),\nn\\
\rR_j(t,x)\rR_j^{-1}&=&(t,x_1,\dots,x_{j-1},-x_j,x_{j+1},\dots,x_{d-1}),
\eea
where T flips the time coordinate and R$_j$ flips the spatial coordinate $x_j$.

However, the charge conjugation C is not a spacetime symmetry but an internal symmetry. It can only manifest itself under an active transformation on a particle or field, such as a complex-valued Lorentz scalar boson $\phi$ (which is a function of the spacetime coordinates). The C flips between particles and anti-particles, or more generally, between energetic excitations and anti-excitations
\bea
\rC(\text{excitations})\rC^{-1}=(\text{anti-excitations})
\eea
which involves the complex conjugate (denoted *). The active transformations act on this Lorentz scalar boson $\phi$ as
\bea
\rC\phi\rC^{-1}&=&\phi^*,\nn\\
\rR_j\phi\rR_j^{-1}&=&\phi(t,x_1,\dots,x_{j-1},-x_j,x_{j+1},\dots,x_{d-1}),\nn\\
\rT\phi\rT^{-1}&=&\phi(-t,x).
\eea

In this scalar boson example, C, R$_j$, and T are $\Z_2$ involutions and commutative with each other. Thus the C-R-T symmetry form a direct product group $G_{\phi}=\Z_2^{\rC}\times\Z_2^{\rR_j}\times\Z_2^{\rT}$.
However, this is not the case in general. We find that all these discrete C, R, and T symmetries can form a nonabelian finite group for Dirac fermions in all spacetime dimensions.

In recent times, the concept of symmetry fractionalization has gained traction, offering a nuanced view of how symmetries behave in certain quantum systems. 
By \cite{2109.15320}, the symmetry fractionalization means the following: the matter field is not in the linear representation of the original symmetry group $G$, but in the projective representation of $G$ and in the linear representation of the extended group $\tilde{G}$. Namely, there is a group extension $1\to N\to\tilde{G}\to G\to1$ where $N$ is the normal subgroup and $G$ is the quotient group of the total group $\tilde{G}$, so $\tilde{G}/N=G$. A well-known example is the gapped 1+1d isospin-1 Haldane chain with $G=\SO(3)$ symmetry \cite{Affleck1988nt1989QuantumSpinChainsHaldaneGap}, whose 0+1d boundary can have a two-fold degenerate isospin-1/2 doublet with $\tilde{G}=\SU(2)$ symmetry and $N=\Z_2$. This doublet is in a projective representation of $G=\SO(3)$, also in a linear representation of $\tilde{G}=\SU(2)$.

In this article, we will find the analogous C-R-T symmetry fractionalization. In contrast to a scalar boson's $G_{\phi}=\Z_2^{\rC}\times\Z_2^{\rR_j}\times\Z_2^{\rT}$, we discover an order 16 nonabelian $\tilde{G}_{\psi}={G}_{\rm D}$ for a Dirac fermion in all spacetime dimensions. 

\begin{definition}
    The C-R$_j$-T fractionalization for a Dirac fermion in any spacetime dimension $d$ is the group $\tilde{G}_{\psi}={G}_{\rm D}$ generated by the symmetry transformations C, R$_j$, and T which is a group extension
\bea\label{eq:extension}
1\to\Z_2^{\rF}\to \tilde{G}_{\psi}={G}_{\rm D}\to G_{\phi}=\Z_2^{\rC}\times\Z_2^{\rR_j}\times\Z_2^{\rT}\to1.
\eea 
\end{definition}
The fermion parity $\Z_2^{\rF}$ generated by $(-1)^{\rF}:\psi\mapsto-\psi$ plays a crucial role in the group extension \eqref{eq:extension}. 
The Dirac fermion $\psi$ is in a projective representation of $G_{\phi}$, also in an (anti-)linear representation of $\tilde{G}_{\psi}$. (It is anti-linear because $\tilde{G}_{\psi}$ contains the anti-linear time-reversal symmetry.)

This group extension \eqref{eq:extension} is classified by its extension class in $\H^2(\B(\Z_2^{\rC}\times\Z_2^{\rR_j}\times\Z_2^{\rT}),\Z_2^{\rF})=\Z_2^6$. We find that only 8 of the 64 kinds of group extensions are realized as the C-R$_j$-T fractionalization with the canonical CRT for a Dirac fermion. See \Table{tab:CPRT-group-dd-1} and \Table{tab:CPRT-group-dd-2}. We require that the combined CRT symmetry is canonical \cite{HasonKomargodskiThorngren1910.14039,1910.14046, Wang:2019obe1910.14664}. This constraint rules out some possibilities of the C-R$_j$-T fractionalization. 

A Majorana fermion is a Dirac fermion with the trivial charge conjugation symmetry C. A Majorana fermion only exists in spacetime dimension $d=0,1,2,3,4\mod8$.
\begin{definition}
    The R$_j$-T fractionalization for a Majorana fermion in any spacetime dimension $d=0,1,2,3,4\mod8$ is the group ${G}_{\rm M}$ generated by the symmetry transformations R$_j$ and T which is a group extension
\bea\label{eq:extension-Majorana}
1\to\Z_2^{\rF}\to {G}_{\rm M}\to \Z_2^{\rR_j}\times\Z_2^{\rT}\to1.
\eea 
\end{definition}
This group extension \eqref{eq:extension-Majorana} is classified by its extension class in $\H^2(\B(\Z_2^{\rR_j}\times\Z_2^{\rT}),\Z_2^{\rF})=\Z_2^3$. We find that only 4 of the 8 kinds of group extensions are realized as the R$_j$-T fractionalization for a Majorana fermion. See \Table{tab:CPRT-group-dd-Majorana}.

If and only if the spacetime dimension $d$ is even, a Dirac fermion can be decomposed as two Weyl fermions. For $d=2\mod4$, a Weyl fermion has only the C symmetry. For $d=0\mod4$, a Weyl fermion has only the T symmetry. See \Table{table:CPT}.
\begin{definition}
    The C or T fractionalization for a Weyl fermion in any even spacetime dimension $d$ is the group ${G}_{\rm W}$ generated by the symmetry transformations C or T which is a group extension
\bea\label{eq:extension-Weyl}
1\to\Z_2^{\rF}\to {G}_{\rm W}\to \Z_2^{\rC\text{ or }\rT}\to1.
\eea 
\end{definition}
This group extension \eqref{eq:extension-Weyl} is classified by its extension class in $\H^2(\B(\Z_2^{\rC\text{ or }\rT}),\Z_2^{\rF})=\Z_2$. We find that all of the 2 kinds of group extensions are realized as the C or T fractionalization for a Weyl fermion. See \Table{table:CT-group-dd-Weyl}.

A Majorana-Weyl fermion is a fermion which is both Majorana and Weyl. A Majorana-Weyl fermion only exists in spacetime dimension $d=2\mod8$. A Majorana-Weyl fermion has only the trivial C symmetry. 
\begin{definition}
    The trivial C fractionalization for a Majorana-Weyl fermion in any spacetime dimension $d=2\mod8$ is the group ${G}_{\rm MW}=\Z_2^{\rF}$.
\end{definition}

In \cite{2109.15320}, the second author studied the analogous C-P-T fractionalization in 3+1d and 1+1d. In this article, we consider the mirror reflection R instead of the parity P. The reasons are as follows:
\begin{itemize}
    \item 
The P is a composition of several mirror reflections, so R is more basic.
In odd dimensional spacetime, $\rP\subset \SO(d-1)\subset\SO^+(d-1,1)$, so P is not an independent discrete symmetry. We should replace P with R. For example, the
CPT theorem \cite{SchwingerPhysRev.82.914, Pauli1955, Pauli1957, Luders1954, LUDERS19571, StreaterWightman1989} 
 should be called the CRT theorem \cite{Witten1508.04715, Freed1604.06527} in any general spacetime dimension.

   \item 
The CRT theorem asserts that the symmetry $H$ of any relativistic QFT extends to a larger symmetry $G$ such that $H$ is a subgroup of index 2 of $G$, elements of $G\setminus H$ act antilinearly, and CRT is an element of $G\setminus H$ \cite{Freed1604.06527}. So there is a natural CRT symmetry for any relativistic QFT but it is not unique.
 However, there is a unique (up to rotations) canonical CRT which is obtained by breaking all the internal global symmetries, it is unbreakable \cite{HasonKomargodskiThorngren1910.14039}.
The domain wall dimensional reduction requires that the combined symmetry CRT is canonical \cite{HasonKomargodskiThorngren1910.14039,1910.14046, Wang:2019obe1910.14664}, so we will only consider the case when CRT is canonical. We observe that CPT can not be canonical in some dimensions. For example, the CPT of the C-P-T fractionalization in 3+1d in \cite{2109.15320} can not be modified to be canonical. We will see later that CRT can be chosen to be canonical in any dimension.
So we consider R instead of P.
\end{itemize}

This article aims to dissect the intricate concept of C-R-T fractionalization, a group extension that has profound implications in various spacetime dimensions and for different types of fermionic fields. Our results are multifaceted and encompass a range of scenarios:
\begin{itemize}
    \item 
    We determine the C-R-T fractionalization with the canonical CRT and its extension class for Dirac fermions in any spacetime dimension, and we find an 8-periodicity. See Theorem \ref{main}, \Table{tab:CPRT-group-dd-1}, and \Table{tab:CPRT-group-dd-2}.
    
\item 
We determine the R-T fractionalization for Majorana fermions in any spacetime dimension $d=0,1,2,3,4\mod8$. See \Table{tab:CPRT-group-dd-Majorana}.

\item
We determine the C or T fractionalization for Weyl fermions in any even spacetime dimension. See \Table{table:CT-group-dd-Weyl}.

\item 
We determine the maximal number of linearly independent Dirac and Majorana mass terms and construct them explicitly. We also discuss how the conventional Dirac and Majorana mass terms break the symmetries C, R, or T. See \Table{table:CRT-break-Dirac} and \Table{table:CRT-break-Majorana}.

\item
\jw{We study the domain wall dimensional reduction of the C-R-T fractionalization and R-T fractionalization. See \eqref{eq:reduction-odd} and \eqref{eq:reduction-even}. We also study the domain wall dimensional reduction of fermions. These results are consistent. See Theorem \ref{thm-DW}, \Fig{fig:domain-wall-fermion}, and \Table{table:domain-wall-fermion}. }
\end{itemize}

In \cite{LWWYY}, we study similar problems in a different Hamiltonian approach. We determine the group structure generated by C-R-T symmetry and the full internal symmetry (not only the fermion parity $\Z_2^{\rF}$). We also study the symplectic Majorana fermions \cite{Stone:2020vva2009.00518}. We study all the Dirac and Majorana mass terms and their interplay with the symmetries C, R, and T. 

In this article, we determine the maximal number of linearly independent Dirac and Majorana mass terms and construct them explicitly. However, we will only consider the conventional Dirac and Majorana mass term $\psi^{\dagger}\Gamma^0\psi$.

In this article, we mainly consider the Lorentz signature version of spacetime symmetry, such as $\SO^+(d-1,1)$ and $\Spin(d-1,1)$. However, the readers may wonder what is the relation between the Lorentz signature and the Euclidean signature of our theory.
By \cite{Freed1604.06527}, the unitarity of Lorentz is equivalent to the reflection positivity of Euclidean under Wick rotation. The (anti)-unitary C, R, and T symmetries in Lorentz become reflection-positive C and R symmetries in Euclidean under Wick rotation.

In \Sec{sec:Clifford}, we introduce the mathematical definitions of Dirac fermions, Weyl fermions, Majorana fermions, Majorana-Weyl fermions, charge conjugation, mirror reflection, and time reversal.
In \Sec{sec:Dirac},
we determine the C-R-T fractionalization with the canonical CRT and its extension class for Dirac fermions in any spacetime dimension, and we find an 8-periodicity.
In \Sec{sec:Majorana}, we determine the R-T fractionalization for Majorana fermions in any spacetime dimension $d=0,1,2,3,4\mod8$.
We also determine the C or T fractionalization for Weyl fermions in any even spacetime dimension.
In \Sec{sec:mass}, we determine the maximal number of Dirac and Majorana mass terms and construct them explicitly. We also discuss how the conventional Dirac and Majorana mass terms break the symmetries C, R, or T.
In \Sec{sec:domain-wall}, 
we study the domain wall dimensional reduction of the C-R-T fractionalization, R-T fractionalization, and fermions.
In \Sec{sec:conclusion}, we summarize the results in this article.

\section{Clifford algebras and representations}\label{sec:Clifford}

 A fermion is a particle that follows Fermi–Dirac statistics. A fermion in a Lorentz invariant field theory is a Grassmann-valued Lorentz covariant spinor. Namely, the fermion is both a spinor and a Grassmann number satisfying the anti-commutation condition. Upon the field quantization, a fermion becomes a field operator, which satisfies the anti-commutation condition of the fermion field operator. 
There are many varieties of fermions, with the four most common types being:
\begin{itemize}
    \item
Dirac fermions, 
    \item 
    Weyl fermions,

\item
Majorana fermions
, and
\item 
Majorana-Weyl fermions.
\end{itemize}

Mathematically, Dirac
fermions, Weyl fermions, Majorana fermions, and Majorana-Weyl fermions are defined in Definitions \ref{def-Dirac} (or \ref{def-Dirac'}), \ref{def-Weyl}, and \ref{def-Majorana}.
(these definitions are consistent with those given in physics textbooks \cite{Peskin1995, Weinberg1995, Zee2003, Srednicki2007, Polchinski:1998rr}). 

More precisely, Dirac fermions are defined as maps from the Minkowski spacetime $\R^{d-1,1}$ to the complex spin representation of the Spin group $\Spin_{d-1,1}$ or maps from the Minkowski spacetime $\R^{d-1,1}$ to an irreducible complex representation of the Clifford algebra $\Cl_{d-1,1}$. Mathematically, a representation of a group $G$ (or an algebra $A$) is a vector space $V$ together with a homomorphism $G\to\GL(V)$ ($A\to\text{End}(V)$) where $\GL(V)$ is the group of invertible linear transformations of $V$ ($\text{End}(V)$ is the algebra of endomorphisms of $V$), see \Sec{sec:representation}. The irreducible complex representation of the Clifford algebra $\Cl_{d-1,1}$ restricts as the complex spin representation of the Spin group $\Spin_{d-1,1}$: the Spin group $\Spin_{d-1,1}$ is a group which is a subset of the Clifford algebra $\Cl_{d-1,1}$, see Definition \ref{def-Pin}. For even $d$, the complex spin representation of the Spin group $\Spin_{d-1,1}$ is the direct sum of two complex semi-spin representations and maps from the Minkowski spacetime $\R^{d-1,1}$ to the two complex semi-spin representations are defined as left and right Weyl fermions respectively. Majorana fermions are defined as Dirac fermions with the trivial charge conjugation symmetry. Majorana-Weyl fermions are defined as fermions which are both Majorana and Weyl.
In the language of fiber bundle theory, Dirac fermions, Weyl fermions, Majorana fermions, and Majorana-Weyl fermions are sections of certain trivial fiber bundles over the Minkowski spacetime $\R^{d-1,1}$.

In order to give the mathematical definitions of Dirac fermions, Weyl fermions, Majorana fermions, and Majorana-Weyl fermions, we first introduce some mathematical notions, such as Clifford algebras, Pin and Spin groups, their representations, and bundle maps.
The main reference is \cite{Spin-geometry}.

\subsection{Clifford algebras}

Mathematically, a vector space over a field ($\R$ or $\C$) is a set with two operations: addition and scalar multiplication (the scalar is in the field) which satisfy certain conditions; a ring is a set with two operations: addition and multiplication which satisfy certain conditions; an algebra over a field ($\R$ or $\C$) is a set with three operations: addition, scalar multiplication (the scalar is in the field), and multiplication such that addition and scalar multiplication form a vector space, addition and multiplication form a ring, and they satisfy certain compatibility conditions.

Let $V$ be a real vector space of dimension $n$ with a \jw{nondegenerate quadratic form $Q$}
of signature $(p,q)$, $p+q=n$. \jw{The Clifford algebra $\Cl_{p,q}(V)$ is the algebra generated by $V$} with the relation
\bea
v^2=-Q(v)1,\;\;\;\forall v\in V.
\eea
This is equivalent to 
\bea
vw+wv=-2f_Q(v,w)1,\;\;\;\forall v,w\in V
\eea
where 
\bea
f_Q(v,w)=\frac{1}{2}(Q(v+w)-Q(v)-Q(w))
\eea
is the associated bilinear form of $Q$. $Q$ is defined to be nondegenerate if $f_Q$ is nondegenerate, namely $f_Q(v,w)=0$ for all $w\in V$ implies that $v=0$. $Q$ is defined to be of signature $(p,q)$ if $p$ is the dimension of the maximal vector subspace of $V$ on which $Q>0$ and $q$ is the dimension of the maximal vector subspace of $V$ on which $Q<0$.
The algebra generated by $V$ is defined as the direct sum of the tensor products $V^{\otimes n}$ of $n$ copies of $V$ for all $n\ge0$ as a vector space, and the multiplication is given by the tensor product.

If $e_1,\dots, e_n$ is an orthogonal basis for $V$ with $e_i^2=1$ for $1\le i\le p$, $e_i^2=-1$ for $p+1\le i\le n$, then $e_ie_j=-e_je_i$, $i\ne j$. The elements $e_I=e_{i_1}\cdots e_{i_k}$, $I=\{1\le i_1<i_2<\cdots<i_k\le n\}$ form a basis for $\Cl_{p,q}(V)$. So $\dim\Cl_{p,q}(V)=2^n$.
Note that as vector spaces, $\Cl_{p,q}(V)$ is isomorphic to the exterior algebra generated by $V$, but the multiplications are different, e.g. $e_ie_i=\pm1\ne0=e_i\wedge e_i$.

We define 
an algebra homomorphism $\alpha(e_I)=(-1)^ke_I=(-1)^{|I|}e_I$ and extend linearly to $\Cl_{p,q}(V)$. Then $\Cl_{p,q}(V)$ is a $\Z_2$-graded algebra
\bea
\Cl_{p,q}(V)=\Cl_{p,q}^0(V)\oplus \Cl_{p,q}^1(V)
\eea
where $\Cl_{p,q}^i(V)=\{x\in \Cl_{p,q}(V)| \alpha(x)=(-1)^ix\}$ and
\bea
\Cl_{p,q}^i(V)\cdot \Cl_{p,q}^j(V)\subseteq \Cl_{p,q}^{i+j}(V).
\eea
We write $\Cl_{p,q}\equiv\Cl_{p,q}(V)$ where $V=\R^{p+q}$.

\subsection{The groups Pin and Spin}

We now consider the multiplicative group of units in the Clifford algebra,
which is defined to be the subset
\bea
\Cl_{p,q}^*(V):=\{x\in \Cl_{p,q}(V)\mid\exists x^{-1}\text{ with }x^{-1}x=xx^{-1}=1\}.
\eea
The group of units always acts naturally as automorphisms of the algebra. That is, there is a homomorphism
\bea
\Ad: \Cl_{p,q}^*(V)\to \Aut(\Cl_{p,q}(V))
\eea
called the adjoint representation, which is given by 
\bea
\Ad_{x}(y):=x y x^{-1}.
\eea
\begin{proposition}(\cite[Proposition I.2.2]{Spin-geometry})\label{adjoint}
Let $v\in V\subset \Cl_{p,q}(V)$ be an element with $Q(v)\ne 0$. Then $\Ad_v(V)=V$. In fact, for all $w\in V$, the following equation holds:
\bea\label{Ad}
-\Ad_v(w)=w-2\frac{f_Q(v,w)}{Q(v)}v.
\eea
\end{proposition}
This leads us naturally to consider the subgroup of elements $x\in \Cl_{p,q}^*(V)$ such that $\Ad_{x}(V)=V$. By Proposition \ref{adjoint}, this group contains all elements $v\in V$ with $Q(v)\ne0$. Furthermore, we see from \eqref{Ad} that whenever $Q(v)\ne0$, the transformation $\Ad_v$ preserves the quadratic form $Q$. That is, $(\Ad_v^*Q)(w)=Q(\Ad_v(w))=Q(w)$ for all $w\in V$. Therefore, we define $P_{p,q}(V)$ to be the subgroup of $\Cl_{p,q}^*(V)$ generated by the elements $v\in V$ with $Q(v)\ne0$, and observe that there is a homomorphism 
\bea
P_{p,q}(V)\xrightarrow{\Ad}\rO_{p,q}(V)
\eea
where $\rO_{p,q}(V)$ is the group of invertible linear transformations of $V$ that preserve the quadratic form $Q$.

\begin{definition}(\cite[Definition I.2.3]{Spin-geometry})\label{def-Pin}
Let $\Pin_{p,q}(V)$ be the set of elements of $\Cl_{p,q}(V)$ which can be written in the form $v_1v_2\cdots v_k$ where $v_i\in V$ with $Q(v_i)=\pm1$. Those elements $v_1v_2\cdots v_k\in \Pin_{p,q}(V)$ for which $k$ is even form $\Spin_{p,q}(V)$. Note that 
$\Spin_{p,q}(V)=\Pin_{p,q}(V)\cap \Cl_{p,q}^0(V)$.
\end{definition}

We write $\Pin_{p,q}\equiv\Pin_{p,q}(V)$ and $\Spin_{p,q}\equiv\Spin_{p,q}(V)$ where $V=\R^{p+q}$.

\subsection{Representations}\label{sec:representation}

Mathematically, a representation of a group $G$ (or an algebra $A$) is defined as a vector space $V$ together with a homomorphism $G\to\GL(V)$ 
(or $A\to\text{End}(V)$) where $\GL(V)$ is the group of invertible linear transformations of $V$ (or $\text{End}(V)$ is the algebra of endomorphisms of $V$).\\

By \cite[Theorem I.3.7 and Table II]{Spin-geometry},
\bea\label{eq:Clifford}
\Spin_{d-1,1}\subset\Cl_{d-1,1}^0\cong\Cl_{d-2,1}\cong\left\{\begin{array}{cccccccc}\C(2^{\frac{d}{2}-1})&d=0\mod8\\
\R(2^{\frac{d-1}{2}})&d=1\mod8\\
\R(2^{\frac{d}{2}-1})\oplus\R(2^{\frac{d}{2}-1}) &d=2\mod8\\
\R(2^{\frac{d-1}{2}})&d=3\mod8\\
\C(2^{\frac{d}{2}-1})&d=4\mod8\\
\bH(2^{\frac{d-3}{2}})&d=5\mod8\\
\bH(2^{\frac{d}{2}-2})\oplus\bH(2^{\frac{d}{2}-2})&d=6\mod8\\
\bH(2^{\frac{d-3}{2}})&d=7\mod8
\end{array}
\right.
\eea
where $A(n)$ is the $n\times n$ matrix algebra over the algebra $A$ for $A=\R,\C,\bH$.

 By the proof of \cite[Theorem I.3.7]{Spin-geometry}, the isomorphism $f:\Cl_{d-2,1}\to\Cl_{d-1,1}^0$ is given by $f(e_{\mu})=e_{d-1}e_{\mu}$ for $\mu=0,1,\dots,d-2$ and extending it linearly where $e_{\mu}$ ($\mu=0,1,\dots,d-1$) form an orthogonal basis of $V=\R^d$ with $e_0^2=1$ and $e_j^2=-1$ ($j=1,2,\dots,d-1$).

 Any complex representation of $\Cl_{p,q}$ automatically extends to a representation of $\Cl_{p,q}\otimes_{\R}\C\cong\CCl_{p+q}$, the complexification of the real Clifford algebra $\Cl_{p,q}$.

By \eqref{eq:Clifford} or the classification of Clifford algebras \cite[Theorem I.5.8]{Spin-geometry}, 
\bea
\Cl_{d-2,1}\otimes_{\R}\C\cong\CCl_{d-1}\cong\left\{\begin{array}{ll} \C(2^{\frac{d-1}{2}}),&\text{for odd }d,\\
\C(2^{\frac{d}{2}-1})\oplus \C(2^{\frac{d}{2}-1}),&\text{for even }d.
\end{array} \right.
\eea

\begin{definition}\label{def-spin-rep}
\begin{enumerate}
\item
When $d$ is even, $\Cl_{d-2,1}$ has two irreducible complex representations. 
The complex semi-spin representation $S_{\pm}$ of $\Spin_{d-1,1}$ is the homomorphism
\bea
\Delta_{d-1,1,\pm}^{\C}:\Spin_{d-1,1}\to \GL_{\C}(S_{\pm})
\eea
given by restricting an irreducible 
complex representation $\Cl_{d-2,1}\to\Hom_{\C}(S_{\pm},S_{\pm})$ to $\Spin_{d-1,1}\subset\Cl_{d-1,1}^0\cong\Cl_{d-2,1}$. 
The complex spin representation $S$ of $\Spin_{d-1,1}$ is the direct sum of the two complex semi-spin representations $S_{+}$ and $S_{-}$. Namely,
\bea
\Delta_{d-1,1}^{\C}=\Delta_{d-1,1,+}^{\C}\oplus \Delta_{d-1,1,-}^{\C}.
\eea
\item 
When $d$ is odd, $\Cl_{d-2,1}$ has only one irreducible complex representation.
The complex spin representation $S$ of $\Spin_{d-1,1}$ is the homomorphism
\bea
\Delta_{d-1,1}^{\C}:\Spin_{d-1,1}\to \GL_{\C}(S)
\eea
given by restricting the irreducible 
complex representation $\Cl_{d-2,1}\to\Hom_{\C}(S,S)$ to $\Spin_{d-1,1}\subset\Cl_{d-1,1}^0\cong\Cl_{d-2,1}$. 
\end{enumerate}
\end{definition}

Mathematically, a fiber bundle \cite{Nakahara} is a quintuple $(E,\pi,M,F,G)$ satisfying the local triviality condition where $E$ is the total space, $M$ is the base space, $\pi:E\to M$ is the projection, $F$ is the fiber, and $G$ is the structure group. A section of a fiber bundle $(E,\pi,M,F,G)$ is a map $s:M\to E$ such that $\pi\circ s=\text{id}_{M}$.
Since the Minkowski spacetime $\R^{d-1,1}$ is contractible, any fiber bundle over the Minkowski spacetime is trivial \cite{Nakahara}, namely the product of the Minkowski spacetime and the fiber. A section of a trivial fiber bundle can be regarded as a map from the base space to the fiber.
We can construct the complex (semi-)spin bundle over the Minkowski spacetime $\R^{d-1,1}$ whose structure group is $\Spin_{d-1,1}$ and fiber is the complex (semi-)spin representation of $\Spin_{d-1,1}$.

\begin{definition}\label{def-Dirac}
    The Dirac fermions are sections of the complex spin bundle over the Minkowski spacetime $\R^{d-1,1}$. 
\end{definition}
We will give another equivalent definition for Dirac fermions later, see Definition \ref{def-Dirac'}.

\begin{definition}\label{def-Weyl}
    For even $d$, the complex spin representation of $\Spin_{d-1,1}$ is a direct sum of two complex semi-spin representations, the left and right Weyl fermions are sections of the two complex semi-spin bundles over the Minkowski spacetime $\R^{d-1,1}$ respectively. Namely, left and right Weyl fermions $\psi_L$ and $\psi_R$ are the components of the Dirac fermion $\psi=\left(\begin{array}{cc}\psi_L\\\psi_R\end{array}\right)$.
\end{definition}

The Dirac Lagrangian $\cL=\bar{\psi}(\ii\Gamma^{\mu}\partial_{\mu}-m)\psi$ where $\bar{\psi}\equiv\psi^{\dagger}\Gamma^0$, $\bar{\psi}\ii\Gamma^{\mu}\partial_{\mu}\psi$ is the kinetic term, and $\bar{\psi}m\psi$ is the mass term.  
Free Dirac fermions with mass $m$ satisfy the Dirac equation 
\bea\label{eq:Dirac}
(\ii\Gamma^{\mu}\partial_{\mu}-m)\psi=0.
\eea
It involves the Gamma matrices which represent the real Clifford algebra $\Cl_{d-1,1}$ with
\bea\label{eq:anticommutator}
\{\Gamma^{\mu},\Gamma^{\nu}\}=2\eta^{\mu\nu}.
\eea
Here $\{,\}$ means the anticommutator. We consider the Lorentz signature $\eta^{\mu\nu}=\diag(+,-,-,\dots,-)$.
Let $d$ denote the total spacetime dimension. 
We shall use $\mu=0,1,2,\dots, d-1$ to denote the spacetime index and $j=1,2,\dots,d-1$ to denote the space index.

Dirac fermions can also be regarded as maps from the Minkowski spacetime $\R^{d-1,1}$ to an irreducible complex representation of $\Cl_{d-1,1}$. We have the following alternative definition for Dirac fermions.

\begin{definition}\label{def-Dirac'}
    The Dirac fermions are sections of the Clifford module bundle\footnote{A Clifford module bundle is a vector bundle whose fibers are Clifford modules, the representations of Clifford algebras.} over the Minkowski spacetime $\R^{d-1,1}$ with fiber an irreducible complex representation of $\Cl_{d-1,1}$. For even $d$, $\Cl_{d-1,1}$ has a unique irreducible complex representation. For odd $d$, $\Cl_{d-1,1}$ has two non-isomorphic irreducible complex representations.
    We choose one of them such that the image of the volume element in the representation is $-1$, see Lemma \ref{volume-element}.
\end{definition}
The following lemma is well-known, but we still include its proof here.
\begin{lemma}\label{volume-element}
For odd $d=2k+3$, the two non-isomorphic irreducible complex representations of $\Cl_{d-1,1}$ can be distinguished by the image of the volume element $\omega=\ii^{k+1}e_0e_1\cdots e_{2k+2}$ in the representation where $e_{\mu}$ ($\mu=0,1,\dots,2k+2$) are generators of $\Cl_{d-1,1}$ with $e_0^2=1$ and $e_j^2=-1$ for $j=1,\dots,2k+2$.
The image of the volume element in the representation is either 1 or $-1$. If $(S,\Delta:\Cl_{d-1,1}\to\text{End}_{\C}(S))$ is one of the two irreducible complex representations, then the other one is $(S,\Delta':\Cl_{d-1,1}\to\text{End}_{\C}(S))$ where $\Delta'(e_{\mu})=\Delta(-e_{\mu})$.
\end{lemma}
\begin{proof}
    By the relations in the Clifford algebra, $\omega^2=1$ and $\omega$ lies in the center of $\Cl_{d-1,1}$. Let $\Gamma=\Delta(\omega)$ be the image of $\omega$ in the irreducible complex representation $(S,\Delta:\Cl_{d-1,1}\to\text{End}_{\C}(S))$, then $\Gamma^2=1$ and $S$ decomposes as the direct sum $S=S_+\oplus S_-$ where $S_{\pm}=\ker(\Gamma\mp1)$. Since $\omega$ lies in the center of $\Cl_{d-1,1}$, $S_{\pm}$ are $\Cl_{d-1,1}$-invariant subspaces of $S$. By the irreducibility of $S$, $S=S_+$ or $S=S_-$. Therefore, $\Gamma=\pm1$. If $\Delta(\omega)=1$, then $\Delta'(\omega)=-1$. Therefore, $(S,\Delta)$ and $(S,\Delta')$ are not isomorphic. By \eqref{eq:complex-Clifford}, these are the only two irreducible representations (up to isomorphism).
\end{proof}

\cred{For both $d=2k+2$ and $d=2k+3$, we can choose the Gamma matrices as complex $2^{k+1}\times2^{k+1}$ matrices (see \App{app:Weyl}). }
Since 
\bea\label{eq:complex-Clifford}
\Cl_{d-1,1}\otimes_{\R}\C\cong\CCl_{d}\cong\left\{\begin{array}{ll} \C(2^{\frac{d}{2}}),&\text{for even }d,\\
\C(2^{\frac{d-1}{2}})\oplus \C(2^{\frac{d-1}{2}}),&\text{for odd }d,
\end{array} \right.
\eea
this amounts to choosing 
the Gamma matrices in only one direct summand of $\C(2^{\frac{d-1}{2}})\oplus \C(2^{\frac{d-1}{2}})$ for odd $d$.
Namely, we choose 
an irreducible complex representation $S'$ of $\Cl_{d-1,1}$:
\bea\label{eq:Clifford-rep}
{\Delta'}^{\C}_{d-1,1}:\Cl_{d-1,1}\to\text{End}_{\C}(S').
\eea
The restriction of $S'$ to $\Spin_{d-1,1}$ is
\bea
S'|_{\Spin_{d-1,1}}=\left\{\begin{array}{cc}
   S=S_+\oplus S_-,  & d\text{ even}, \\
   S,  & d\text{ odd}
\end{array}\right.
\eea
where $S_{\pm}$ are the complex semi-spin representations of $\Spin_{d-1,1}$ and $S$ is the complex spin representation of $\Spin_{d-1,1}$. Namely, we have
\bea\label{eq:restriction}
\xymatrix{ \Spin_{d-1,1}  \ar@{^{(}->}[d]\ar[r]^{\Delta_{d-1,1}^{\C}}&\GL_{\C}(S)\ar@{^{(}->}[d]\\
\Cl_{d-1,1}\ar[r]^{{\Delta'}^{\C}_{d-1,1}}&\text{End}_{\C}(S').}
\eea
Therefore, Definitions \ref{def-Dirac} and \ref{def-Dirac'} are indeed equivalent.

\subsection{C, R, and T as bundle maps}
The symmetries of the Dirac equation (the transformations that preserve the equation) include charge conjugation C, mirror reflection $\rR_j$, time-reversal T, and fermion parity $(-1)^{\rF}$.

If we regard Dirac fermions as sections of a Clifford module bundle $M\times S'$ over the Minkowski spacetime $M=\R^{d-1,1}$ where $S'$ is the irreducible complex representation of $\Cl_{d-1,1}$ (see Definition \ref{def-Dirac'}), then
the charge conjugation sends the Clifford module bundle to its complex conjugate, and the mirror reflection and time-reversal send the Clifford module bundle to itself. Mathematically, they are bundle maps between Clifford module bundles.\footnote{Let $E\xrightarrow{\pi} M$ and $E'\xrightarrow{\pi'} M'$ be fiber bundles. A map $\overline{f}:E\to E'$
is called a bundle map if it maps each fiber $\pi^{-1}(p)$ of $E$ onto $\pi'^{-1}(q)$ of $E'$. Then $\overline{f}$ naturally induces a map $f : M\to M'$ such that $f (p) = q$.}
More precisely, we have
\bea
\xymatrix{E\ar[r]^{\rC}\ar[d]^{\pi}&E'\ar[d]^{\pi'}\\M\ar[r]^{\text{id}_M}&M,}
\eea
\bea
\xymatrixcolsep{12pc}\xymatrix{E\ar[r]^{\rR_j}\ar[d]^{\pi}&E\ar[d]^{\pi}\\M\ar[r]^{(t,x)\mapsto(t,x_1,\dots,x_{j-1},-x_j,x_{j+1},\dots,x_{d-1})}&M,}
\eea
and 
\bea
\xymatrixcolsep{5pc}\xymatrix{E\ar[r]^{\rT}\ar[d]^{\pi}&E\ar[d]^{\pi}\\M\ar[r]^{(t,x)\mapsto(-t,x)}&M}
\eea
where $E=M\times S'$ is the Clifford module bundle over the Minkowski spacetime $M=\R^{d-1,1}$, $S'$ is the irreducible complex representation of $\Cl_{d-1,1}$, $E'=M\times \overline{S'}$ is the complex conjugate of the Clifford module bundle $E$, and $\overline{S'}$ is the complex conjugate of the irreducible complex representation $S'$,\footnote{Mathematically, if $A$ is an algebra and $\rho:A\to\text{End}(V)$ is a representation of it over the complex vector space $V$, then the complex conjugate representation $\overline{\rho}$ is defined over the complex conjugate vector space $\overline{V}$ as follows:
$\overline{\rho}(a)$ is the complex conjugate of $\rho(a)$ for all $a$ in $A$.\label{complex-conjugate}
}. 
We should consider $S'$ and $\overline{S'}$ simultaneously and regard C, R$_j$, and T as bundle automorphisms of $E\oplus E'$ (the Whitney sum of $E$ and $E'$).
Namely,
\bea
\xymatrix{E\oplus E'\ar[r]^{\rC}\ar[d]^{\pi''}&E\oplus E'\ar[d]^{\pi''}\\M\ar[r]^{\text{id}_M}&M,}
\eea
\bea
\xymatrixcolsep{12pc}\xymatrix{E\oplus E'\ar[r]^{\rR_j}\ar[d]^{\pi''}&E\oplus E'\ar[d]^{\pi''}\\M\ar[r]^{(t,x)\mapsto(t,x_1,\dots,x_{j-1},-x_j,x_{j+1},\dots,x_{d-1})}&M,}
\eea
and 
\bea
\xymatrixcolsep{5pc}\xymatrix{E\oplus E'\ar[r]^{\rT}\ar[d]^{\pi''}&E\oplus E'\ar[d]^{\pi''}\\M\ar[r]^{(t,x)\mapsto(-t,x)}&M.}
\eea

For a Dirac fermion $\psi$, the symmetry transformations are (see \eqref{eq:composition-matrix1} and \eqref{eq:composition-matrix2})
\bea
\rC&:&\left(\begin{array}{cc}\psi\\\psi^*\end{array}\right)\mapsto \left(\begin{array}{cc}\psi\\\psi^*\end{array}\right)^{\rC}=\left(\begin{array}{cc}0&M_{\rC}\\M_{\rC}^*&0\end{array}\right)\left(\begin{array}{cc}\psi\\\psi^*\end{array}\right)\nn\\
\rR_j&:&\left(\begin{array}{cc}\psi\\\psi^*\end{array}\right)\mapsto \left(\begin{array}{cc}\psi\\\psi^*\end{array}\right)^{\rR_j}=\left(\begin{array}{cc}M_{\rR_j}&0\\0&M_{\rR_j}^*\end{array}\right)\left(\begin{array}{cc}\psi\\\psi^*\end{array}\right)(t,x_1,\dots,x_{j-1},-x_j,x_{j+1},\dots,x_{d-1})\nn\\
\rT&:&\left(\begin{array}{cc}\psi\\\psi^*\end{array}\right)\mapsto \left(\begin{array}{cc}\psi\\\psi^*\end{array}\right)^{\rT}=\left(\begin{array}{cc}M_{\rT}&0\\0&M_{\rT}^*\end{array}\right)\left(\begin{array}{cc}\psi\\\psi^*\end{array}\right)(-t,x)\nn\\
(-1)^{\rF}&:&\left(\begin{array}{cc}\psi\\\psi^*\end{array}\right)\mapsto-\left(\begin{array}{cc}\psi\\\psi^*\end{array}\right).
\eea
Here $\psi^*$ is the complex conjugate of $\psi$ but it does take the adjoint of the components of the field operator. (The complex conjugate of a field operator is not defined in any representation-independent way. The right concept here is adjoint, not complex conjugate.) Here $M^*$ is the complex conjugate of the matrix $M$, see footnote~\ref{complex-conjugate},
$M_{\rC}$ is an invertible matrix in $\text{End}_{\C}(\overline{S'})$, $M_{\rR_j}$ and $M_{\rT}$ are invertible matrices in $\text{End}_{\C}(S')$, hence $M_{\rC}^*$, $M_{\rR_j}$, and $M_{\rT}$ are in $\text{End}_{\C}(S')$.

The matrices $M_{\rC}$, $M_{\rR_j}$, and $M_{\rT}$ satisfy (see \eqref{eq:charge}, \eqref{eq:reflection}, and \eqref{eq:time})
\bea\label{eq:CPRT-condition}
   M_{\rC}\Gamma^{\mu*}M_{\rC}^{-1}&=&-\Gamma^{\mu},\;\;\;\mu=0,1,\dots,d-1,\nn\\
M_{\rR_j}\Gamma^0M_{\rR_j}^{-1}&=&\Gamma^0,\nn\\
M_{\rR_j}\Gamma^jM_{\rR_j}^{-1}&=&-\Gamma^j,\nn\\
M_{\rR_j}\Gamma^iM_{\rR_j}^{-1}&=&\Gamma^i, \;\;\; i=1,\dots,j-1,j+1,\dots,d-1,\nn\\
M_{\rT}\Gamma^{0*}M_{\rT}^{-1}&=&\Gamma^0,\nn\\
M_{\rT}\Gamma^{j*}M_{\rT}^{-1}&=&-\Gamma^j,\;\;\;j=1,\dots,d-1.
\eea

\begin{definition}\label{def-Majorana}
    Majorana fermions are the Dirac fermions $\psi$ with the trivial C symmetry, namely $\psi^{\rC}=\psi$.
    Majorana-Weyl fermions are the fermions which are both Majorana and Weyl.
\end{definition}
 We will see later that Majorana fermions only exist for $d=0,1,2,3,4\mod8$ and Majorana-Weyl fermions only exist for $d=2\mod8$. 
We may give another \emph{inequivalent} definition of Majorana fermions and Majorana-Weyl fermions: 

\begin{definition}\label{def-another-Majorana}
Majorana(-Weyl) fermions are sections of a real (semi-)spin bundle over the Minkowski spacetime $\R^{d-1,1}$, namely maps from $\R^{d-1,1}$ to a real (semi-)spin representation of $\Spin_{d-1,1}$.  
\end{definition}

 By \eqref{eq:Clifford}, the real spin representation of $\Spin_{d-1,1}$ exists for all $d$. However, for $d=5,6,7\mod8$, the real spin representation of $\Spin_{d-1,1}$ is of quaternionic type. For $d=5,6,7\mod8$, if we define Majorana fermions as in Definition \ref{def-another-Majorana}, then a Majorana fermion is a \emph{pair} of two Dirac fermions with the trivial charge conjugation symmetry \cite{Stone:2020vva2009.00518}, not a single Dirac fermion with the trivial charge conjugation symmetry.
 If we define Majorana fermions as in Definition \ref{def-another-Majorana},
 for $d=5,6,7 \mod 8$, since $\bH_{\C}=\C(2)$ and $\bH(n)_{\C}=\C(2n)$, the real representation $\bH^n$ becomes $\C^{2n}$ after complexification. So the real dimension does not change, we must impose the trivial charge conjugation condition on two Dirac fermions to match the dimension.
While for $d=0,1,2,3,4 \mod 8$, the real dimension becomes twice after complexification. So, it is sufficient to impose the trivial charge conjugation condition on a single Dirac fermion to match the dimension. 
For $d=5,6,7\mod8$, 
the corresponding Majorana fermions in Definition \ref{def-another-Majorana} are called \emph{symplectic Majorana} fermions \cite{Stone:2020vva2009.00518} that we will not study in this article.
We will see later that the matrix of charge conjugation can always be chosen to be the identity matrix, so the Majorana condition $\psi^{\rC}=\psi$ implies the real condition on Dirac fermions. Therefore, the Majorana(-Weyl) fermions in Definition \ref{def-Majorana} are always sections of a real (semi-)spin bundle over the Minkowski spacetime $\R^{d-1,1}$, namely Definition \ref{def-Majorana} is more restrictive than Definition \ref{def-another-Majorana}.

\section{C-R-T fractionalization for Dirac fermions}\label{sec:Dirac}
In this section, we determine the C-R-T fractionalization with the canonical CRT and its extension class for Dirac
fermions in any spacetime dimension, and we find an 8-periodicity.

Our main result is the following theorem.
 \begin{theorem}\label{main}
\begin{enumerate}
    \item 
    For a Dirac fermion and each $d$, there are two possibilities of the C-R$_j$-T fractionalization with the canonical CRT.
    \item
    For a Dirac fermion, the two possibilities of the C-R$_j$-T fractionalization with the canonical CRT depend only on the remainder of the spacetime dimension $d$ modulo 8.
\end{enumerate}

\end{theorem}

By \cite{HasonKomargodskiThorngren1910.14039}, the canonical CRT satisfies the following properties:
\begin{enumerate}
    \item 
    Any unitary internal symmetry U commutes with the canonical CRT:
    \bea\label{eq:U-CPT}
\U\cdot \rC\rR_j\rT=\rC\rR_j\rT\cdot \U.
    \eea
    \item 
    Any time reversal symmetry T commutes with the canonical CRT up to the fermion parity:
    \bea\label{eq:T-CPT}
\rT\cdot \rC\rR_j\rT=(-1)^{\rF}\rC\rR_j\rT\cdot \rT.
    \eea
    \item 
    The canonical CRT is of order 2:
    \bea\label{eq:CPT}
    (\rC\rR_j\rT)^2=1.
    \eea
\end{enumerate}
We check that the CRT for our two possibilities of the C-R$_j$-T fractionalization satisfies \eqref{eq:CPT}, \eqref{eq:T-CPT}, and \eqref{eq:U-CPT} for $\U=\rC$ (see the discussion around \eqref{eq:CRT2}) and the internal charge-like and isospin-like U(1) symmetries (see \App{app:further}), but we do not check \eqref{eq:U-CPT} for other U.

To prove Theorem \ref{main}, we prove the following lemmas.
\begin{lemma}\label{order}
    The order of elements of $\tilde{G}_{\psi}$ in the group extension 
\eqref{eq:extension}
divides 4.
\end{lemma}

\begin{proof}
    The elements of $\tilde{G}_{\psi}$ in the group extension \eqref{eq:extension} are expressed as pairs $(a,b)$ where $a\in \Z_2^{\rC}\times\Z_2^{\rR_j}\times\Z_2^{\rT}$ and $b\in \Z_2^{\rF}$. Then $(a,b)^2=(a+a,b+b+f(a,a))=(0,f(a,a))$ and $(a,b)^4=(0,0)$ where $f$ is a 2-cocycle on $\Z_2^{\rC}\times\Z_2^{\rR_j}\times\Z_2^{\rT}$ valued in $\Z_2^{\rF}$.  
\end{proof}

\begin{lemma}\label{phase-factor}
The matrices of C, R$_j$, and T for a Dirac fermion are uniquely determined up to a phase factor in U(1). 
\end{lemma}
\begin{proof}
   Let $M_{\rC}$, $M_{\rR_j}$, and $M_{\rT}$ be the matrices of C, R$_j$, and T, then $M_{\rC}^*$, $M_{\rR_j}$, and $M_{\rT}$ are invertible matrices in $\text{End}_{\C}(S')$.
   The matrices of C, R$_j$, and T act on Gamma matrices by conjugation, see
   \eqref{eq:CPRT-condition}. 
Note that $\Gamma^{\mu}$ is either real or imaginary, so \eqref{eq:CPRT-condition} implies that $M_{\rC}^*$, $M_{\rR_j}$, and $M_{\rT}$ commute with some Gamma matrices while anti-commute with other Gamma matrices.

If $M_{\rC}$ and $M_{\rC}'$ are two solutions to \eqref{eq:CPRT-condition}, 
then $M_{\rC}'M_{\rC}^{-1}$ commutes with all Gamma matrices, so does $M_{\rC}'^*{M_{\rC}^*}^{-1}$.
Since $S'$ is an irreducible complex representation of $\Cl_{d-1,1}$, by Schur's lemma, $M_{\rC}'M_{\rC}^{-1}=\eta_{\rC}\cdot I$ where $\eta_{\rC}\in \C$ and $I$ is the identity matrix. Since C, R$_j$, and T preserve the norm of $\psi$, we have $\eta_{\rC}\in \U(1)$. Similarly for R$_j$ and T.

\end{proof}

The proof of Theorem \ref{main} consists of the following steps:
\begin{itemize}
\item 
To determine the C-R$_j$-T fractionalization, we only need to determine the squares of C, R$_j$, T, CR$_j$, R$_j$T, and CT, hence we obtain a group presentation of the C-R$_j$-T fractionalization.
    \item 
    We observe that 
    the order of the elements in the group $\tilde{G}_{\psi}$ divides 4, see Lemma \ref{order}.
\item 
By Schur's lemma, we can prove that the matrices of C, R$_j$, and T for a Dirac fermion are unique up to a phase factor in U(1), see Lemma \ref{phase-factor}. 
Moreover, we can determine the matrices of C, R$_j$, and T for a Dirac fermion up to a phase factor, see the discussion below.

\item 
Based on the composition rules \eqref{eq:composition-matrix1} and \eqref{eq:composition-matrix2}, we can prove that among the squares of C, R$_j$, T, CR$_j$, R$_j$T, and CT, only the squares of R$_j$ and CT depend on the phase factors. 
More precisely,
\bea
\rC^2\psi\rC^{-2}&=&M_{\rC}M_{\rC}^*\psi,\nn\\
\rR_j^2\psi\rR_j^{-2}&=&M_{\rR_j}^2\psi,\nn\\
\rT^2\psi\rT^{-2}&=&M_{\rT}^*M_{\rT}\psi,\nn\\
(\rC\rR_j)^2\psi(\rC\rR_j)^{-2}&=&M_{\rR_j}M_{\rC}M_{\rR_j}^*M_{\rC}^*\psi,\nn\\
(\rR_j\rT)^2\psi(\rR_j\rT)^{-2}&=&(M_{\rT}M_{\rR_j})^*M_{\rT}M_{\rR_j}\psi,\nn\\
(\rC\rT)^2\psi(\rC\rT)^{-2}&=&(M_{\rT}M_{\rC})^*M_{\rT}^*M_{\rC}^*\psi
\eea
where $M_{\rC}$, $M_{\rR_j}$, and $M_{\rT}$ are the matrices of C, R$_j$, and T respectively.
In particular, the phase factors of R$_j$ and CT can only be $\pm1$ or $\pm\ii$ since their order divides 4.
\item 
Since the order of the canonical CRT is 2, see \eqref{eq:CPT}, we can determine the product of the phase factors of C, R$_j$, and T if we consider the canonical CRT.

\item 
In summary, there are two possibilities of the C-R$_j$-T fractionalization with the canonical CRT corresponding to the two choices of the phase factor of R$_j$: $\pm1$ and $\pm\ii$.
\end{itemize}

Now we determine the matrices of C, R$_j$, and T for a Dirac fermion up to a phase factor.
\begin{enumerate}
\item{Charge conjugation}

The Dirac equation interacting with an external electromagnetic field is
\bea\label{eq:Dirac1}
(\ii\Gamma^{\mu}\partial_{\mu}-e\Gamma^{\mu}A_{\mu}-m)\psi=0
\eea
and hence the charge-conjugated field equation is
\bea\label{eq:Dirac2}
(\ii\Gamma^{\mu}\partial_{\mu}+e\Gamma^{\mu}A_{\mu}-m)\psi^{\rC}=0.
\eea
We find that if $\psi$ satisfies \eqref{eq:Dirac1}, then $\psi^{\rC}=M_{\rC}\psi^*$ satisfies \eqref{eq:Dirac2} provided that
\bea\label{eq:charge}
M_{\rC}\Gamma^{\mu*}M_{\rC}^{-1}=-\Gamma^{\mu}.
\eea

First, we consider $d=2k+2$.
Note that only $\Gamma^2,\Gamma^4,\dots,\Gamma^{2k}$ are imaginary, this implies that $M_{\rC}$ commutes with $\Gamma^2,\Gamma^4,\dots,\Gamma^{2k}$, anticommutes with $\Gamma^0,\Gamma^1,\Gamma^3,\dots,\Gamma^{2k+1}$. Therefore, 
\bea
M_{\rC}=\left\{\begin{array}{ll}\Gamma^2\Gamma^4\cdots\Gamma^{2k}&k\text{ odd,}\\
\Gamma^0\Gamma^1\Gamma^3\cdots\Gamma^{2k+1}&k\text{ even}\end{array}\right.
\eea
is a solution to \eqref{eq:charge}. 

Next, we consider $d=2k+3$.
Recall that $\Gamma^{2k+2}=\ii^{k+1}\Gamma^0\Gamma^1\cdots\Gamma^{2k+1}$. So if $M_{\rC}\Gamma^{\mu*}M_{\rC}^{-1}=-\Gamma^{\mu}$ for $\mu=0,1,\dots,2k+1$, then $M_{\rC}\Gamma^{2k+2*}M_{\rC}^{-1}=(-1)^{k+1}\Gamma^{2k+2}.$
Let $\Gamma=\Gamma^0\Gamma^1\cdots\Gamma^{2k+1}$. 

For even $k$, we can choose $M_{\rC}$ as we did for $d=2k+2$ such that \eqref{eq:charge} is satisfied. So $\psi^{\rC}=M_{\rC}\psi^*=\Gamma^0\Gamma^1\Gamma^3\cdots\Gamma^{2k+1}\psi^*$.

For odd $k$, if we choose $M_{\rC}$ as we did for $d=2k+2$, then \eqref{eq:charge} can not be satisfied, while $(\Gamma M_{\rC})\Gamma^{\mu*}(\Gamma M_{\rC})^{-1}=\Gamma^{\mu}$ for $\mu=0,1,\dots,2k+2$. Therefore, the mass term breaks the C symmetry in this case. For the massless case, $\psi^{\rC}=\Gamma M_{\rC}\psi^*=\Gamma^0\Gamma^1\Gamma^3\cdots\Gamma^{2k+1}\psi^*$.

\item{Mirror reflection}

The mirror reflected Dirac equation is
\bea
&&(\ii\Gamma^0\partial_0+\ii\Gamma^1\partial_1+\cdots+\ii\Gamma^{j-1}\partial_{j-1}-\ii\Gamma^j\partial_j+\ii\Gamma^{j+1}\partial_{j+1}+\cdots+\ii\Gamma^{d-1}\partial_{d-1}-m)\cdot\nn\\
&&\psi^{\rR_j}(t,x_1,\dots,x_{j-1},-x_j,x_{j+1},\dots,x_{d-1})=0.
\eea

Let $\psi^{\rR_j}(t,x)=M_{\rR_j}\psi(t,x_1,\dots,x_{j-1},-x_j,x_{j+1},\dots,x_{d-1})$, then by comparing with \eqref{eq:Dirac}, we get
\bea
&&M_{\rR_j}(\ii\Gamma^0\partial_0+\ii\Gamma^1\partial_1+\cdots+\ii\Gamma^{d-1}\partial_{d-1}-m)M_{\rR_j}^{-1}\nn\\
&=&\ii\Gamma^0\partial_0+\ii\Gamma^1\partial_1+\cdots+\ii\Gamma^{j-1}\partial_{j-1}-\ii\Gamma^j\partial_j+\ii\Gamma^{j+1}\partial_{j+1}+\cdots+\ii\Gamma^{d-1}\partial_{d-1}-m.
\eea
Since $\rR_j$ is unitary, we get
\bea\label{eq:reflection}
M_{\rR_j}\Gamma^0M_{\rR_j}^{-1}&=&\Gamma^0,\nn\\
M_{\rR_j}\Gamma^jM_{\rR_j}^{-1}&=&-\Gamma^j,\nn\\
M_{\rR_j}\Gamma^iM_{\rR_j}^{-1}&=&\Gamma^i, \;\;\; i=1,\dots,j-1,j+1,\dots,d-1.
\eea

For $d=2k+2$,
we find that $M_{\rR_j}=\Gamma^0\Gamma^1\cdots\hat{\Gamma}^j\cdots\Gamma^{2k+1}$ is a solution to \eqref{eq:reflection} and $\psi^{\rR_j}(t,x)=\Gamma^0\Gamma^1\cdots\hat{\Gamma}^j\cdots\Gamma^{2k+1}\psi(t,x_1,\dots,x_{j-1},-x_j,x_{j+1},\dots,x_{d-1})$ where $\hat{\Gamma}^j$ means that $\Gamma^j$ is omitted.

For $d=2k+3$,
recall that $\Gamma^{2k+2}=\ii^{k+1}\Gamma^0\Gamma^1\cdots\Gamma^{2k+1}$. If $j<2k+2$ and $M_{\rR_j}\Gamma^0M_{\rR_j}^{-1}=\Gamma^0$, 
$M_{\rR_j}\Gamma^jM_{\rR_j}^{-1}=-\Gamma^j$, and
$M_{\rR_j}\Gamma^iM_{\rR_j}^{-1}=\Gamma^i$ for $i=1,\dots,j-1,j+1,\dots,2k+1$, then $M_{\rR_j}\Gamma^{2k+2}M_{\rR_j}^{-1}=-\Gamma^{2k+2}$. So \eqref{eq:reflection} can not be satisfied. Therefore, the mass term breaks the $\rR_j$ symmetry in odd dimensional spacetime. For the massless case, $M_{\rR_j}=\Gamma^j$ and $\psi^{\rR_j}(t,x)=\Gamma^j\psi(t,x_1,\dots,x_{j-1},-x_j,x_{j+1},\dots,x_{d-1})$.
If $j=2k+2$ and $M_{\rR_{2k+2}}\Gamma^0M_{\rR_{2k+2}}^{-1}=\Gamma^0$, 
$M_{\rR_{2k+2}}\Gamma^iM_{\rR_{2k+2}}^{-1}=\Gamma^i$ for $i=1,\dots,2k+1$, then
$M_{\rR_{2k+2}}\Gamma^{2k+2}M_{\rR_{2k+2}}^{-1}=\Gamma^{2k+2}$. So \eqref{eq:reflection} can not be satisfied. Therefore, the mass term breaks the $\rR_{2k+2}$ symmetry in odd dimensional spacetime. For the massless case, $M_{\rR_{2k+2}}=\Gamma^{2k+2}$ and $\psi^{\rR_{2k+2}}(t,x)=\Gamma^{2k+2}\psi(t,x_1,\dots,x_{2k+1},-x_{2k+2})$.

\item{Time-reversal}

Since T is anti-unitary, $\rT\ii \rT^{-1}=-\ii$. 
The time-reversed Dirac equation is
\bea
(-\ii\Gamma^0\partial_0+\ii\Gamma^j\partial_j-m)^*\psi^{\rT}(-t,x)=0.
\eea

Let $\psi^{\rT}(t,x)=M_{\rT}\psi(-t,x)$, then by comparing with \eqref{eq:Dirac}, we get 
\bea
M_{\rT}(\ii\Gamma^0\partial_0+\ii\Gamma^j\partial_j-m)^*M_{\rT}^{-1}=-\ii\Gamma^0\partial_0+\ii\Gamma^j\partial_j-m.
\eea
 Then we get
\bea\label{eq:time}
M_{\rT}\Gamma^{0*}M_{\rT}^{-1}&=&\Gamma^0,\nn\\
M_{\rT}\Gamma^{j*}M_{\rT}^{-1}&=&-\Gamma^j.
\eea

First, we consider $d=2k+2$.
Note that only $\Gamma^2,\Gamma^4,\dots,\Gamma^{2k}$ are imaginary, this implies that $M_{\rT}$ commutes with $\Gamma^0,\Gamma^2,\dots,\Gamma^{2k}$, anticommutes with $\Gamma^1,\Gamma^3,\dots,\Gamma^{2k+1}$. Therefore, 
\bea
M_{\rT}=\left\{\begin{array}{ll} \Gamma^1\Gamma^3\cdots\Gamma^{2k+1}&k\text{ odd,}\\
\Gamma^0\Gamma^2\cdots\Gamma^{2k}&k\text{ even}\end{array}\right.
\eea
is a solution to \eqref{eq:time}. 

Next, we consider $d=2k+3$.
Recall that $\Gamma^{2k+2}=\ii^{k+1}\Gamma^0\Gamma^1\cdots\Gamma^{2k+1}$. So if $M_{\rT}\Gamma^{0*}M_{\rT}^{-1}=\Gamma^0$ and 
$M_{\rT}\Gamma^{j*}M_{\rT}^{-1}=-\Gamma^j$ for $j=1,2,\dots,2k+1$, then $M_{\rT}\Gamma^{2k+2*}M_{\rT}^{-1}=(-1)^k\Gamma^{2k+2}$.
Recall that $\Gamma=\Gamma^0\Gamma^1\cdots\Gamma^{2k+1}$. 

For odd $k$, we can choose $M_{\rT}$ as we did for $d=2k+2$ such that \eqref{eq:time} is satisfied. So $\psi^{\rT}(t,x)=M_{\rT}\psi(-t,x)=\Gamma^1\Gamma^3\cdots\Gamma^{2k+1}\psi(-t,x)$.

For even $k$, if we choose $M_{\rT}$ as we did for $d=2k+2$, then \eqref{eq:time} can not be satisfied, while $(\Gamma M_{\rT})\Gamma^{0*}(\Gamma M_{\rT})^{-1}=-\Gamma^0$ and $(\Gamma M_{\rT})\Gamma^{j*}(\Gamma M_{\rT})^{-1}=\Gamma^j$ for $j=1,2,\dots,2k+2$. Therefore, the mass term breaks the T symmetry in this case. For the massless case,
$\psi^{\rT}(t,x)=\Gamma M_{\rT}\psi(-t,x)=\Gamma^1\Gamma^3\cdots\Gamma^{2k+1}\psi(-t,x)$.

\end{enumerate}

 We summarize the $\rC$, $\rR_j$, $\rT$ transformations on a Dirac fermion in $d$d for all $d$ in \Table{tab:CPRT-trans-dd}.\footnote{Similar results were obtained in \cite{shimizu1985} but P instead of R was considered and only one choice of the phase factors was specified there.}
\begin{table}[htp]
\begin{center}
\hspace*{-1cm}
\begin{tabular}{c|c|c | c}
\hline
Dimension & $M_{\rC}\psi^*$ & $M_{\rR_j}\psi(t,x_1,\dots,x_{j-1},-x_j,x_{j+1},\dots,x_{d-1})$ & $M_{\rT}\psi(-t,x)$ \\
\hline
$d=2k+2$ ($k$ odd) &  $\eta_{\rC}\Gamma^2\Gamma^4\cdots\Gamma^{2k}\psi^*$ & $\eta_{\rR_j}\Gamma^0\Gamma^1\cdots\hat{\Gamma}^j\cdots\Gamma^{2k+1}\psi(t,x_1,\dots,x_{j-1},-x_j,x_{j+1},\dots,x_{d-1})$ & $\eta_{\rT}\Gamma^1\Gamma^3\cdots\Gamma^{2k+1}\psi(-t,x)$\\
$d=2k+2$ ($k$ even)  & $\eta_{\rC}\Gamma^0\Gamma^1\Gamma^3\cdots\Gamma^{2k+1}\psi^*$ & $\eta_{\rR_j}\Gamma^0\Gamma^1\cdots\hat{\Gamma}^j\cdots\Gamma^{2k+1}\psi(t,x_1,\dots,x_{j-1},-x_j,x_{j+1},\dots,x_{d-1})$ & $\eta_{\rT}\Gamma^0\Gamma^2\cdots\Gamma^{2k}\psi(-t,x)$ \\
$d=2k+3$  & $\eta_{\rC}\Gamma^0\Gamma^1\Gamma^3\cdots\Gamma^{2k+1}\psi^*$ & $\eta_{\rR_j}\Gamma^j\psi(t,x_1,\dots,x_{j-1},-x_j,x_{j+1},\dots,x_{d-1}) $ & $\eta_{\rT}\Gamma^1\Gamma^3\cdots\Gamma^{2k+1}\psi(-t,x)$ \\
\hline
\end{tabular}
\caption{The $\rC$, $\rR_j$, $\rT$ transformations on a Dirac fermion in $d$d for all $d$. Here the phase factors $\eta_{\rC}$, $\eta_{\rR_j}$, and $\eta_{\rT}$ are to be determined. Here $\hat{\Gamma}^j$ means $\Gamma^j$ is omitted in the product of gamma matrices. See \App{app:Weyl} for the explicit form of the Gamma matrices.  }
\label{tab:CPRT-trans-dd}
\end{center}
\end{table} 
We find that the \cred{Dirac} mass term breaks the C symmetry for $d=2k+3$ ($k$ odd), it breaks the T symmetry for $d=2k+3$ ($k$ even), and it breaks the R$_j$ symmetry for $d=2k+3$ and $j=1,2,\dots,2k+2$, see \Table{table:CRT-break-Dirac}.

Now we determine the group structure of the C-R$_j$-T fractionalization for Dirac fermions. We need to determine the squares of C, R$_j$, T, CR$_j$, R$_j$T, and CT.
Because C and R$_j$ are linear and unitary, while T is anti-linear and anti-unitary, the composition obeys the following rule:
\bea\label{eq:composition}
\rC(z\psi)\rC^{-1}&=&zM_{\rC}\psi^*,\nn\\
\rR_j(z\psi)\rR_j^{-1}&=&zM_{\rR_j}\psi(t,x_1,\dots,x_{j-1},-x_j,x_{j+1},\dots,x_{d-1}),\nn\\
\rT(z\psi)\rT^{-1}&=&z^*M_{\rT}\psi(-t,x)
\eea
where $z$ is a complex number and $M_{\rC}$, $M_{\rR_j}$, and $M_{\rT}$ are the matrices of C, R$_j$, and T transformations respectively.

\cred{In general, for a matrix $M$,
\bea\label{eq:composition-matrix1}
\rC(M\psi)\rC^{-1}&=&MM_{\rC} \psi^*,\nn\\
\rR_j(M \psi)\rR_j^{-1}&=&MM_{\rR_j}\psi(t,x_1,\dots,x_{j-1},-x_j,x_{j+1},\dots,x_{d-1}),\nn\\
\rT(M \psi)\rT^{-1}&=&M^*M_{\rT}\psi(-t,x).
\eea
}
If $M$ is a scalar matrix, then these rules are reduced to \eq{eq:composition}.

Moreover,
 $\psi^*$ actually means the adjoint of the components of $\psi$. 
 
 Since C is unitary, $\rC \rC^{-1}= \rC \rC^{\dagger}=1$, 
 we have $\rC\psi^{\dagger}\rC^{-1}=(\rC\psi\rC^{-1})^{\dagger}$,
and
 $\rC\psi^*\rC^{-1}=(\rC\psi\rC^{-1})^*=M_{\rC}^*\psi$.
 
 Since C is linear, we get $\rC(M\psi^*)\rC^{-1}=MM_{\rC}^*\psi$. 
 
 Similarly, R$_j$ is linear and unitary, and T is anti-linear and anti-unitary.
Therefore,
\bea\label{eq:composition-matrix2}
\rC(M\psi^*)\rC^{-1}&=&MM_{\rC}^*\psi,\nn\\
\rR_j(M \psi^*)\rR_j^{-1}&=&MM_{\rR_j}^*\psi^*(t,x_1,\dots,x_{j-1},-x_j,x_{j+1},\dots,x_{d-1})\nn\\
\rT(M\psi^*)\rT^{-1}&=&M^*M_{\rT}^*\psi^*(-t,x).
\eea

The full matrices of C, R$_j$, and T in $\text{End}_{\C}(S'\oplus\overline{S'})$ are
$\left(\begin{array}{cc}
  0   &  M_{\rC}\\
  M_{\rC}^*   & 0
\end{array}\right)$,
$\left(\begin{array}{cc}
  M_{\rR_j}   &  0\\
  0   & M_{\rR_j}^*
\end{array}\right)$, and
$\left(\begin{array}{cc}
  M_{\rT}   &  0\\
  0   & M_{\rT}^*
\end{array}\right)$ respectively.
Based on these rules, we do the following computation.

For $d=2k+2$ ($k$ odd), $\rC\psi\rC^{-1}=\eta_{\rC}\Gamma^2\Gamma^4\cdots\Gamma^{2k}\psi^*$, then
 \bea
  \rC^2\psi\rC^{-2}&=&\Gamma^2\Gamma^4\cdots\Gamma^{2k}(\Gamma^2\Gamma^4\cdots\Gamma^{2k})^*\psi\nn\\
  &=&(-1)^k\Gamma^2\Gamma^4\cdots\Gamma^{2k}\Gamma^2\Gamma^4\cdots\Gamma^{2k}\psi\nn\\
  &=&(-1)^k(-1)^{k-1}(-1)(-1)^{k-2}(-1)\cdots(-1)\psi\nn\\
  &=&(-1)^{\frac{k(k-1)}{2}}\psi\nn\\
  &=&\left\{\begin{array}{ll}
  \psi &\text{if }k=4l,4l+1,\\
 -\psi&\text{if }k=4l+2,4l+3.
  \end{array}\right.
 \eea

For $d=2k+2$ ($k$ even) and $d=2k+3$, $\rC\psi\rC^{-1}=\eta_{\rC}\Gamma^0\Gamma^1\Gamma^3\cdots\Gamma^{2k+1}\psi^*$, then
 \bea
  \rC^2\psi\rC^{-2}&=&\Gamma^0\Gamma^1\Gamma^3\cdots\Gamma^{2k+1}(\Gamma^0\Gamma^1\Gamma^3\cdots\Gamma^{2k+1})^*\psi\nn\\
  &=&(-1)^{k+1}(-1)^{k}(-1)(-1)^{k-1}(-1)\cdots(-1)\psi\nn\\
  &=&(-1)^{\frac{k(k+1)}{2}}\psi\nn\\
  &=&\left\{\begin{array}{ll}
  \psi &\text{if }k=4l,4l+3,\\
  -\psi&\text{if }k=4l+1,4l+2.
  \end{array}\right.
 \eea

Note that the phase factor of C does not affect the order of C.
 The square of C in $d$d for all $d$ is summarized in \Table{tab:C-square-dd}.
 
  \begin{table}[htp]
\begin{center}
\begin{tabular}{c|c|c}
\hline
Dimension &  & $\rC^2\psi\rC^{-2}$ \\
\hline
$d\equiv 0\mod 8$ & $d=2k+2,\; k\equiv3\mod4$ &$-\psi$ \\
$d\equiv 1\mod 8$ & $d=2k+3,\; k\equiv3\mod4$ &$\psi$ \\
$d\equiv 2\mod 8$ &  $d=2k+2,\; k\equiv0\mod4$ &$\psi$ \\
$d\equiv 3\mod 8$ &  $d=2k+3,\; k\equiv0\mod4$ &$\psi$\\
$d\equiv 4\mod 8$ & $d=2k+2,\; k\equiv1\mod4$ &$\psi$ \\
$d\equiv 5\mod 8$ & $d=2k+3,\; k\equiv1\mod4$ &$-\psi$ \\
$d\equiv 6\mod 8$ & $d=2k+2,\; k\equiv2\mod4$ &$-\psi$ \\
$d\equiv 7\mod 8$ & $d=2k+3,\; k\equiv2\mod4$ &$-\psi$ \\
\hline
\end{tabular}
\caption{The square of C in $d$d for all $d$.}
\label{tab:C-square-dd}
\end{center}
\end{table}

 For $d=2k+2$, $\rR_j\psi\rR_j^{-1}=\eta_{\rR_j}\Gamma^0\Gamma^1\cdots\hat{\Gamma}^j\cdots\Gamma^{2k+1}\psi(t,x_1,\dots,x_{j-1},-x_j,x_{j+1},\dots,x_{d-1})$, then
  \bea
\rR_j^2\psi\rR_j^{-2}&=&\eta_{\rR_j}\Gamma^0\Gamma^1\cdots\hat{\Gamma}^j\cdots\Gamma^{2k+1}\eta_{\rR_j}\Gamma^0\Gamma^1\cdots\hat{\Gamma}^j\cdots\Gamma^{2k+1}\psi(t,x)\nn\\
&=&(-1)^k\eta_{\rR_j}^2\psi(t,x).
 \eea
 For $d=2k+3$, 
 $\rR_j\psi\rR_j^{-1}=\eta_{\rR_j}\Gamma^j\psi(t,x_1,\dots,x_{j-1},-x_j,x_{j+1},\dots,x_{d-1})$, then
  \bea
\rR_j^2\psi\rR_j^{-2}&=&\eta_{\rR_j}\Gamma^j\eta_{\rR_j}\Gamma^j\psi(t,x)=-\eta_{\rR_j}^2\psi(t,x).
 \eea

 The square of $\rR_j$ in $d$d for all $d$ is summarized in \Table{tab:R-square-dd}.
 
  \begin{table}[htp]
\begin{center}
\begin{tabular}{c|c|c}
\hline
Dimension &  & $\rR_j^2\psi\rR_j^{-2}$ \\
\hline
$d\equiv 0\mod 8$ & $d=2k+2,\; k\equiv3\mod4$ &$-\eta_{\rR_j}^2\psi$ \\
$d\equiv 1\mod 8$ & $d=2k+3,\; k\equiv3\mod4$ &$-\eta_{\rR_j}^2\psi$ \\
$d\equiv 2\mod 8$ &  $d=2k+2,\; k\equiv0\mod4$ &$\eta_{\rR_j}^2\psi$ \\
$d\equiv 3\mod 8$ &  $d=2k+3,\; k\equiv0\mod4$ &$-\eta_{\rR_j}^2\psi$\\
$d\equiv 4\mod 8$ & $d=2k+2,\; k\equiv1\mod4$ &$-\eta_{\rR_j}^2\psi$ \\
$d\equiv 5\mod 8$ & $d=2k+3,\; k\equiv1\mod4$ &$-\eta_{\rR_j}^2\psi$ \\
$d\equiv 6\mod 8$ & $d=2k+2,\; k\equiv2\mod4$ &$\eta_{\rR_j}^2\psi$ \\
$d\equiv 7\mod 8$ & $d=2k+3,\; k\equiv2\mod4$ &$-\eta_{\rR_j}^2\psi$ \\
\hline
\end{tabular}
\caption{The square of $\rR_j$ in $d$d for all $d$.}
\label{tab:R-square-dd}
\end{center}
\end{table}

For $d=2k+2$ ($k$ even), $\rT\psi\rT^{-1}=\eta_{\rT}\Gamma^0\Gamma^2\cdots\Gamma^{2k}\psi(-t,x)$, then
 \bea
  \rT^2\psi\rT^{-2}&=&(\Gamma^0\Gamma^2\cdots\Gamma^{2k})^*\Gamma^0\Gamma^2\cdots\Gamma^{2k}\psi\nn\\
  &=&(-1)^{k}(-1)^{k}(-1)^{k-1}(-1)(-1)^{k-2}(-1)\cdots(-1)\psi\nn\\
  &=&(-1)^{\frac{k(k+1)}{2}}\psi\nn\\
  &=&\left\{\begin{array}{ll}
  \psi &\text{if }k=4l,4l+3,\\
  -\psi&\text{if }k=4l+1,4l+2.
  \end{array}\right.
 \eea
 
For $d=2k+2$ ($k$ odd) and $d=2k+3$, $\rT\psi\rT^{-1}=\eta_{\rT}\Gamma^1\Gamma^3\cdots\Gamma^{2k+1}\psi(-t,x)$, then
 \bea
  \rT^2\psi\rT^{-2}&=&(\Gamma^1\Gamma^3\cdots\Gamma^{2k+1})^*\Gamma^1\Gamma^3\cdots\Gamma^{2k+1}\psi\nn\\
  &=&(-1)^{k}(-1)(-1)^{k-1}(-1)\cdots(-1)\psi\nn\\
  &=&(-1)^{\frac{(k+1)(k+2)}{2}}\psi\nn\\
  &=&\left\{\begin{array}{ll}
  -\psi &\text{if }k=4l,4l+1,\\
  \psi&\text{if }k=4l+2,4l+3.
  \end{array}\right.
 \eea
Note that the phase factor of T does not affect the order of T.
  The square of T in $d$d for all $d$ is summarized in \Table{tab:T-square-dd}.
 
  \begin{table}[htp]
\begin{center}
\begin{tabular}{c|c|c}
\hline
Dimension &  & $\rT^2\psi\rT^{-2}$ \\
\hline
$d\equiv 0\mod 8$ & $d=2k+2,\; k\equiv3\mod4$ &$\psi$ \\
$d\equiv 1\mod 8$ & $d=2k+3,\; k\equiv3\mod4$ &$\psi$ \\
$d\equiv 2\mod 8$ &  $d=2k+2,\; k\equiv0\mod4$ &$\psi$ \\
$d\equiv 3\mod 8$ &  $d=2k+3,\; k\equiv0\mod4$ &$-\psi$\\
$d\equiv 4\mod 8$ & $d=2k+2,\; k\equiv1\mod4$ &$-\psi$ \\
$d\equiv 5\mod 8$ & $d=2k+3,\; k\equiv1\mod4$ &$-\psi$ \\
$d\equiv 6\mod 8$ & $d=2k+2,\; k\equiv2\mod4$ &$-\psi$ \\
$d\equiv 7\mod 8$ & $d=2k+3,\; k\equiv2\mod4$ &$\psi$ \\
\hline
\end{tabular}
\caption{The square of T in $d$d for all $d$.}
\label{tab:T-square-dd}
\end{center}
\end{table}

For $d=2k+2$ ($k$ odd), $\rC\psi\rC^{-1}=\eta_{\rC}\Gamma^2\Gamma^4\cdots\Gamma^{2k}\psi^*$,
$$
\rR_j\psi\rR_j^{-1}=\eta_{\rR_j}\Gamma^0\Gamma^1\cdots\hat{\Gamma}^j\cdots\Gamma^{2k+1}\psi(t,x_1,\dots,x_{j-1},-x_j,x_{j+1},\dots,x_{d-1}),
$$
then
\bea
\rC\rR_j\psi(\rC\rR_j)^{-1}=\eta_{\rR_j}\Gamma^0\Gamma^1\cdots\hat{\Gamma}^j\cdots\Gamma^{2k+1}\eta_{\rC}\Gamma^2\Gamma^4\cdots\Gamma^{2k}\psi^*(t,x_1,\dots,x_{j-1},-x_j,x_{j+1},\dots,x_{d-1}).
\eea
  Thus
 \bea
(\rC\rR_j)^2\psi(\rC\rR_j)^{-2}&=&\Gamma^0\Gamma^1\cdots\hat{\Gamma}^j\cdots\Gamma^{2k+1}\Gamma^2\Gamma^4\cdots\Gamma^{2k}(\Gamma^0\Gamma^1\cdots\hat{\Gamma}^j\cdots\Gamma^{2k+1})^*(\Gamma^2\Gamma^4\cdots\Gamma^{2k})^*\psi\nn\\
  &=&(-1)^{\frac{k(k-1)}{2}}\psi\nn\\
  &=&\left\{\begin{array}{ll}
  \psi &\text{if }k=4l,4l+1,\\
  -\psi&\text{if }k=4l+2,4l+3.
  \end{array}\right.
 \eea

For $d=2k+2$ ($k$ even), $\rC\psi\rC^{-1}=\eta_{\rC}\Gamma^0\Gamma^1\Gamma^3\cdots\Gamma^{2k+1}\psi^*$, 
$$
\rR_j\psi\rR_j^{-1}=\eta_{\rR_j}\Gamma^0\Gamma^1\cdots\hat{\Gamma}^j\cdots\Gamma^{2k+1}\psi(t,x_1,\dots,x_{j-1},-x_j,x_{j+1},\dots,x_{d-1}),
$$
then
\bea
\rC\rR_j\psi(\rC\rR_j)^{-1}=\eta_{\rR_j}\Gamma^0\Gamma^1\cdots\hat{\Gamma}^j\cdots\Gamma^{2k+1}\eta_{\rC}\Gamma^0\Gamma^1\Gamma^3\cdots\Gamma^{2k+1}\psi^*(t,x_1,\dots,x_{j-1},-x_j,x_{j+1},\dots,x_{d-1}).
\eea
  Thus
 \bea
(\rC\rR_j)^2\psi(\rC\rR_j)^{-2}&=&\Gamma^0\Gamma^1\cdots\hat{\Gamma}^j\cdots\Gamma^{2k+1}\Gamma^0\Gamma^1\Gamma^3\cdots\Gamma^{2k+1}(\Gamma^0\Gamma^1\cdots\hat{\Gamma}^j\cdots\Gamma^{2k+1})^*(\Gamma^0\Gamma^1\Gamma^3\cdots\Gamma^{2k+1})^*\psi\nn\\
  &=&-(-1)^{\frac{k(k+1)}{2}}\psi\nn\\
  &=&\left\{\begin{array}{ll}
  -\psi &\text{if }k=4l,4l+3,\\
  \psi&\text{if }k=4l+1,4l+2.
  \end{array}\right.
 \eea

 For $d=2k+3$, $\rC\psi\rC^{-1}=\eta_{\rC}\Gamma^0\Gamma^1\Gamma^3\cdots\Gamma^{2k+1}\psi^*$, $\rR_j\psi\rR_j^{-1}=\eta_{\rR_j}\Gamma^j\psi(t,x_1,\dots,x_{j-1},-x_j,x_{j+1},\dots,x_{d-1})$, then 
 \bea
 \rC\rR_j\psi(\rC\rR_j)^{-1}=\eta_{\rR_j}\Gamma^j\eta_{\rC}\Gamma^0\Gamma^1\Gamma^3\cdots\Gamma^{2k+1}\psi^*(t,x_1,\dots,x_{j-1},-x_j,x_{j+1},\dots,x_{d-1}).
 \eea
 Thus
 \bea
 (\rC\rR_j)^2\psi(\rC\rR_j)^{-2}&=&\Gamma^j\Gamma^0\Gamma^1\Gamma^3\cdots\Gamma^{2k+1}(\Gamma^j)^*(\Gamma^0\Gamma^1\Gamma^3\cdots\Gamma^{2k+1})^*\psi\nn\\
 &=&\left\{\begin{array}{ll}
    \psi &\text{if }k=4l,4l+1,\\
  -\psi&\text{if }k=4l+2,4l+3.
  \end{array}\right.
 \eea

 Note that the phase factors of C and R$_j$ do not affect the order of CR$_j$.
   The square of CR$_j$ in $d$d for all $d$ is summarized in \Table{tab:CP-square-dd}.
 
  \begin{table}[htp]
\begin{center}

\begin{tabular}{c|c|c}
\hline
Dimension &  & $(\rC\rR_j)^2\psi(\rC\rR_j)^{-2}$ \\
\hline
$d\equiv 0\mod 8$ & $d=2k+2,\; k\equiv3\mod4$ &$-\psi$ \\
$d\equiv 1\mod 8$ & $d=2k+3,\; k\equiv3\mod4$ &$-\psi$ \\
$d\equiv 2\mod 8$ &  $d=2k+2,\; k\equiv0\mod4$ &$-\psi$ \\
$d\equiv 3\mod 8$ &  $d=2k+3,\; k\equiv0\mod4$ &$\psi$\\
$d\equiv 4\mod 8$ & $d=2k+2,\; k\equiv1\mod4$ &$\psi$ \\
$d\equiv 5\mod 8$ & $d=2k+3,\; k\equiv1\mod4$ &$\psi$ \\
$d\equiv 6\mod 8$ & $d=2k+2,\; k\equiv2\mod4$ &$\psi$ \\
$d\equiv 7\mod 8$ & $d=2k+3,\; k\equiv2\mod4$ &$-\psi$ \\
\hline
\end{tabular}

\caption{The square of CR$_j$ in $d$d for all $d$. }
\label{tab:CP-square-dd}
\end{center}
\end{table}

 For $d=2k+2$ ($k$ even), $\rR_j\psi\rR_j^{-1}=\eta_{\rR_j}\Gamma^0\Gamma^1\cdots\hat{\Gamma}^j\cdots\Gamma^{2k+1}\psi(t,x_1,\dots,x_{j-1},-x_j,x_{j+1},\dots,x_{d-1})$, $\rT\psi\rT^{-1}=\eta_{\rT}\Gamma^0\Gamma^2\cdots\Gamma^{2k}\psi(-t,x)$, then
\bea  (\rR_j\rT)\psi(\rR_j\rT)^{-1}=\eta_{\rT}\Gamma^0\Gamma^2\cdots\Gamma^{2k}\eta_{\rR_j}\Gamma^0\Gamma^1\cdots\hat{\Gamma}^j\cdots\Gamma^{2k+1}\psi(-t,x_1,\dots,x_{j-1},-x_j,x_{j+1},\dots,x_{d-1}).
\eea
  Thus
 \bea
(\rR_j\rT)^2\psi(\rR_j\rT)^{-2}&=&(\Gamma^0\Gamma^2\cdots\Gamma^{2k}\Gamma^0\Gamma^1\cdots\hat{\Gamma}^j\cdots\Gamma^{2k+1})^*\Gamma^0\Gamma^2\cdots\Gamma^{2k}\Gamma^0\Gamma^1\cdots\hat{\Gamma}^j\cdots\Gamma^{2k+1}\psi\nn\\
  &=&(-1)^{\frac{k(k+1)}{2}}\psi\nn\\
  &=&\left\{\begin{array}{ll}
  \psi &\text{if }k=4l,4l+3,\\
  -\psi&\text{if }k=4l+1,4l+2.
  \end{array}\right.
 \eea

For $d=2k+2$ ($k$ odd), $\rR_j\psi\rR_j^{-1}=\eta_{\rR_j}\Gamma^0\Gamma^1\cdots\hat{\Gamma}^j\cdots\Gamma^{2k+1}\psi(t,x_1,\dots,x_{j-1},-x_j,x_{j+1},\dots,x_{d-1})$, $\rT\psi\rT^{-1}=\eta_{\rT}\Gamma^1\Gamma^3\cdots\Gamma^{2k+1}\psi(-t,x)$, then
\bea  (\rR_j\rT)\psi(\rR_j\rT)^{-1}=\eta_{\rT}\Gamma^1\Gamma^3\cdots\Gamma^{2k+1}\eta_{\rR_j}\Gamma^0\Gamma^1\cdots\hat{\Gamma}^j\cdots\Gamma^{2k+1}\psi(-t,x_1,\dots,x_{j-1},-x_j,x_{j+1},\dots,x_{d-1}).
\eea
  Thus
 \bea
(\rR_j\rT)^2\psi(\rR_j\rT)^{-2}&=&(\Gamma^1\Gamma^3\cdots\Gamma^{2k+1}\Gamma^0\Gamma^1\cdots\hat{\Gamma}^j\cdots\Gamma^{2k+1})^*\Gamma^1\Gamma^3\cdots\Gamma^{2k+1}\Gamma^0\Gamma^1\cdots\hat{\Gamma}^j\cdots\Gamma^{2k+1}\psi\nn\\
  &=&-(-1)^{\frac{(k+1)(k+2)}{2}}\psi\nn\\
  &=&\left\{\begin{array}{ll}
  -\psi &\text{if }k=4l,4l+3,\\
  \psi&\text{if }k=4l+1,4l+2.
  \end{array}\right.
 \eea
 
 For $d=2k+3$, 
 $\rR_j\psi\rR_j^{-1}=\eta_{\rR_j}\Gamma^j\psi(t,x_1,\dots,x_{j-1},-x_j,x_{j+1},\dots,x_{d-1})$, $\rT\psi\rT^{-1}=\eta_{\rT}\Gamma^1\Gamma^3\cdots\Gamma^{2k+1}\psi(-t,x)$, then
 \bea
 (\rR_j\rT)\psi(\rR_j\rT)^{-1}=\eta_{\rT}\Gamma^1\Gamma^3\cdots\Gamma^{2k+1}\eta_{\rR_j}\Gamma^j\psi(-t,x_1,\dots,x_{j-1},-x_j,x_{j+1},\dots,x_{d-1}).
 \eea
Thus
 \bea
 (\rR_j\rT)^2\psi(\rR_j\rT)^{-2}&=&(\Gamma^1\Gamma^3\cdots\Gamma^{2k+1}\Gamma^j)^*\Gamma^1\Gamma^3\cdots\Gamma^{2k+1}\Gamma^j\psi\nn\\
 &=&(-1)^{\frac{k(k+1)}{2}}\psi\nn\\
  &=&\left\{\begin{array}{ll}
  \psi &\text{if }k=4l,4l+3,\\
  -\psi&\text{if }k=4l+1,4l+2.
  \end{array}\right.
 \eea
  Note that the phase factors of R$_j$ and T do not affect the order of R$_j$T.
    The square of R$_j$T in $d$d for all $d$ is summarized in \Table{tab:PT-square-dd}.
 
  \begin{table}[htp]
\begin{center}

\begin{tabular}{c|c|c}
\hline
Dimension &  &  $(\rR_j\rT)^2\psi(\rR_j\rT)^{-2}$ \\
\hline
$d\equiv 0\mod 8$ & $d=2k+2,\; k\equiv3\mod4$ &$-\psi$ \\
$d\equiv 1\mod 8$ & $d=2k+3,\; k\equiv3\mod4$ &$\psi$ \\
$d\equiv 2\mod 8$ &  $d=2k+2,\; k\equiv0\mod4$ &$\psi$ \\
$d\equiv 3\mod 8$ &  $d=2k+3,\; k\equiv0\mod4$ &$\psi$\\
$d\equiv 4\mod 8$ & $d=2k+2,\; k\equiv1\mod4$ &$\psi$ \\
$d\equiv 5\mod 8$ & $d=2k+3,\; k\equiv1\mod4$ &$-\psi$ \\
$d\equiv 6\mod 8$ & $d=2k+2,\; k\equiv2\mod4$ &$-\psi$ \\
$d\equiv 7\mod 8$ & $d=2k+3,\; k\equiv2\mod4$ &$-\psi$ \\
\hline
\end{tabular}

\caption{The square of R$_j$T in $d$d for all $d$. }
\label{tab:PT-square-dd}
\end{center}
\end{table}

 For $d=2k+2$ ($k$ odd),
 $\rT\psi\rT^{-1}=\eta_{\rT}\Gamma^1\Gamma^3\cdots\Gamma^{2k+1}\psi(-t,x)$ and $\rC\psi\rC^{-1}=\eta_{\rC}\Gamma^2\Gamma^4\cdots\Gamma^{2k}\psi^*$. Therefore, 
 $(\rC\rT)\psi(\rC\rT)^{-1}=\eta_{\rT}\Gamma^1\Gamma^3\cdots\Gamma^{2k+1}\eta_{\rC}\Gamma^2\Gamma^4\cdots\Gamma^{2k}\psi^*(-t,x)$.
 Thus
 \bea
 (\rC\rT)^2\psi(\rC\rT)^{-2}&=&\rC(\eta_{\rT}\Gamma^1\Gamma^3\cdots\Gamma^{2k+1}\eta_{\rC}\Gamma^2\Gamma^4\cdots\Gamma^{2k})^*(\eta_{\rT}\Gamma^1\Gamma^3\cdots\Gamma^{2k+1})^*\psi^*(t,x)\rC^{-1}\nn\\
 &=&(\eta_{\rT}\Gamma^1\Gamma^3\cdots\Gamma^{2k+1}\eta_{\rC}\Gamma^2\Gamma^4\cdots\Gamma^{2k})^*(\eta_{\rT}\Gamma^1\Gamma^3\cdots\Gamma^{2k+1})^*(\eta_{\rC}\Gamma^2\Gamma^4\cdots\Gamma^{2k})^*\psi(t,x)\nn\\
  &=&(\eta_{\rT}^*\eta_{\rC}^*)^2(-1)^{k+1}\psi\nn\\
 &=&(\eta_{\rT}^*\eta_{\rC}^*)^2\psi.
 \eea

  For $d=2k+2$ ($k$ even),
 $\rT\psi\rT^{-1}=\eta_{\rT}\Gamma^0\Gamma^2\cdots\Gamma^{2k}\psi(-t,x)$ and $\rC\psi\rC^{-1}=\eta_{\rC}\Gamma^0\Gamma^1\Gamma^3\cdots\Gamma^{2k+1}\psi^*$. Therefore, 
 $(\rC\rT)\psi(\rC\rT)^{-1}=\eta_{\rT}\Gamma^0\Gamma^2\cdots\Gamma^{2k}\eta_{\rC}\Gamma^0\Gamma^1\Gamma^3\cdots\Gamma^{2k+1}\psi^*(-t,x)$.
 Thus
 \bea
 (\rC\rT)^2\psi(\rC\rT)^{-2}&=&\rC(\eta_{\rT}\Gamma^0\Gamma^2\cdots\Gamma^{2k}\eta_{\rC}\Gamma^0\Gamma^1\Gamma^3\cdots\Gamma^{2k+1})^*(\eta_{\rT}\Gamma^0\Gamma^2\cdots\Gamma^{2k})^*\psi^*(t,x)\rC^{-1}\nn\\
 &=&(\eta_{\rT}\Gamma^0\Gamma^2\cdots\Gamma^{2k}\eta_{\rC}\Gamma^0\Gamma^1\Gamma^3\cdots\Gamma^{2k+1})^*(\eta_{\rT}\Gamma^0\Gamma^2\cdots\Gamma^{2k})^*(\eta_{\rC}\Gamma^0\Gamma^1\Gamma^3\cdots\Gamma^{2k+1})^*\psi(t,x)\nn\\
 &=&(\eta_{\rT}^*\eta_{\rC}^*)^2(-1)^{k+1}\psi\nn\\
 &=&-(\eta_{\rT}^*\eta_{\rC}^*)^2\psi.
 \eea

 For $d=2k+3$,
$\rT\psi\rT^{-1}=\eta_{\rT}\Gamma^1\Gamma^3\cdots\Gamma^{2k+1}\psi(-t,x)$ and $\rC\psi\rC^{-1}=\eta_{\rC}\Gamma^0\Gamma^1\Gamma^3\cdots\Gamma^{2k+1}\psi^*$. Therefore, 
 $(\rC\rT)\psi(\rC\rT)^{-1}=\eta_{\rT}\Gamma^1\Gamma^3\cdots\Gamma^{2k+1}\eta_{\rC}\Gamma^0\Gamma^1\Gamma^3\cdots\Gamma^{2k+1}\psi^*(-t,x)$.
 Thus
 \bea
 (\rC\rT)^2\psi(\rC\rT)^{-2}&=&\rC(\eta_{\rT}\Gamma^1\Gamma^3\cdots\Gamma^{2k+1}\eta_{\rC}\Gamma^0\Gamma^1\Gamma^3\cdots\Gamma^{2k+1})^*(\eta_{\rT}\Gamma^1\Gamma^3\cdots\Gamma^{2k+1})^*\psi^*(t,x)\rC^{-1}\nn\\
 &=&(\eta_{\rT}\Gamma^1\Gamma^3\cdots\Gamma^{2k+1}\eta_{\rC}\Gamma^0\Gamma^1\Gamma^3\cdots\Gamma^{2k+1})^*(\eta_{\rT}\Gamma^1\Gamma^3\cdots\Gamma^{2k+1})^*(\eta_{\rC}\Gamma^0\Gamma^1\Gamma^3\cdots\Gamma^{2k+1})^*\psi(t,x)\nn\\
 &=&(\eta_{\rT}^*\eta_{\rC}^*)^2\psi.
 \eea

 The square of CT in $d$d for all $d$ is summarized in \Table{tab:CT-square-dd}.
  
 \begin{table}[htp]
\begin{center}

\begin{tabular}{c|c|c}
\hline
Dimension &  & $(\rC\rT)^2\psi(\rC\rT)^{-2}$ \\
\hline
$d\equiv 0\mod 8$ & $d=2k+2,\; k\equiv3\mod4$ &$(\eta_{\rT}^*\eta_{\rC}^*)^2\psi$ \\
$d\equiv 1\mod 8$ & $d=2k+3,\; k\equiv3\mod4$ &$(\eta_{\rT}^*\eta_{\rC}^*)^2\psi$ \\
$d\equiv 2\mod 8$ &  $d=2k+2,\; k\equiv0\mod4$ &$-(\eta_{\rT}^*\eta_{\rC}^*)^2\psi$ \\
$d\equiv 3\mod 8$ &  $d=2k+3,\; k\equiv0\mod4$ &$(\eta_{\rT}^*\eta_{\rC}^*)^2\psi$\\
$d\equiv 4\mod 8$ & $d=2k+2,\; k\equiv1\mod4$ &$(\eta_{\rT}^*\eta_{\rC}^*)^2\psi$ \\
$d\equiv 5\mod 8$ & $d=2k+3,\; k\equiv1\mod4$ &$(\eta_{\rT}^*\eta_{\rC}^*)^2\psi$ \\
$d\equiv 6\mod 8$ & $d=2k+2,\; k\equiv2\mod4$ &$-(\eta_{\rT}^*\eta_{\rC}^*)^2\psi$ \\
$d\equiv 7\mod 8$ & $d=2k+3,\; k\equiv2\mod4$ &$(\eta_{\rT}^*\eta_{\rC}^*)^2\psi$ \\
\hline
\end{tabular}

\caption{The square of CT in $d$d for all $d$. }
\label{tab:CT-square-dd}
\end{center}
\end{table}

Now we consider the canonical CRT \cite{HasonKomargodskiThorngren1910.14039}.
By \eqref{eq:CPT}, we have $(\rC\rR_j\rT)^2=1$.
Note that 
\bea\label{eq:CRT2}
(\rC\rR_j\rT)^2\psi(\rC\rR_j\rT)^{-2}=(M_{\rT}M_{\rR_j}M_{\rC})^*M_{\rT}^*M_{\rR_j}^*M_{\rC}^*\psi.
\eea
By \Table{tab:CPRT-trans-dd}, $M_{\rT}M_{\rR_j}M_{\rC}=\eta_{\rT}\eta_{\rR_j}\eta_{\rC}\Gamma^0\Gamma^j$,
we get $(\eta_{\rT}\eta_{\rR_j}\eta_{\rC})^2=1$ and $\eta_{\rT}\eta_{\rR_j}\eta_{\rC}=\pm1$.

By \eqref{eq:U-CPT}, we have $\rC\cdot\rC\rR_j\rT=\rC\rR_j\rT\cdot\rC$, hence $\rC\rR_j\rT=\rR_j\rT\rC$ and $\rC\rC\rR_j\rT\rR_j\rT=\rC\rR_j\rT\rC\rR_j\rT$. By \eqref{eq:CPT}, we have $\rC^2(\rR_j\rT)^2=1$. This is consistent with \Table{tab:C-square-dd} and \Table{tab:PT-square-dd}. Conversely, $\rC^2(\rR_j\rT)^2=1$ and \eqref{eq:CPT} imply \eqref{eq:U-CPT} for $\U=\rC$.

By \eqref{eq:T-CPT}, we have 
$\rT\cdot\rC\rR_j\rT=(-1)^{\rF}\rC\rR_j\rT\cdot\rT$, hence $\rC\rR_j\rT=(-1)^{\rF}\rT\rC\rR_j$ and $\rC\rR_j\rC\rR_j\rT\rT=(-1)^{\rF}\rC\rR_j\rT\rC\rR_j\rT$. Here we have used the fact that $(-1)^{\rF}$ commutes with $\rC\rR_j$. By \eqref{eq:CPT}, we have $(\rC\rR_j)^2\rT^2=(-1)^{\rF}$. This is consistent with \Table{tab:T-square-dd} and \Table{tab:CP-square-dd}.
Conversely, $(\rC\rR_j)^2\rT^2=(-1)^{\rF}$ and \eqref{eq:CPT} imply \eqref{eq:T-CPT}.

Combining \Table{tab:C-square-dd}, \Table{tab:R-square-dd}, \Table{tab:T-square-dd}, \Table{tab:CP-square-dd}, \Table{tab:PT-square-dd}, and \Table{tab:CT-square-dd},
we get \Table{tab:CPRT-square-dd-1} and \Table{tab:CPRT-square-dd-2}.
The CRT for our two possibilities satisfies \eqref{eq:CPT}, $\rC^2(\rR_j\rT)^2=1$, and $\rT^2(\rC\rR_j)^2=(-1)^{\rF}$, hence also satisfies \eqref{eq:T-CPT} and \eqref{eq:U-CPT} for $\U=\rC$. We will check in \App{app:further} that it also satisfies \eqref{eq:U-CPT} for the internal charge-like and isospin-like U(1) symmetries. Thus the CRT is canonical.

\begin{table}[htp]
\begin{center}

\begin{tabular}{c|c|c|c|c|c|c}
\hline
Dimension & $\rC^2$ & $\rR_j^2$ &$\rT^2$ &  $(\rC\rR_j)^2$ & $(\rR_j\rT)^2$ &$(\rC\rT)^2$\\
\hline
$d\equiv 0\mod 8$ & $(-1)^{\rF}$ &$(-1)^{\rF}$&1&$(-1)^{\rF}$&$(-1)^{\rF}$&1 \\
$d\equiv 1\mod 8$ & 1&$(-1)^{\rF}$&1&$(-1)^{\rF}$&1&1 \\
$d\equiv 2\mod 8$ & 1&1&1&$(-1)^{\rF}$&1&$(-1)^{\rF}$  \\
$d\equiv 3\mod 8$ & 1&$(-1)^{\rF}$&$(-1)^{\rF}$&1&1&1 \\
$d\equiv 4\mod 8$ & 1&$(-1)^{\rF}$&$(-1)^{\rF}$&1&1&1 \\
$d\equiv 5\mod 8$ & $(-1)^{\rF}$ &$(-1)^{\rF}$ &$(-1)^{\rF}$&1&$(-1)^{\rF}$&1  \\
$d\equiv 6\mod 8$ &$(-1)^{\rF}$ &1&$(-1)^{\rF}$&1&$(-1)^{\rF}$&$(-1)^{\rF}$ \\
$d\equiv 7\mod 8$ & $(-1)^{\rF}$&$(-1)^{\rF}$&1&$(-1)^{\rF}$&$(-1)^{\rF}$&1 \\
\hline
\end{tabular}

\caption{The first possibility for the squares of C, R$_j$, T, CR$_j$, R$_j$T, and CT with the canonical CRT in $d$d for all $d$.   }
\label{tab:CPRT-square-dd-1}
\end{center}
\end{table}

\begin{table}[htp]
\begin{center}

\begin{tabular}{c|c|c|c|c|c|c}
\hline
Dimension & $\rC^2$ &$\rR_j^2$ &$\rT^2$ &  $(\rC\rR_j)^2$ & $(\rR_j\rT)^2$ &$(\rC\rT)^2$\\
\hline
$d\equiv 0\mod 8$ & $(-1)^{\rF}$ &1&1&$(-1)^{\rF}$&$(-1)^{\rF}$&$(-1)^{\rF}$ \\
$d\equiv 1\mod 8$ & 1&1&1&$(-1)^{\rF}$&1&$(-1)^{\rF}$ \\
$d\equiv 2\mod 8$ & 1&$(-1)^{\rF}$&1&$(-1)^{\rF}$&1&1  \\
$d\equiv 3\mod 8$ & 1&1&$(-1)^{\rF}$&1&1&$(-1)^{\rF}$ \\
$d\equiv 4\mod 8$ & 1&1&$(-1)^{\rF}$&1&1&$(-1)^{\rF}$ \\
$d\equiv 5\mod 8$ & $(-1)^{\rF}$ &1 &$(-1)^{\rF}$&1&$(-1)^{\rF}$&$(-1)^{\rF}$  \\
$d\equiv 6\mod 8$ &$(-1)^{\rF}$ &$(-1)^{\rF}$&$(-1)^{\rF}$&1&$(-1)^{\rF}$&1 \\
$d\equiv 7\mod 8$ & $(-1)^{\rF}$&1&1&$(-1)^{\rF}$&$(-1)^{\rF}$&$(-1)^{\rF}$ \\
\hline
\end{tabular}

\caption{The second possibility for the squares of C, R$_j$, T, CR$_j$, R$_j$T, and CT with the canonical CRT in $d$d for all $d$.  }
\label{tab:CPRT-square-dd-2}
\end{center}
\end{table}

The extension class of the group extension \eqref{eq:extension} is represented by a 2-cocycle $f$ on $\Z_2^{\rC}\times\Z_2^{\rR_j}\times\Z_2^{\rT}$ valued in $\Z_2^{\rF}$. Note that for $a\in \Z_2^{\rC}\times\Z_2^{\rR_j}\times\Z_2^{\rT}$ and $b\in\Z_2^{\rF}$, $(a,b)^2=(a+a,b+b+f(a,a))=(0,f(a,a))$. 
We can determine the extension class of the C-R$_j$-T fractionalization from the squares of C, R$_j$, T, CR$_j$, R$_j$T, and CT.

We summarize the two possibilities of the C-R$_j$-T fractionalization with the canonical CRT and its extension class for a Dirac fermion in \Table{tab:CPRT-group-dd-1} and \Table{tab:CPRT-group-dd-2}. 
We find that the C-R$_j$-T fractionalizations for a Dirac fermion only depend on the remainder of the spacetime dimension $d$ modulo 8. Hence we have proved the second part of Theorem \ref{main}.
Note that the case $d=1$ is special, there is only one possibility because there is no R$_j$ for $d=1$.

\begin{table}[htp]
\begin{center}
\hspace*{-1cm}
\begin{tabular}{c|c|c|c}
\hline

Dimension &
 Group presentation & 
 Group & 
 Extension class\\
\hline

$d\equiv 0\mod 8$ & $\rC^2=\rR_j^2=(\rC\rR_j)^2=(\rR_j\rT)^2=(-1)^{\rF}$, $\rT^2=(\rC\rT)^2=1$ & $\frac{\bD_8^{\rF,\rC,\rT}\times \Z_4^{\rR_j\rT\rF}}{\Z_2^{\rF}}$ &$(a^{\rC})^2+(a^{\rR_j})^2+a^{\rC}a^{\rR_j}+a^{\rC}a^{\rT}$\\
\hline
$d=1$ &   $\rC^2=\rT^2=(\rC\rT)^2=1$ & $\Z_2^{\rC}\times\Z_2^{\rT}\times\Z_2^{\rF}$ & $0$\\
\hline
$d\equiv 1\mod 8$ ($d>1$) &  $\rR_j^2=(\rC\rR_j)^2=(-1)^{\rF}$, $\rC^2=\rT^2=(\rR_j\rT)^2=(\rC\rT)^2=1$ & $\bD_8^{\rF,\rC\rR_j ,\rC\rT}\times \Z_2^{\rC}$ & $(a^{\rR_j})^2+a^{\rR_j}a^{\rT}$\\
\hline
$d\equiv 2\mod 8$ &  $(\rC\rR_j)^2=(\rC\rT)^2=(-1)^{\rF}$, $\rC^2=\rR_j^2=\rT^2=(\rR_j\rT)^2=1$ & $\bD_8^{\rF,\rC \rT,\rC}\times \Z_2^{\rR_j\rT}$ & $a^{\rC}a^{\rR_j}+a^{\rC}a^{\rT}$\\
\hline
$d\equiv 3\mod 8$ & $\rR_j^2=\rT^2=(-1)^{\rF}$, $\rC^2=(\rC\rR_j)^2=(\rR_j\rT)^2=(\rC\rT)^2=1$ &  $\bD_8^{\rF,\rT ,\rC\rR_j}\times \Z_2^{\rR_j\rT}$ & $(a^{\rR_j}+a^{\rT})(a^{\rC}+a^{\rR_j}+a^{\rT})$\\
\hline
$d\equiv 4\mod 8$ & $\rR_j^2=\rT^2=(-1)^{\rF}$, $\rC^2=(\rC\rR_j)^2=(\rR_j\rT)^2=(\rC\rT)^2=1$ &  $\bD_8^{\rF,\rT ,\rC\rR_j}\times \Z_2^{\rR_j\rT}$ & $(a^{\rR_j}+a^{\rT})(a^{\rC}+a^{\rR_j}+a^{\rT})$\\
\hline
$d\equiv 5\mod 8$ & $\rC^2=\rR_j^2=\rT^2=(\rR_j\rT)^2=(-1)^{\rF}$, $(\rC\rR_j)^2=(\rC\rT)^2=1$ &  $\frac{\bD_8^{\rF,\rT,\rC\rR_j}\times \Z_4^{\rC\rF}}{\Z_2^{\rF}}$ & $(a^{\rC})^2+(a^{\rR_j})^2+(a^{\rT})^2+a^{\rR_j}a^{\rT}$\\
\hline
$d\equiv 6\mod 8$ & $\rR_j^2=(\rC\rR_j)^2=1$, $\rC^2=\rT^2=(\rR_j\rT)^2=(\rC\rT)^2=(-1)^{\rF}$ & $\frac{\bD_8^{\rF,\rT, \rC\rR_j\rT}\times \Z_4^{\rR_j\rT\rF}}{\Z_2^{\rF}}$ & $(a^{\rC})^2+(a^{\rT})^2+a^{\rC}a^{\rR_j}+a^{\rC}a^{\rT}$\\
\hline
$d\equiv 7\mod 8$ & $\rC^2=\rR_j^2=(\rC\rR_j)^2=(\rR_j\rT)^2=(-1)^{\rF}$, $\rT^2=(\rC\rT)^2=1$ & $\frac{\bD_8^{\rF,\rC,\rT}\times \Z_4^{\rR_j\rT\rF}}{\Z_2^{\rF}}$ &$(a^{\rC})^2+(a^{\rR_j})^2+a^{\rC}a^{\rR_j}+a^{\rC}a^{\rT}$\\
\hline
\end{tabular}
\caption{The first possibility for the structure of the C-R$_j$-T fractionalization with the canonical CRT for a Dirac fermion and its extension class in $d$d for all $d$.
Here $\bD_8^{\rF,x ,y}\equiv\langle x,y |x^4=y^2=1, yxy=x^{-1}, x^2=(-1)^{\rF}\rangle$, and $\frac{\bD_8^{\rF,x,y}\times \Z_4^{z\rF}}{\Z_2^{\rF}}\equiv\langle x,y,z| x^4=y^2=z^4=1,  yxy=x^{-1}, zx=xz, zy=yz,  x^2=z^2=(-1)^{\rF}\rangle $.  Here $a^X$ ($X=\rC$, $\rR_j$, and $\rT$) are the three generators of $\H^1(\B(\Z_2^{\rC}\times\Z_2^{\rR_j}\times\Z_2^{\rT}),\Z_2^{\rF})=\Z_2^3$ respectively.}
\label{tab:CPRT-group-dd-1}
\end{center}
\end{table}

\begin{table}[htp]
\begin{center}
\hspace*{-1cm}
\begin{tabular}{c|c|c|c}
\hline
 Dimension &  Group presentation & Group & Extension class\\
\hline
$d\equiv 0\mod 8$ & $\rC^2=(\rC\rR_j)^2=(\rR_j\rT)^2=(\rC\rT)^2=(-1)^{\rF}$, $\rR_j^2=\rT^2=1$ & $\frac{\bD_8^{\rF,\rC\rR_j,\rT}\times \Z_4^{\rC\rF}}{\Z_2^{\rF}}$ &$(a^{\rC})^2+a^{\rR_j}a^{\rT}$\\
\hline
$d=1$ &   $\rC^2=\rT^2=(\rC\rT)^2=1$ & $\Z_2^{\rC}\times\Z_2^{\rT}\times\Z_2^{\rF}$ &$0$\\
\hline
$d\equiv 1\mod 8$ ($d>1$) &  $(\rC\rR_j)^2=(\rC\rT)^2=(-1)^{\rF}$, $\rC^2=\rR_j^2=\rT^2=(\rR_j\rT)^2=1$ & $\bD_8^{\rF,\rC\rR_j ,\rC\rR_j\rT}\times \Z_2^{\rR_j\rT}$ &$a^{\rC}a^{\rR_j}+a^{\rC}a^{\rT}$\\
\hline
$d\equiv 2\mod 8$ &  $\rR_j^2=(\rC\rR_j)^2=(-1)^{\rF}$, $\rC^2=\rT^2=(\rR_j\rT)^2=(\rC\rT)^2=1$ & $\bD_8^{\rF,\rC \rR_j,\rC\rR_j\rT}\times \Z_2^{\rC}$ & $(a^{\rR_j})^2+a^{\rR_j}a^{\rT}$\\
\hline
$d\equiv 3\mod 8$ & $\rT^2=(\rC\rT)^2=(-1)^{\rF}$, $\rC^2=\rR_j^2=(\rC\rR_j)^2=(\rR_j\rT)^2=1$ &  $\bD_8^{\rF,\rT,\rR_j }\times \Z_2^{\rC}$ & $(a^{\rT})^2+a^{\rR_j}a^{\rT}$\\
\hline
$d\equiv 4\mod 8$ & $\rT^2=(\rC\rT)^2=(-1)^{\rF}$, $\rC^2=\rR_j^2=(\rC\rR_j)^2=(\rR_j\rT)^2=1$ &  $\bD_8^{\rF,\rT,\rR_j }\times \Z_2^{\rC}$ & $(a^{\rT})^2+a^{\rR_j}a^{\rT}$\\
\hline
$d\equiv 5\mod 8$ & $\rC^2=\rT^2=(\rR_j\rT)^2=(\rC\rT)^2=(-1)^{\rF}$, $\rR_j^2=(\rC\rR_j)^2=1$ &  $\frac{\bD_8^{\rF,\rT,\rC\rR_j\rT}\times \Z_4^{\rR_j\rT\rF}}{\Z_2^{\rF}}$& $(a^{\rC})^2+(a^{\rT})^2+a^{\rC}a^{\rR_j}+a^{\rC}a^{\rT}$\\
\hline
$d\equiv 6\mod 8$ & $(\rC\rR_j)^2=(\rC\rT)^2=1$, $\rC^2=\rR_j^2=\rT^2=(\rR_j\rT)^2=(-1)^{\rF}$ & $\frac{\bD_8^{\rF,\rR_j, \rC\rR_j\rT}\times \Z_4^{\rC\rF}}{\Z_2^{\rF}}$ &$(a^{\rC})^2+(a^{\rR_j})^2+(a^{\rT})^2+a^{\rR_j}a^{\rT}$\\
\hline
$d\equiv 7\mod 8$ & $\rC^2=(\rC\rR_j)^2=(\rR_j\rT)^2=(\rC\rT)^2=(-1)^{\rF}$, $\rR_j^2=\rT^2=1$ & $\frac{\bD_8^{\rF,\rC\rR_j,\rT}\times \Z_4^{\rC\rF}}{\Z_2^{\rF}}$ &$(a^{\rC})^2+a^{\rR_j}a^{\rT}$\\
\hline
\end{tabular}
\caption{The second possibility for the structure of the C-R$_j$-T fractionalization with the canonical CRT for a Dirac fermion and its extension class in $d$d for all $d$.
Here $\bD_8^{\rF,x ,y}\equiv\langle x,y |x^4=y^2=1, yxy=x^{-1}, x^2=(-1)^{\rF}\rangle$, and $\frac{\bD_8^{\rF,x,y}\times \Z_4^{z\rF}}{\Z_2^{\rF}}\equiv\langle x,y,z| x^4=y^2=z^4=1,  yxy=x^{-1}, zx=xz, zy=yz,  x^2=z^2=(-1)^{\rF}\rangle $. Here $a^X$ ($X=\rC$, $\rR_j$, and $\rT$) are the three generators of $\H^1(\B(\Z_2^{\rC}\times\Z_2^{\rR_j}\times\Z_2^{\rT}),\Z_2^{\rF})=\Z_2^3$ respectively. }
\label{tab:CPRT-group-dd-2}
\end{center}
\end{table}

Let $G_{{\rm D},i}(d\mod8)$ $(d>1)$ denote the $i$-th possibility of the C-R$_j$-T fractionalization with the canonical CRT for a Dirac fermion. We will use this notation in \Sec{sec:domain-wall}. We find that 
 
\begin{itemize}
    \item 
$G_{{\rm D},i}(0\mod8)= G_{{\rm D},i}(7\mod8)$ for $i=1,2$,
\item 
$G_{{\rm D},i}(3\mod8)= G_{{\rm D},i}(4\mod8)$ for $i=1,2$,
\item 
$G_{{\rm D},1}(1\mod8)= G_{{\rm D},2}(2\mod8)$, $G_{{\rm D},1}(2\mod8)= G_{{\rm D},2}(1\mod8)$,
\item 
$G_{{\rm D},1}(5\mod8)= G_{{\rm D},2}(6\mod8)$, $G_{{\rm D},1}(6\mod8)= G_{{\rm D},2}(5\mod8)$.
    
\end{itemize}

\section{C-R-T fractionalization for Majorana, Weyl, and Majorana-Weyl fermions}\label{sec:Majorana}
In this section, we determine the R-T fractionalization for
Majorana fermions in any spacetime dimension $d = 0, 1, 2, 3, 4 \mod 8$. We also determine the C or T fractionalization
for Weyl fermions in any even spacetime dimension.

Majorana fermions are defined as Dirac fermions with the trivial C symmetry. Majorana fermions exist only for $d=0,1,2,3,4\mod 8$ \cite{Polchinski:1998rr,Stone:2020vva2009.00518}. 
For $d=0,1\mod8$, there exist only massless Majorana fermions \cite{Stone:2020vva2009.00518}. First, we reproduce these results in terms of the chiral representation (see \App{app:Weyl}). 

To determine the R$_j$-T fractionalization for Majorana fermions, we then construct imaginary Gamma matrices for the massive case (for $d=2,3,4\mod8$) and real Gamma matrices for the massless case (for $d=0,1,2\mod8$) (see \App{app:Majorana}), then we can choose the matrix of C to be the identity matrix. Since the matrices of R$_j$ and T for a Majorana fermion have to be real, we can determine the matrices of R$_j$ and T for a Majorana fermion. Moreover, we can determine the squares of R$_j$, T, and R$_j$T, hence we obtain a group presentation of the R$_j$-T fractionalization.

For Majorana fermions, the C symmetry is trivial. 
First, we consider the massless Majorana fermions.
The massless Dirac equation is
\bea
\ii\Gamma^{\mu}\partial_{\mu}\psi=0.
\eea
The charge-conjugated massless Dirac equation is 
\bea
\ii\Gamma^{\mu}\partial_{\mu}\psi^{\rC}=0.
\eea
Since $\psi^{\rC}=M_{\rC}\psi^*$, we get 
\bea\label{eq:cc1}
M_{\rC}\Gamma^{\mu}=-\Gamma^{\mu*}M_{\rC}
\eea
or
\bea\label{eq:cc2}
M_{\rC}\Gamma^{\mu}=\Gamma^{\mu*}M_{\rC}
\eea
for $\mu=0,1,\dots,d-1$.
Only \eqref{eq:cc1} is valid if we consider the massive Majorana fermions.

First, we choose the Gamma matrices as \eqref{eq:Gamma}.
For $d=2k+2$, we get $M_{\rC}=\Gamma^2\Gamma^4\cdots\Gamma^{2k}$ (it satisfies \eqref{eq:cc1} for odd $k$ and \eqref{eq:cc2} for even $k$) or
$M_{\rC}=\Gamma^0\Gamma^1\Gamma^3\cdots\Gamma^{2k+1}$ (it satisfies \eqref{eq:cc1} for even $k$ and \eqref{eq:cc2} for odd $k$).
For $d=2k+3$, since $\Gamma^{2k+2}=\ii^{k+1}\Gamma^0\Gamma^1\cdots\Gamma^{2k+1}$, if \eqref{eq:cc1} holds for $\mu=0,1,\dots,2k+1$, then \eqref{eq:cc1} holds for $\mu=2k+2$ only if $k$ is even. If \eqref{eq:cc2} holds for $\mu=0,1,\dots,2k+1$, then \eqref{eq:cc2} holds for $\mu=2k+2$ only if $k$ is odd.
We get $M_{\rC}=\Gamma^0\Gamma^1\Gamma^3\cdots\Gamma^{2k+1}$ for $d=2k+3$ and all $k$ (it satisfies \eqref{eq:cc1} for even $k$ and \eqref{eq:cc2} for odd $k$).
The Majorana condition is $\psi^{\rC}=\psi=M_{\rC}\psi^*=M_{\rC}M_{\rC}^*\psi$, hence $M_{\rC}M_{\rC}^*=1$.

For $d=2k+2$, if $M_{\rC}=\Gamma^2\Gamma^4\cdots\Gamma^{2k}$, then 
$M_{\rC}M_{\rC}^*=1$ implies that $(-1)^{\frac{k(k-1)}{2}}=1$ and $k=0,1\mod4$ where $M_{\rC}$ satisfies \eqref{eq:cc2} for $k=0\mod4$ and \eqref{eq:cc1} for $k=1\mod4$. If $M_{\rC}=\Gamma^0\Gamma^1\Gamma^3\cdots\Gamma^{2k+1}$, then 
$M_{\rC}M_{\rC}^*=1$ implies that $(-1)^{\frac{k(k+1)}{2}}=1$ and $k=0,3\mod4$ where $M_{\rC}$ satisfies \eqref{eq:cc1} for $k=0\mod4$ and \eqref{eq:cc2} for $k=3\mod4$.
For $d=2k+3$, $M_{\rC}=\Gamma^0\Gamma^1\Gamma^3\cdots\Gamma^{2k+1}$.
$M_{\rC}M_{\rC}^*=1$ implies that $(-1)^{\frac{k(k+1)}{2}}=1$ and $k=0,3\mod4$ where $M_{\rC}$ satisfies \eqref{eq:cc1} for $k=0\mod4$ and \eqref{eq:cc2} for $k=3\mod4$.

In summary, the Majorana condition implies that $d=0,1,2,3,4\mod8$.
This agrees with Table B.1 of \cite{Polchinski:1998rr}.
Here $d=0\mod8$ requires \eqref{eq:cc2}, $d=1\mod8$ requires \eqref{eq:cc2}, $d=2\mod8$ requires either \eqref{eq:cc1} or \eqref{eq:cc2}, $d=3\mod8$ requires \eqref{eq:cc1}, and $d=4\mod8$ requires \eqref{eq:cc1}. Therefore, there exist only massless Majorana fermions for $d=0,1\mod8$. This also agrees with the results in \cite{Stone:2020vva2009.00518}.

We can construct imaginary Gamma matrices for $d=2,3,4\mod8$, namely $\Gamma^{\mu*}=-\Gamma^{\mu}$. Then we can choose the matrix of the charge conjugation to be the identity matrix since $M_{\rC}\Gamma^{\mu*}M_{\rC}^{-1}=-\Gamma^{\mu}$. 
We can construct real Gamma matrices for $d=0,1,2\mod8$, namely $\Gamma^{\mu*}=\Gamma^{\mu}$. In the massless case, we can choose the matrix of the charge conjugation to be the identity matrix since $M_{\rC}\Gamma^{\mu*}M_{\rC}^{-1}=\Gamma^{\mu}$.

Therefore, for $d=0,1,2,3,4\mod8$, the Majorana condition becomes $\psi^*=\psi$.
In the massive case, the matrices of R$_j$ and T should satisfy \eqref{eq:CPRT-condition}.
If the Gamma matrices are imaginary ($d=2,3,4\mod8$), then \eqref{eq:CPRT-condition} implies that 
\bea
M_{\rT}\Gamma^0M_{\rT}^{-1}&=&-\Gamma^0,\nn\\
M_{\rT}\Gamma^jM_{\rT}^{-1}&=&\Gamma^j,\;\;\;j=1,\dots,d-1.
\eea
If the Gamma matrices are real ($d=0,1,2\mod8$), then \eqref{eq:CPRT-condition} implies that 
\bea
M_{\rT}\Gamma^0M_{\rT}^{-1}&=&\Gamma^0,\nn\\
M_{\rT}\Gamma^jM_{\rT}^{-1}&=&-\Gamma^j,\;\;\;j=1,\dots,d-1.
\eea

For Majorana fermions, the matrices of R$_j$ and T should be real. For $d=2k+2=4\mod8$, we can choose imaginary Gamma matrices (see \App{app:Majorana}), and we have 
$M_{\rT}=\ii\Gamma^1\Gamma^2\cdots\Gamma^{2k+1}$
and 
$M_{\rR_j}=\ii\Gamma^0\Gamma^1\cdots\hat{\Gamma}^j\cdots\Gamma^{2k+1}$.
In this case, 
\bea
\rT^2\psi\rT^{-2}&=&M_{\rT}^2\psi\nn\\
&=&(\ii\Gamma^1\Gamma^2\cdots\Gamma^{2k+1})^2\psi\nn\\
&=&(-1)^k\psi=-\psi,
\eea
\bea
\rR_j^2\psi\rR_j^{-2}&=&M_{\rR_j}^2\psi\nn\\
&=&(\ii\Gamma^0\Gamma^1\cdots\hat{\Gamma}^j\cdots\Gamma^{2k+1})^2\psi\nn\\
&=&(-1)^{k+1}\psi=\psi,
\eea
and
\bea
(\rR_j\rT)^2\psi(\rR_j\rT)^{-2}&=&(M_{\rT}M_{\rR_j})^2\psi\nn\\
&=&(\ii\Gamma^1\Gamma^2\cdots\Gamma^{2k+1}\ii\Gamma^0\Gamma^1\cdots\hat{\Gamma}^j\cdots\Gamma^{2k+1})^2\psi\nn\\
&=&\psi.
\eea

For $d=2k+3=3\mod8$, we can choose imaginary Gamma matrices (see \App{app:Majorana}), and \eqref{eq:CPRT-condition} has no solutions for the matrices of R$_j$ and T. In the massless case, we have 
$M_{\rT}=\ii\Gamma^0$
and 
$M_{\rR_j}=\ii\Gamma^j$.
In this case, 
\bea
\rT^2\psi\rT^{-2}&=&M_{\rT}^2\psi\nn\\
&=&(\ii\Gamma^0)^2\psi\nn\\
&=&-\psi,
\eea
\bea
\rR_j^2\psi\rR_j^{-2}&=&M_{\rR_j}^2\psi\nn\\
&=&(\ii\Gamma^j)^2\psi\nn\\
&=&\psi,
\eea
and
\bea
(\rR_j\rT)^2\psi(\rR_j\rT)^{-2}&=&(M_{\rT}M_{\rR_j})^2\psi\nn\\
&=&(\ii\Gamma^0\ii\Gamma^j)^2\psi\nn\\
&=&\psi.
\eea

For $d=2k+2=2\mod8$, the Majorana fermions can be massive or massless and we can choose imaginary or real Gamma matrices (see \App{app:Majorana}). We have 
$M_{\rT}=\ii\Gamma^1\Gamma^2\cdots\Gamma^{2k+1}$
and 
$M_{\rR_j}=\ii\Gamma^0\Gamma^1\cdots\hat{\Gamma}^j\cdots\Gamma^{2k+1}$
(if the Gamma matrices are imaginary) or 
$M_{\rT}=\Gamma^0$
and 
$M_{\rR_j}=\Gamma^0\Gamma^1\cdots\hat{\Gamma}^j\cdots\Gamma^{2k+1}$
(if the Gamma matrices are real). In the first case, 
\bea
\rT^2\psi\rT^{-2}&=&M_{\rT}^2\psi\nn\\
&=&(\ii\Gamma^1\Gamma^2\cdots\Gamma^{2k+1})^2\psi\nn\\
&=&(-1)^k\psi=\psi,
\eea
\bea
\rR_j^2\psi\rR_j^{-2}&=&M_{\rR_j}^2\psi\nn\\
&=&(\ii\Gamma^0\Gamma^1\cdots\hat{\Gamma}^j\cdots\Gamma^{2k+1})^2\psi\nn\\
&=&(-1)^{k+1}\psi=-\psi,
\eea
and
\bea
(\rR_j\rT)^2\psi(\rR_j\rT)^{-2}&=&(M_{\rT}M_{\rR_j})^2\psi\nn\\
&=&(\ii\Gamma^1\Gamma^2\cdots\Gamma^{2k+1}\ii\Gamma^0\Gamma^1\cdots\hat{\Gamma}^j\cdots\Gamma^{2k+1})^2\psi\nn\\
&=&\psi.
\eea
In the second case,
\bea
\rT^2\psi\rT^{-2}&=&M_{\rT}^2\psi\nn\\
&=&(\Gamma^0)^2\psi\nn\\
&=&\psi,
\eea
\bea
\rR_j^2\psi\rR_j^{-2}&=&M_{\rR_j}^2\psi\nn\\
&=&(\Gamma^0\Gamma^1\cdots\hat{\Gamma}^j\cdots\Gamma^{2k+1})^2\psi\nn\\
&=&(-1)^k\psi=\psi,
\eea
and
\bea
(\rR_j\rT)^2\psi(\rR_j\rT)^{-2}&=&(M_{\rT}M_{\rR_j})^2\psi\nn\\
&=&(\Gamma^0\Gamma^0\Gamma^1\cdots\hat{\Gamma}^j\cdots\Gamma^{2k+1})^2\psi\nn\\
&=&(-1)^k\psi=\psi.
\eea

For $d=2k+3=1\mod8$, the Majorana fermions have to be massless and we can choose real Gamma matrices (see \App{app:Majorana}). We have 
$M_{\rT}=\Gamma^0$
and \eqref{eq:CPRT-condition} has no solutions for the matrix of R$_j$. In the massless case, we have 
$M_{\rR_j}=\Gamma^j$.
In this case,
\bea
\rT^2\psi\rT^{-2}&=&M_{\rT}^2\psi\nn\\
&=&(\Gamma^0)^2\psi\nn\\
&=&\psi,
\eea
\bea
\rR_j^2\psi\rR_j^{-2}&=&M_{\rR_j}^2\psi\nn\\
&=&(\Gamma^j)^2\psi\nn\\
&=&-\psi,
\eea
and
\bea
(\rR_j\rT)^2\psi(\rR_j\rT)^{-2}&=&(M_{\rT}M_{\rR_j})^2\psi\nn\\
&=&(\Gamma^0\Gamma^j)^2\psi\nn\\
&=&\psi.
\eea

For $d=2k+2=0\mod8$, the Majorana fermions have to be massless and we can choose real Gamma matrices (see \App{app:Majorana}). We have 
$M_{\rT}=\Gamma^0$
and 
$M_{\rR_j}=\Gamma^0\Gamma^1\cdots\hat{\Gamma}^j\cdots\Gamma^{2k+1}$.
In this case, 
\bea
\rT^2\psi\rT^{-2}&=&M_{\rT}^2\psi\nn\\
&=&(\Gamma^0)^2\psi\nn\\
&=&\psi,
\eea
\bea
\rR_j^2\psi\rR_j^{-2}&=&M_{\rR_j}^2\psi\nn\\
&=&(\Gamma^0\Gamma^1\cdots\hat{\Gamma}^j\cdots\Gamma^{2k+1})^2\psi\nn\\
&=&(-1)^k\psi=-\psi,
\eea
and
\bea
(\rR_j\rT)^2\psi(\rR_j\rT)^{-2}&=&(M_{\rT}M_{\rR_j})^2\psi\nn\\
&=&(\Gamma^0\Gamma^0\Gamma^1\cdots\hat{\Gamma}^j\cdots\Gamma^{2k+1})^2\psi\nn\\
&=&(-1)^k\psi=-\psi.
\eea

Note that the case $d=1$ is special. In this case, there is only the T symmetry and a single Majorana fermion can not have a mass. In the massless case, $M_{\rT}=1$ and $\rT^2=1$.

 We summarize the $\rR_j$, $\rT$ transformations on a Majorana fermion in $d$d for $d=0,1,2,3,4\mod8$ in \Table{tab:CPRT-trans-dd-Majorana}.
\begin{table}[htp]
\begin{center}
\begin{tabular}{c|c | c}
\hline
Dimension &  $M_{\rR_j}\psi(t,x_1,\dots,x_{j-1},-x_j,x_{j+1},\dots,x_{d-1})$ & $M_{\rT}\psi(-t,x)$ \\
\hline
$d=0\mod8$  & $\Gamma^0\Gamma^1\cdots\hat{\Gamma}^j\cdots\Gamma^{d-1}\psi(t,x_1,\dots,x_{j-1},-x_j,x_{j+1},\dots,x_{d-1})$ & $\Gamma^0\psi(-t,x)$ \\
$d=1\mod8$  & $\Gamma^j\psi(t,x_1,\dots,x_{j-1},-x_j,x_{j+1},\dots,x_{d-1})$ & $\Gamma^0\psi(-t,x)$\\
$d=2\mod8$ (massive)  & $\ii\Gamma^0\Gamma^1\cdots\hat{\Gamma}^j\cdots\Gamma^{d-1}\psi(t,x_1,\dots,x_{j-1},-x_j,x_{j+1},\dots,x_{d-1})$ & $\ii\Gamma^1\Gamma^2\cdots\Gamma^{d-1}\psi(-t,x)$ \\
$d=2\mod8$ (massless)  & $\Gamma^0\Gamma^1\cdots\hat{\Gamma}^j\cdots\Gamma^{d-1}\psi(t,x_1,\dots,x_{j-1},-x_j,x_{j+1},\dots,x_{d-1})$ & $\Gamma^0\psi(-t,x)$ \\
$d=3\mod8$  & $\ii\Gamma^j\psi(t,x_1,\dots,x_{j-1},-x_j,x_{j+1},\dots,x_{d-1})$ & $\ii\Gamma^0\psi(-t,x)$ \\
$d=4\mod8$  & $\ii\Gamma^0\Gamma^1\cdots\hat{\Gamma}^j\cdots\Gamma^{d-1}\psi(t,x_1,\dots,x_{j-1},-x_j,x_{j+1},\dots,x_{d-1})$ & $\ii\Gamma^1\Gamma^2\cdots\Gamma^{d-1}\psi(-t,x)$ \\
\hline
\end{tabular}
\caption{The $\rR_j$, $\rT$ transformations on a Majorana fermion in $d$d for $d=0,1,2,3,4\mod8$. Here $\hat{\Gamma}^j$ means $\Gamma^j$ is omitted in the product. See \App{app:Majorana} for the explicit form of the Gamma matrices.  }
\label{tab:CPRT-trans-dd-Majorana}
\end{center}
\end{table}

Therefore, a single Majorana fermion can not have a mass for $d=0,1\mod8$. The Majorana mass term breaks the $\rR_j$ and T symmetry for $d=3\mod8$.
See \Table{table:CRT-break-Majorana}.

We summarize the R$_j$-T fractionalization for a Majorana fermion in \Table{tab:CPRT-group-dd-Majorana}. Note that the case $d=1$ is special because there is no R$_j$ for $d=1$.
\begin{table}[htp]
\begin{center}
\begin{tabular}{c|c|c}
\hline
Dimension & Group presentation & Group\\
\hline
$d\equiv 0\mod 8$ & $\rT^2=1$, $\rR_j^2=(\rR_j\rT)^2=(-1)^{\rF}$ & $\Z_4^{\rR_j\rF}\times\Z_2^{\rT}$ \\
\hline
$d=1$ &  $\rT^2=1$ & $\Z_2^{\rT}\times\Z_2^{\rF}$\\
\hline
$d\equiv 1\mod 8$ ($d>1$) &  $\rT^2=(\rR_j\rT)^2=1$, $\rR_j^2=(-1)^{\rF}$ & $\bD_8^{\rF ,\rR_j,\rT}$\\
\hline
$d\equiv 2\mod 8$ &  $\begin{array}{cc}\rR_j^2=(-1)^{\rF}, \rT^2=(\rR_j\rT)^2=1\\\text{or }\rR_j^2=\rT^2=(\rR_j\rT)^2=1
\end{array}$ & $\begin{array}{cc}\bD_8^{\rF,\rR_j ,\rT}\\\text{or }\Z_2^{\rR_j}\times\Z_2^{\rT}\times\Z_2^{\rF}
\end{array}$\\
\hline
$d\equiv 3\mod 8$ & $\rT^2=(-1)^{\rF}$, $\rR_j^2=(\rR_j\rT)^2=1$ &  $\bD_8^{\rF,\rT, \rR_j}$\\
\hline
$d\equiv 4\mod 8$ & $\rT^2=(-1)^{\rF}$, $\rR_j^2=(\rR_j\rT)^2=1$ &  $\bD_8^{\rF,\rT ,\rR_j}$\\
\hline
\end{tabular}
\caption{The structure of the R$_j$-T fractionalization for a Majorana fermion in $d$d for all $d=0,1,2,3,4\mod8$. 
Here $\bD_8^{\rF,x ,y}\equiv\langle x,y |x^4=y^2=1, yxy=x^{-1}, x^2=(-1)^{\rF}\rangle$.
For $d=2\mod8$, there are two choices because Majorana fermions can be massive or massless in $d$d. The first choice is for the massive case, and the second choice is for the massless case.}
\label{tab:CPRT-group-dd-Majorana}
\end{center}
\end{table} 

Let $G_{\rm M}(d\mod 8)(d>1)$ denote the R$_j$-T fractionalization for a Majorana fermion and $G_{{\rm M},i}(2\mod 8)$ denote the $i$-th choice of the R$_j$-T fractionalization for a Majorana fermion. We will use this notation in \Sec{sec:domain-wall}. We find that 
\begin{itemize}
    \item 
    $G_{\rm M,1}(2\mod8)=G_{\rm M}(1\mod8)$.
 \item 
$G_{\rm M}(3\mod 8)=G_{\rm M}(4\mod8)$.
\item 
$G_{\rm M}(d\mod8)$ is a subgroup of $G_{{\rm D},1}(d\mod8)$ for $d=0,1\mod8$.
\item 
$G_{\rm M}(d\mod8)$ is a subgroup of $G_{{\rm D},2}(d\mod8)$ for $d=3,4\mod8$.
\item 
$G_{\rm M,1}(2\mod8)$ is a subgroup of $G_{{\rm D},2}(2\mod8)$ and
$G_{\rm M,2}(2\mod8)$ is a subgroup of $G_{{\rm D},1}(2\mod8)$.

\end{itemize}

For Weyl fermions, the spacetime dimension $d$ has to be even because of Definition \ref{def-Weyl}. 
For $d=2k+2$, the Dirac fermion $\psi$ decomposes as $\psi=\left(\begin{array}{cc}\psi_L\\\psi_R\end{array}\right)$ where $\psi_L$ and $\psi_R$ are left and right Weyl fermions. Since $M_{\rR_j}=\Gamma^0\Gamma^1\cdots\hat{\Gamma}^j\cdots\Gamma^{2k+1}$ is not block diagonal in the chiral representation (see \App{app:Weyl}), 
\bea
\psi_L^{\rR_j}= (-1)^k\gamma^1\cdots\hat{\gamma}^j\cdots\gamma^{2k+1}\psi_R ( t,x_1,\dots, x_{j-1}, -x_j,x_{j+1},\dots,x_{d-1})
\eea
(see \eqref{eq:Gamma}),
the left Weyl fermion $\psi_L$ does not have the $\rR_j$ symmetry. 
$\psi_L$ has the C or T symmetry if and only if $M_{\rC}$ (or $M_{\rT}$) is block diagonal in the chiral representation (see \App{app:Weyl}). Note that the phase factors of C and T do not affect their orders, hence we ignore the phase factors of C and T in the following discussion.

For odd $k$, $d=0\mod4$, $M_{\rC}=\Gamma^2\Gamma^4\cdots\Gamma^{2k}$ is not block diagonal in the chiral representation (see \App{app:Weyl}), while $M_{\rT}=\Gamma^1\Gamma^3\cdots\Gamma^{2k+1}$ is block diagonal in the chiral representation (see \App{app:Weyl}), hence $\psi_L$ has the T symmetry but not the C symmetry. In this case, 
\bea
\psi_L^{\rT}=(-1)^{\frac{k+1}{2}}\gamma^1\gamma^3\cdots\gamma^{2k+1}\psi_L(-t,x)
\eea
(see \eqref{eq:Gamma}) and $\rT^2=(-1)^{\frac{k(k+1)}{2}}$.

For even $k$, $d=2\mod4$,
$M_{\rC}=\Gamma^0\Gamma^1\Gamma^3\cdots\Gamma^{2k+1}$ is block diagonal in the chiral representation (see \App{app:Weyl}), while $M_{\rT}=\Gamma^0\Gamma^2\Gamma^4\cdots\Gamma^{2k}$ is not block diagonal in the chiral representation (see \App{app:Weyl}), hence $\psi_L$ has the C symmetry but not the T symmetry. In this case, 
\bea
\psi_L^{\rC}=(-1)^{\frac{k}{2}+1}\gamma^1\gamma^3\cdots\gamma^{2k+1}\psi_L^*
\eea
(see \eqref{eq:Gamma}) and $\rC^2=(-1)^{\frac{k(k+1)}{2}}$.

We summarize the C or T fractionalization for a Weyl fermion in \Table{table:CT-group-dd-Weyl}.
\begin{table}[!h]
\hspace{0.mm}
    \begin{tabular}{|c| c| c|}
    \hline
Dimension    &  Group presentation &  Group \\
     \hline
$d=0\mod8$ &  $\rT^2=1$ & $\Z_2^{\rT}\times\Z_2^{\rF}$\\
\hline
$d=2\mod8$ &  $\rC^2=1$ & $\Z_2^{\rC}\times\Z_2^{\rF}$\\
     \hline
$d=4\mod8$ & $\rT^2=(-1)^{\rF}$ &$\Z_4^{\rT\rF}$\\
     \hline
$d=6\mod8$ & $\rC^2=(-1)^{\rF}$ &$\Z_4^{\rC\rF}$\\
     \hline
     \end{tabular}
\caption{The C or T fractionalization for a Weyl fermion in $d$d for all even $d$. }
\label{table:CT-group-dd-Weyl}
\end{table}

Let $G_{\rm W}(d\mod8)$ denote the C or T fractionalization for a Weyl fermion. We will use this notation in \Sec{sec:domain-wall}.
We find that the $G_{\rm W}(d\mod8)$ is a subgroup of $G_{{\rm D},i}(d\mod8)$ ($i=1,2$) for all even $d$.

For Majorana-Weyl fermions, since the Majorana-Weyl fermion is Majorana, it has the trivial C symmetry and $d=0,1,2,3,4\mod8$. Since the Majorana-Weyl fermion is also Weyl, $d$ is even. By the above discussion, Weyl fermion does not have the C symmetry unless $d = 2 \mod 4$. But Majorana-Weyl fermion has a trivial C symmetry, thus there is only a trivial C symmetry but no R$_j$ or T symmetries, and $d=2\mod8$.
Let $G_{\rm MW}(2\mod8)=\Z_2^{\rF}$ denote the trivial C fractionalization for a Majorana-Weyl fermion. We will use this notation in \Sec{sec:domain-wall}.

We summarize the symmetries of Dirac, Majorana, Weyl, and Majorana-Weyl fermions in even spacetime dimension $d$ in \Table{table:CPT}.

\begin{table}[!h]
\hspace{0.mm}
    \begin{tabular}{|c| ccc| c|}
    \hline
    &  C &  $\rR_j$ &  T & C$\rR_j$T \\
     \hline
Dirac & $\checkmark$  &   $\checkmark$ & $\checkmark$ & $\checkmark$ \\
\hline
Majorana & trivial & $\checkmark$  & $\checkmark$ & $\checkmark$  \\
     \hline
Weyl ($d=2 \mod 4$) & $\checkmark$  & \xmark  & \xmark & $\checkmark$ \\
     \hline
Weyl ($d=0 \mod 4$) &  \xmark  &  \xmark  & $\checkmark$ & $\checkmark$ \\
\hline
Majorana-Weyl  ($d=2 \mod 8$)  & trivial &  \xmark  & \xmark & $\checkmark$   \\
     \hline
     \end{tabular}
\caption{The symmetries of Dirac, Majorana, Weyl, and Majorana-Weyl fermions in even spacetime dimension $d$. Here $\checkmark$ means the symmetry is preserved.
Here trivial means that the symmetry acts trivially. 
The C acts on Majorana fermion $\psi$ as $\psi^{\rm C} =\psi$. 
Here \xmark \; means the symmetry is violated.}
\label{table:CPT}
\end{table}

\section{Mass term}\label{sec:mass}
In this section, we determine the maximal number of linearly independent Dirac and Majorana mass terms and discuss how the Dirac and Majorana mass terms break the symmetries C, R, or T.

\subsection{Dirac mass term}

In this subsection, we will determine the maximal number of linearly independent Dirac mass terms. We will use Definition \ref{def-Dirac} for Dirac fermions, namely a Dirac fermion is a map from the Minkowski spacetime $\R^{d-1,1}$ to the complex spin representation $S$ of the spin group $\Spin_{d-1,1}$. Since the complex representations of $\Spin_{d-1,1}$ are the same as the complex representations of the spin group $\Spin(d)$, we will regard $S$ as the complex spin representation of $\Spin(d)$. The main reference for this subsection is \cite{Brocker}.

\begin{proposition}(\cite[Proposition VI.6.1]{Brocker})
The group $\Spin(2n)$ has as fundamental system of weights
\bea
\lambda_{\nu}&=&e_1+\cdots+e_{\nu},\;\;\;1\le \nu \le n-2,\nn\\
\partial_n^{\pm}&=&\frac{1}{2}(e_1+\cdots +e_{n-1}\pm e_n),
\eea
where $e_{\nu}$ is the $\nu$-th standard unit vector of $\R^n$.
The group $\Spin(2n+1)$ has as fundamental system of weights
\bea
&&\lambda_{\nu},1\le \nu\le n-1,\nn\\
&&\partial_n=\frac{1}{2}(e_1+\cdots+e_n).
\eea
\end{proposition}
These fundamental weights are dominant for fundamental representations
\bea
\wedge^1,\wedge^2,\dots, \wedge^{n-2},S_+^n,S_-^n\;\;\;\text{of }\Spin(2n),\;\;\;\text{and}\nn\\
\wedge^1,\wedge^2,\dots, \wedge^{n-1},S^n\;\;\;\text{of }\Spin(2n+1).
\eea
The representations $\wedge^{\nu}$ arise from the representations of $\SO(d)$ with identical notation via the projection $\rho:\Spin(d)\to\SO(d)$. In other words, $\wedge^{\nu}\C^d$ is the $\nu$-th exterior power of the $d$-dimensional complex vector space $\C^d$, we also shorthand the notation $\wedge^{\nu}\C^d$  as  $\wedge^{\nu}$. The representations $S$, $S_+$, and $S_-$ are the half-spin representations of $\Spin(d)$.

The irreducible representations of $\Spin(2n+1)$ with dominant weights $2\partial$ is denoted by $\wedge^n$.
The irreducible representations of $\Spin(2n)$ with dominant weights $2\partial^+$, $2\partial^-$, and $\partial^++\partial^-$ are denoted by $\wedge_+^n$, $\wedge_-^n$, and $\wedge^{n-1}$.

The Dirac mass term can be defined as a homomorphism of representations from $S\otimes S$ to $1^{\C}$ (the trivial complex representation) where $S$ is the complex spin representation of $\Spin(d-1,1)$ and $d$ is the spacetime dimension. If $d$ is odd, $\Hom(S\otimes S,1^{\C})=\C$ because $S\otimes S$ contains a trivial summand, so there is only one linearly independent Dirac mass term. If $d$ is even, $\Hom(S\otimes S,1^{\C})=\C^2$ because $S\otimes S$ contains two trivial summands, so there are two linearly independent Dirac mass terms. 

The following results are consequences of \cite[Theorem VI.6.2]{Brocker}.
If $S$ is a complex spin representation of $\Spin(d)$ and $d=2n+1$ is odd, then
\bea
S\otimes S=1\oplus\wedge^1\oplus\cdots\oplus\wedge^{n-1}\oplus\wedge^n.
\eea
Since $S\otimes S$ contains only one trivial summand, there is only one linearly independent Dirac mass term.    

If $S$ is a complex spin representation of $\Spin(d)$ and $d=2n$ is even, then $S=S_+\oplus S_-$ and
\bea\label{eq:tensor-square-even}
S_+\otimes S_+&=&\wedge^n_+\oplus\wedge^{n-2}\oplus \wedge^{n-4}\oplus\cdots,\nn\\
S_+\otimes S_-&=&\wedge^{n-1}\oplus\wedge^{n-3}\oplus \wedge^{n-5}\oplus\cdots,\nn\\
S_-\otimes S_-&=&\wedge^n_-\oplus\wedge^{n-2}\oplus \wedge^{n-4}\oplus\cdots.
\eea 
The sums end in $\wedge^4\oplus\wedge^2\oplus1$ or $\wedge^3\oplus\wedge^1$.
If $n$ is even, each of $S_+\otimes S_+$ and $S_-\otimes S_-$ contains one trivial summand. If $n$ is odd, each of $S_+\otimes S_-$ and $S_-\otimes S_+$ contains one trivial summand.
Since $S\otimes S$ always contains two trivial summands, there are two linearly independent Dirac mass terms.

We can explicitly construct the Dirac mass terms.
For odd $d$, the only one linearly independent Dirac mass term is $\psi^{\dagger}\Gamma^0\psi$.
For even $d$, the two linearly independent Dirac mass terms are $\psi^{\dagger}\Gamma^0\psi$ and $\psi^{\dagger}\diag(-I,I)\psi$ where $\diag(-I,I)=\ii^{\frac{d}{2}-1}\Gamma^0\Gamma^1\cdots\Gamma^{d-1}$.

In this article, we will only consider the conventional Dirac mass term $\psi^{\dagger}\Gamma^0\psi$.

\subsection{Majorana mass term}
In this subsection, we determine the maximal number of linearly independent Majorana mass terms.

The Majorana mass term can be defined as a homomorphism of representations from $S_{\rC}\otimes S_{\rC}$ to $1^{\R}$ (the trivial real representation) where $S_{\rC}$ is the charge-conjugation-invariant part of the complex spin representation $S$ of $\Spin(d-1,1)$ and $d$ is the spacetime dimension.

By Definition \ref{def-Majorana}, a Majorana fermion is a single Dirac fermion with the trivial charge conjugation, namely a Dirac fermion $\psi$ satisfying the Majorana condition $\psi=M_{\rC}\psi^*$. If a Dirac mass term is real under Majorana condition, then it is a Majorana mass term.

If $M_{\rC}^{\mathsf{T}}=\eta M_{\rC}$, then we have
\bea
(\psi^{\dagger}\Gamma^0\psi)^*&=&-\psi^{\mathsf{T}}\Gamma^0\psi^*\nn\\
&=&-\psi^{\dagger}M_{\rC}^{\mathsf{T}}\Gamma^0M_{\rC}^{-1}\psi\nn\\
&=&\eta\psi^{\dagger}\Gamma^0\psi.
\eea
In the first equality, we have used the fact that $(\psi\psi')^*=-\psi^*\psi'^*$ for fermions $\psi$ and $\psi'$.
In the second equality, we have used the Majorana condition $\psi=M_{\rC}\psi^*$.
In the third equality, we have used $M_{\rC}^{\mathsf{T}}=\eta M_{\rC}$ and $M_{\rC}\Gamma^0M_{\rC}^{-1}=-\Gamma^0$. Therefore, the Dirac mass term $\psi^{\dagger}\Gamma^0\psi$ is real if and only if $M_{\rC}^{\mathsf{T}}=M_{\rC}$.

If $d$ is even, there is another Dirac mass term $\psi^{\dagger}\diag(-I,I)\psi$.
We have
\bea
(\psi^{\dagger}\diag(-I,I)\psi)^*&=&(\psi^{\dagger}\ii^{\frac{d}{2}-1}\Gamma^0\Gamma^1\cdots\Gamma^{d-1}\psi)^*\nn\\
&=&-\psi^{\mathsf{T}}(-\ii)^{\frac{d}{2}-1}\Gamma^{0*}\Gamma^{1*}\cdot\Gamma^{d-1*}\psi^*\nn\\
&=&-\psi^{\dagger}M_{\rC}^{\mathsf{T}}(-\ii)^{\frac{d}{2}-1}\Gamma^{0*}\Gamma^{1*}\cdot\Gamma^{d-1*}M_{\rC}^{-1}\psi\nn\\
&=&-\psi^{\dagger}\eta M_{\rC}(-\ii)^{\frac{d}{2}-1}\Gamma^{0*}\Gamma^{1*}\cdot\Gamma^{d-1*}M_{\rC}^{-1}\psi\nn\\
&=&-(-\ii)^{\frac{d}{2}-1}\eta \psi^{\dagger}\Gamma^0\Gamma^1\cdots\Gamma^{d-1}\psi\nn\\
&=&(-1)^{\frac{d}{2}}\eta \psi^{\dagger}\diag(-I,I)\psi
\eea
In the first and last equality, we have used $\diag(-I,I)=\ii^{\frac{d}{2}-1}\Gamma^0\Gamma^1\cdots\Gamma^{d-1}$.
In the second equality, we have used the fact that $(\psi\psi')^*=-\psi^*\psi'^*$ for fermions $\psi$ and $\psi'$.
In the third equality, we have used the Majorana condition $\psi=M_{\rC}\psi^*$.
In the forth equality, we have used $M_{\rC}^{\mathsf{T}}=\eta M_{\rC}$.
In the fifth equality, we have used 
$M_{\rC}\Gamma^{\mu*}M_{\rC}^{-1}=-\Gamma^{\mu}$.
Therefore, the Dirac mass term $\psi^{\dagger}\diag(-I,I)\psi$ is real if and only if $M_{\rC}^{\mathsf{T}}=(-1)^{\frac{d}{2}}M_{\rC}$.

For $d=0,1\mod 8$, by the discussion in \Sec{sec:Majorana} (see also \cite{Stone:2020vva2009.00518}), there are no Majorana mass terms.

For $d=2\mod8$ ($d=2k+2$, $k=0\mod4$), by \Table{tab:CPRT-trans-dd}, $M_{\rC}=\Gamma^0\Gamma^1\Gamma^3\cdots\Gamma^{2k+1}$ (we omit the phase factor because it doesn't matter here). We have
\bea
M_{\rC}^{\mathsf{T}}&=&(-1)^{k+1}\Gamma^{2k+1}\cdots\Gamma^3\Gamma^1\Gamma^0\nn\\
&=&(-1)^{\frac{k(k+1)}{2}}M_{\rC}\nn\\
&=&M_{\rC}.
\eea
Therefore, there is only one linearly independent Majorana mass term $\psi^{\dagger}\Gamma^0\psi$ for $d=2\mod8$.

For $d=3\mod8$ ($d=2k+3$, $k=0\mod4$), by \Table{tab:CPRT-trans-dd}, $M_{\rC}=\Gamma^0\Gamma^1\Gamma^3\cdots\Gamma^{2k+1}$ (we omit the phase factor because it doesn't matter here). We have
\bea
M_{\rC}^{\mathsf{T}}&=&(-1)^{k+1}\Gamma^{2k+1}\cdots\Gamma^3\Gamma^1\Gamma^0\nn\\
&=&(-1)^{\frac{k(k+1)}{2}}M_{\rC}\nn\\
&=&M_{\rC}.
\eea
Therefore, there is only one linearly independent Majorana mass term $\psi^{\dagger}\Gamma^0\psi$ for $d=3\mod8$.

For $d=4\mod8$ ($d=2k+2$, $k=1\mod4$), by \Table{tab:CPRT-trans-dd}, $M_{\rC}=\Gamma^2\Gamma^4\cdots\Gamma^{2k}$ (we omit the phase factor because it doesn't matter here). We have
\bea
M_{\rC}^{\mathsf{T}}&=&\Gamma^{2k}\cdots\Gamma^4\Gamma^2\nn\\
&=&(-1)^{\frac{k(k-1)}{2}}M_{\rC}\nn\\
&=&M_{\rC}.
\eea
Therefore, there are two linearly independent Majorana mass terms $\psi^{\dagger}\Gamma^0\psi$ and $\psi^{\dagger}\diag(-I,I)\psi$ for $d=4\mod8$.

In this article, we will only consider the conventional Majorana mass term $\psi^{\dagger}\Gamma^0\psi$.

\subsection{How the conventional Dirac and Majorana mass terms break the symmetries C, R, or T}

Based on the discussion in \Sec{sec:Dirac}, we find that the conventional {Dirac} mass term breaks the C symmetry for $d=2k+3$ ($k$ odd), it breaks the T symmetry for $d=2k+3$ ($k$ even), and it breaks the R$_j$ symmetry for $d=2k+3$ and $j=1,2,\dots,2k+2$, see \Table{table:CRT-break-Dirac}.

\begin{table}[!h]
\hspace{0.mm}
    \begin{tabular}{|c| ccc| }
    \hline
  Dirac mass term &  C &   R$_j$ &  T  \\
   \hline
$d=2k+2$ ($k \in \Z_{\geq 0}$) or $d= 0,2 \mod 4$   & $\checkmark$  &  $\checkmark$  &  $\checkmark$\\
     \hline
$d=2k+3$ ($k$ odd) or $d= 1 \mod 4$ & \xmark  &   \xmark & $\checkmark$ \\
\hline
$d=2k+3$ ($k$ even) or $d= 3 \mod 4$ & $\checkmark$ & \xmark  & \xmark  \\
     \hline

     \end{tabular}
\caption{The conventional Dirac mass term (for a single particle field) breaks the symmetries C, R$_j$, and T or not. Here $\checkmark$ means the conventional Dirac mass term preserves the symmetry.
Here \xmark\; means the conventional Dirac mass term breaks the symmetry.
}
\label{table:CRT-break-Dirac}
\end{table}

Based on the discussion in \Sec{sec:Majorana}, we find that a single Majorana fermion can not have a mass for $d=0,1\mod8$. The conventional Majorana mass term breaks the $\rR_j$ and T symmetry for $d=3\mod8$.
See \Table{table:CRT-break-Majorana}.
\begin{table}[!h]
\hspace{0.mm}
    \begin{tabular}{|c| cc| }
    \hline
  Majorana mass term &   R$_j$ &  T  \\
  \hline
$d=0,1\mod8$ (massless)    &    &  \\
   \hline
$d=2,4\mod8$     &  $\checkmark$  &  $\checkmark$\\
     \hline
$d=3\mod8$ & \xmark  &   \xmark  \\
\hline
     \end{tabular}
\caption{The conventional Majorana mass term (for a single particle field)
breaks the symmetries R$_j$ and T or not. Here $\checkmark$ means the conventional Majorana mass term preserves the symmetry.
Here \xmark\; means the conventional Majorana mass term breaks the symmetry. 
\cred{For a single particle Majorana fermion,
the dimension $d=0,1\mod8$ forbids Majorana mass thus Majorana fermion
is massless.
}
}
\label{table:CRT-break-Majorana}
\end{table}

Dirac mass and Majorana mass break the same kind of symmetry in the same pattern. Majorana has trivial C symmetry, but Dirac has nontrivial C symmetry. Dirac mass can break C symmetry in $d= 1 \mod 4$,
thus the corresponding Majorana mass breaks C symmetry which contradicts the premise Majorana condition --- this implies that Majorana fermion must be massless in $d= 1 \mod 4$.
Dirac mass exists in all dimensions.
Majorana fermion exists in $d=0,1,2,3,4 \mod 8$, 
while Majorana mass exists in $d=2,3,4 \mod 8$.

\section{Domain wall dimensional reduction}\label{sec:domain-wall}

In this section, we study the domain wall dimensional reduction of the C-R-T fractionalization and R-T fractionalization. We also study the domain wall dimensional reduction of fermions and arrive at a consistent story.
Our main result is the following theorem.
\begin{theorem}\label{thm-DW}
    If a $d$d fermion $\psi^{d}$ is in the $d$d bulk and a $(d-1)$d fermion $\psi^{d-1}$ is on the $(d-1)$d domain wall, then the C-R-T or R-T fractionalization of $\psi^d$ reduces to the corresponding fractionalization of $\psi^{d-1}$ under the domain wall dimensional reduction. In particular, for even $d$, the C-R-T or R-T fractionalization of $\psi^d$ is isomorphic to that of $\psi^{d-1}$.
\end{theorem}
The proof of Theorem \ref{thm-DW} consists of the following steps:
\begin{itemize}
    \item 
    First, for a given $d$d fermion $\psi^{d}$ in the $d$d bulk, we determine the corresponding $(d-1)$d fermion $\psi^{d-1}$ on the $(d-1)$d domain wall by the degree-of-freedom argument.
    \item 
    Next, for a given $d$d fermion $\psi^{d}$ in the $d$d bulk, we determine how C, R, and T change under the domain wall dimensional reduction.
    \item 
    Finally, we find that the reduced C, R, and T generate the corresponding fractionalization of $\psi^{d-1}$.
\end{itemize} 

The mass profile of a fermion in $d$d is shown in \Fig{fig:mass-profile} where the domain wall that separates the two domains $m>0$ and $m<0$ is $m=0,x_{d-1}=0$.
\begin{figure}[!h]
\begin{center}
\begin{tikzpicture}[scale=0.8]
\draw[->] (-2, 0) -- (2, 0) node[right] {$x_{d-1}$};
  \draw[->] (0, -2) -- (0, 2) node[above] {$m$};
  
  \draw[scale=1, domain=-1.5:1.5, smooth, variable=\y]  plot ({0.15*tan(\y r)}, {\y}) node[right] {$d$d};
\node[below] at (0.8,0) {$(d-1)$d};
\node[left] at (0,0.2) {$m=0,x_{d-1}=0$};

\end{tikzpicture}
\end{center}
\caption{Mass profile.
\cred{The $d$d fermion obeys the Lagrangian $\bar \psi^{d} (\ii \Gamma^\mu \prt_\mu - m(x)) \psi^{d}$,
then there is an effective massless $(d-1)$d domain wall fermion theory $\psi^{d-1}$ at $x_{d-1}=0$, 
with its time and spatial coordinates $(t,x_1,x_2,\dots,x_{d-2})$,
which is obtained by the projection ${\rm P}_{\pm}= \frac{1 \pm \ii \Gamma^{d-1}}{2}$,
$
\psi^{d-1}_{\pm} = {\rm P}_{\pm} \psi^{d}= \frac{1 \pm \ii \Gamma^{d-1}}{2} \psi^{d}
$ 
where ${\rm P}_{\pm}^2={\rm P}_{\pm}$, ${\rm P}_{+}{\rm P}_{-}={\rm P}_{-}{\rm P}_{+}=0$,
and $\ii\Gamma^{d-1}{\rm P}_{\pm}=\pm {\rm P}_{\pm}$.
}For odd $d$, $\Gamma^{d-1}=\ii\diag(-I,I)$ in the chiral representation (see \App{app:Weyl}). The Gamma matrices $\Gamma^{\mu}$ in $(d-1)$d are the same as those in $d$d for odd $d$ and $\mu=0,1,\dots,d-2$. Therefore, for odd $d$, the projection maps a Dirac (Majorana) fermion in $d$d to a Weyl (Majorana-Weyl) fermion in $(d-1)$d. For even $d$, the projection maps a Dirac (Majorana) fermion in $d$d to a Dirac (Majorana) fermion in $(d-1)$d. See also \Fig{fig:domain-wall-fermion}.}
\label{fig:mass-profile}
\end{figure}
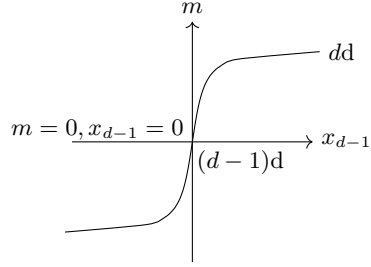

To prove Thorem \ref{thm-DW}, we first answer the following question.
If a $d$d fermion is in the bulk with the mass profile shown in \Fig{fig:mass-profile}, then what is the fermion on $(d-1)$d domain wall?

\begin{figure}[!h]
\begin{center}
\begin{tikzpicture}[scale=0.8]

  \draw[scale=1, domain=-1.5:1.5, smooth, variable=\y]  plot ({0.15*tan(\y r)}, {\y}); 
   
  \draw[scale=1, domain=-1.5:1.5, smooth, variable=\y]  plot ({4-0.15*tan(\y r)}, {\y});
\draw[dashed] (0,-2)--(0,2);
\draw[dashed] (4,-2)--(4,2);
\draw[dashed,->] (2,0)--(0,0);
\draw[dashed,->] (2,0)--(4,0);
  \node[below] at (2,0) {$x$};
\end{tikzpicture}
\end{center}
\caption{The distance between two domain walls is $x$, we can take $x \to0$, which becomes massive $-m$ bulk fermion,
and $x \to\infty$, which becomes massive $+m$ bulk fermion,
but degree of freedom matches. So 2 domain wall degree of freedom in $(d-1)$-dim = bulk degree of freedom in $d$-dim.}
\label{fig:DOF}
\end{figure}

By \Fig{fig:DOF}, the degree of freedom of the fermion on $(d-1)$d domain wall becomes half of the degree of freedom of the fermion in the $d$d bulk, we find that 
 Dirac fermions become Dirac fermions and Majorana fermions become Majorana fermions under the domain wall dimensional reduction for even $d$, and Dirac fermions become Weyl fermions and Majorana fermions become Majorana-Weyl fermions under the domain wall dimensional reduction for odd $d$. The domain wall dimensional reduction of Majorana fermions only exists for $d=2,3,4\mod8$ because the fermion in the bulk has to be massive and a massive Majorana fermion only exists for $d=2,3,4\mod8$.
We summarize the result in low dimensions in \Fig{fig:domain-wall-fermion}. 
This is consistent with the domain wall dimensional reduction of the C-R-T fractionalization and R-T fractionalization, see \eqref{eq:reduction-odd} and \eqref{eq:reduction-even}.
\begin{figure}[!h]
\begin{center}
\begin{tikzpicture}
\matrix[matrix of math nodes,inner sep=1pt,row sep=1em,column sep=1em] (M)
{
d & \text{Majorana-Weyl} & \text{Majorana} & \text{Weyl} & \text{Dirac}\\
1 & & 1^{\R} & & 1^{\C}\\
2 & 1_L^{\R} & 1_L^{\R}\oplus 1_R^{\R} & 1_L^{\C} & 1_L^{\C}\oplus 1_R^{\C}\\
3 & & 2^{\R} &  & 2^{\C}\\
4 & & 4^{\R} & 2_L^{\C} &  2_L^{\C}\oplus 2_R^{\C}\\
5 & & & & 4^{\C} \\
6 & & & 4_L^{\C} & 4_L^{\C}\oplus 4_R^{\C}\\
7 & & & & 8^{\C} \\
8 & &  16^{\R} & 8_L^{\C} & 8_L^{\C}\oplus 8_R^{\C}\\
9 & & 16^{\R} &  & 16^{\C}\\
10 &  16_L^{\R} & 16_L^{\R}\oplus 16_R^{\R} & 16_L^{\C} & 16_L^{\C}\oplus 16_R^{\C}\\
11 & & 32^{\R} &  & 32^{\C}\\
12 & & 64^{\R} & 32_L^{\C} & 32_L^{\C}\oplus 32_R^{\C}\\
}
;
\draw[->] (M-3-3.north) -- (M-2-3.south);
\draw[->] (M-3-5.north) -- (M-2-5.south);
\draw[->] (M-4-3.north west) -- (M-3-2.south east);
\draw[->] (M-4-5.north west) -- (M-3-4.south east);
\draw[->] (M-5-3.north) -- (M-4-3.south);
\draw[->] (M-5-5.north) -- (M-4-5.south);
\draw[->] (M-6-5.north west) -- (M-5-4.south east);
\draw[->] (M-7-5.north) -- (M-6-5.south);
\draw[->] (M-8-5.north west) -- (M-7-4.south east);
\draw[->] (M-9-5.north) -- (M-8-5.south);
\draw[->] (M-10-5.north west) -- (M-9-4.south east);
\draw[->] (M-11-3.north) -- (M-10-3.south);
\draw[->] (M-11-5.north) -- (M-10-5.south);
\draw[->] (M-12-3.north west) -- (M-11-2.south east);
\draw[->] (M-12-5.north west) -- (M-11-4.south east);
\draw[->] (M-13-3.north) -- (M-12-3.south);
\draw[->] (M-13-5.north) -- (M-12-5.south);

\end{tikzpicture}
\end{center}
\caption{Domain wall dimensional reduction of fermions. The target of the arrow is the domain wall of the source of the arrow. The number indicates the dimension of the representation of the fermion, the upper index indicates that the representation is real or complex, and the lower index indicates the chirality of the fermion. }
\label{fig:domain-wall-fermion}
\end{figure}

For comparison with \Fig{fig:domain-wall-fermion}, we summarize the domain wall dimensional reduction of fermions and their symmetry breaking in \Table{table:domain-wall-fermion} where we omit the arrows that indicate the domain wall dimensional reduction. \cred{We note that the domain wall reduction in \Fig{fig:domain-wall-fermion} is helpful to derive the relation between $N_f=3$ families of 16 Weyl fermions of the Standard Model in 4d to the 48 
Majorana-Weyl fermions in 2d \cite{Wang:2023tbjFamily2312.14928}.}

\begin{table}
\begin{tabular}{ccccc}
   $d$ & Majorana-Weyl & Majorana & Weyl & Dirac\\
   1 & &
\begin{tabular}{c|cc}
 & C & T\\\hline
0 & 1  & $\checkmark$
\end{tabular}
     &    
& \begin{tabular}{c|cc}
 & C & T\\\hline
0 & $\checkmark$  & $\checkmark$\\
$m$ & \xmark & $\checkmark$
\end{tabular}\\ 
2 & \begin{tabular}{c|ccc}
 & C & R & T\\\hline
0 & 1  & \xmark & \xmark 
\end{tabular}
& \begin{tabular}{c|ccc}
 & C & R & T\\\hline
0 & 1  & $\checkmark$ & $\checkmark$\\
$m$ & 1 & $\checkmark$ & $\checkmark$
\end{tabular}
& \begin{tabular}{c|ccc}
 & C & R & T\\\hline
0 & $\checkmark$  & \xmark & \xmark
\end{tabular}
& \begin{tabular}{c|ccc}
 & C & R & T\\\hline
0 & $\checkmark$  & $\checkmark$ & $\checkmark$\\
$m$ & $\checkmark$ & $\checkmark$ & $\checkmark$
\end{tabular}\\
3 & &  \begin{tabular}{c|ccc}
 & C & R & T\\\hline
0 & 1  & $\checkmark$ & $\checkmark$\\
$m$ & 1 & \xmark & \xmark
\end{tabular}
& & \begin{tabular}{c|ccc}
 & C & R & T\\\hline
0 & $\checkmark$  & $\checkmark$ & $\checkmark$\\
$m$ & $\checkmark$ & \xmark & \xmark
\end{tabular}\\
4 & & \begin{tabular}{c|ccc}
 & C & R & T\\\hline
0 & 1  & $\checkmark$ & $\checkmark$\\
$m$ & 1 & $\checkmark$ & $\checkmark$
\end{tabular}
& \begin{tabular}{c|ccc}
 & C & R & T\\\hline
0 & \xmark  & \xmark & $\checkmark$
\end{tabular}
& \begin{tabular}{c|ccc}
 & C & R & T\\\hline
0 & $\checkmark$  & $\checkmark$ & $\checkmark$\\
$m$ & $\checkmark$ & $\checkmark$ & $\checkmark$
\end{tabular}\\
5 & &  & & \begin{tabular}{c|ccc}
 & C & R & T\\\hline
0 & $\checkmark$  & $\checkmark$ & $\checkmark$\\
$m$ & \xmark & \xmark & $\checkmark$
\end{tabular}\\
6 & 
& 
& \begin{tabular}{c|ccc}
 & C & R & T\\\hline
0 & $\checkmark$  & \xmark & \xmark
\end{tabular}
& \begin{tabular}{c|ccc}
 & C & R & T\\\hline
0 & $\checkmark$  & $\checkmark$ & $\checkmark$\\
$m$ & $\checkmark$ & $\checkmark$ & $\checkmark$
\end{tabular}\\
7 & &  
& & \begin{tabular}{c|ccc}
 & C & R & T\\\hline
0 & $\checkmark$  & $\checkmark$ & $\checkmark$\\
$m$ & $\checkmark$ & \xmark & \xmark
\end{tabular}\\
8 & & \begin{tabular}{c|ccc}
 & C & R & T\\\hline
0 & 1  & $\checkmark$ & $\checkmark$
\end{tabular}
& \begin{tabular}{c|ccc}
 & C & R & T\\\hline
0 & \xmark  & \xmark & $\checkmark$
\end{tabular}
& \begin{tabular}{c|ccc}
 & C & R & T\\\hline
0 & $\checkmark$  & $\checkmark$ & $\checkmark$\\
$m$ & $\checkmark$ & $\checkmark$ & $\checkmark$
\end{tabular}\\
   9 & &
\begin{tabular}{c|ccc}
 & C & R & T\\\hline
0 & 1 & $\checkmark$ & $\checkmark$
\end{tabular}
     &    
& \begin{tabular}{c|ccc}
 & C & R & T\\\hline
0 & $\checkmark$  & $\checkmark$ & $\checkmark$\\
$m$ & \xmark &\xmark &  $\checkmark$
\end{tabular}\\ 
10 & \begin{tabular}{c|ccc}
 & C & R & T\\\hline
0 & 1  & \xmark & \xmark 
\end{tabular}
& \begin{tabular}{c|ccc}
 & C & R & T\\\hline
0 & 1  & $\checkmark$ & $\checkmark$\\
$m$ & 1 & $\checkmark$ & $\checkmark$
\end{tabular}
& \begin{tabular}{c|ccc}
 & C & R & T\\\hline
0 & $\checkmark$  & \xmark & \xmark
\end{tabular}
& \begin{tabular}{c|ccc}
 & C & R & T\\\hline
0 & $\checkmark$  & $\checkmark$ & $\checkmark$\\
$m$ & $\checkmark$ & $\checkmark$ & $\checkmark$
\end{tabular}\\
11 & &  \begin{tabular}{c|ccc}
 & C & R & T\\\hline
0 & 1  & $\checkmark$ & $\checkmark$\\
$m$ & 1 & \xmark & \xmark
\end{tabular}
& & \begin{tabular}{c|ccc}
 & C & R & T\\\hline
0 & $\checkmark$  & $\checkmark$ & $\checkmark$\\
$m$ & $\checkmark$ & \xmark & \xmark
\end{tabular}\\
12 & & \begin{tabular}{c|ccc}
 & C & R & T\\\hline
0 & 1  & $\checkmark$ & $\checkmark$\\
$m$ & 1 & $\checkmark$ & $\checkmark$
\end{tabular}
& \begin{tabular}{c|ccc}
 & C & R & T\\\hline
0 & \xmark  & \xmark & $\checkmark$
\end{tabular}
& \begin{tabular}{c|ccc}
 & C & R & T\\\hline
0 & $\checkmark$  & $\checkmark$ & $\checkmark$\\
$m$ & $\checkmark$ & $\checkmark$ & $\checkmark$
\end{tabular}
\end{tabular}
\caption{Domain wall dimensional reduction of fermions and their symmetry breaking. Here $\checkmark$ means the symmetry is preserved, \xmark\; means the symmetry is broken, $0$ means massless, and $m$ means the Dirac or Majorana mass implemented on the Dirac or Majorana fermion respectively. Here 1 in C symmetry means C becomes the identity operator.}
\label{table:domain-wall-fermion}
\end{table}

\begin{figure}
          \begin{center}
\begin{tikzpicture}[scale=0.6]

  \draw[scale=0.5, domain=-1.5:1.5, smooth, variable=\y]  plot ({0.5*tan(\y r)}, {\y}); 
\draw[scale=0.5, domain=-7.78:-4.78, smooth, variable=\y]  plot ({-0.5*tan(\y r)}, {\y});
  \draw[scale=0.5, domain=-14.07:-11.07, smooth, variable=\y]  plot ({0.5*tan(\y r)}, {\y});
\draw[dashed] (0,-7)--(0,1);
\node[above] at (0,1) {$x_{d-1}=0$};
\draw[->] (-1,-1)--(-1,-2);
\draw[->] (1,0.5)--(1,-3.5);
\draw[->] (-1,-3) -- (1,-5);
\draw[->] (1,-4.5) -- (-1,-6.5);
\node[left] at (-1,-1.5) {$\text{S}_b$};
\node[right] at (1,-1.5) {$\text{S}_b$};
\node[left] at (0,-4.75) {$\rC\rR_{d-1}\rT$};
\end{tikzpicture}

\end{center}
            \caption{Domain wall dimensional reduction. Suppose $\text{S}_b$ is (spontaneously or explicitly) broken by mass, $\text{S}_b$ switches the two degenerate ground states, and assume that $\text{S}_b$ does not flip the spatial coordinate $x_{d-1}$ (which excludes $\rR_{d-1}$, but $\text{S}_b$ can be $\rC$, $\rR_i$ for $i\ne d-1$, and $\rT$, or other internal symmetries), 
            we need to combine it with the space-orientation-reversing symmetry $\rC\rR_{d-1}\rT$ to obtain a new symmetry on the domain wall 
            \cite{HasonKomargodskiThorngren1910.14039,1910.14046, Wang:2019obe1910.14664}.
            \cred{However, some of the domain wall 
            symmetries may not be directly induced from $\text{S}_b\cdot (\rC \rR_{d-1} \rT)$, 
             the domain wall symmetries may be induced from the $\text{S}_b= \rR_{d-1}$ itself or other symmetries.}  
            }
           \label{fig:domain-wall}
        \end{figure}

Next, for a given $d$d fermion $\psi^{d}$ in the $d$d bulk, we determine how C, R, and T change under the domain wall dimensional reduction. The cases are different for odd $d$ and even $d$. We will follow the principle in \cite{HasonKomargodskiThorngren1910.14039,1910.14046, Wang:2019obe1910.14664}.

\cred{We assume that the fermion parity $(-1)^{\rF}$ is not spontaneously broken, if the C, R, and T symmetries are spontaneously broken $\Z_4$ symmetries (they are $\Z_2$ or $\Z_4$ symmetries in our assumption, see Lemma \ref{order}), then there are still only two degenerate ground states.} We consider the mass profile (see \Fig{fig:mass-profile}) which can be spontaneously or explicitly generated, and see which symmetry is broken by mass.

If the spacetime dimension $d$ is odd, the conventional 
Dirac and Majorana mass terms break part of the C, R, or T symmetries, 
then we will denote the broken symmetry generator S$_b$ as the generator of a part of the broken C, R, or T symmetries broken by the mass term.\footnote{Note 
that the broken symmetry may not form a group. But we consider the broken symmetry generated by S$_b$ that forms 
a cyclic group.}
See \Fig{fig:domain-wall}.

We will denote
$\text{S}_b$ ($b$ for broken) as the broken symmetry \jw{(in the bulk)}, 
and $\text{S}_p$ ($p$ for preserved) as the preserved symmetry \jw{(both in the bulk and on the domain wall)}. 
If there are exactly two broken symmetries $\text{S}_{b_1}$ and $\text{S}_{b_2}$ among the C, R, or T symmetries, then by the CRT theorem, $\text{S}_{b_1} \cdot \text{S}_{b_2}$ is preserved, namely
\bea
 \text{S}_{b_1} \cdot \text{S}_{b_2}=\text{S}_p.
\eea

First, we consider Dirac fermions.
By \Table{table:CRT-break-Dirac}, 
the Dirac mass term breaks part of the C, R, or T symmetries for $d=1\mod4$ and for $d=3\mod4$. \\
$\bullet$ For $d=1\mod4$ Dirac fermion, 
the Dirac mass term breaks the C and R symmetries,
we consider $\text{S}_b=\rC$ or $\rR$.\\
$\bullet$ For $d=3\mod4$ Dirac fermion, 
the Dirac mass term breaks the T and R symmetries,
we consider $\text{S}_b=\rT$ or $\rR$.\\

Next, we consider Majorana fermions.\\
$\bullet$ By \Table{table:CRT-break-Majorana}, for $d=3\mod8$ Majorana fermion, 
the Majorana mass term breaks the T and R symmetries,
we consider $\text{S}_b=\rT$ or $\rR$.\\

For either Dirac or Majorana fermion case, the $\text{S}_b$ can be unitary or anti-unitary:

If $\text{S}_b=\rC$ is unitary, then it becomes anti-unitary $\rT'=\text{S}_b\cdot\rC\rR_{d-1}\rT=\rR_{d-1}\rT$ up to $(-1)^{\rF}$ on the domain wall.

.
If $\text{S}_b=\rT$ is anti-unitary, then it becomes unitary $\rC'=\text{S}_b\cdot\rC\rR_{d-1}\rT=\rC\rR_{d-1}$ up to $(-1)^{\rF}$ on the domain wall.

If $\text{S}_b=\rR_{d-1}$ is unitary, then it switches the two degenerate ground states and flips the spatial coordinate $x_{d-1}$, hence it preserves the mass profile, so it becomes an internal unitary symmetry $X$ on the domain wall.
By \Table{tab:CPRT-trans-dd}, the matrix of $\rR_{d-1}$ is $\eta_{\rR_{d-1}}\Gamma^{d-1}=\eta_{\rR_{d-1}}\ii\diag(-I,I)$ in the chiral representation where $I$ is an identity matrix (see \App{app:Weyl}), the matrix of $X$ is $\mp\eta_{\rR_{d-1}}\ii I$ 
\cred{(which generates a unitary $\Z_{4,X}$ symmetry if $\eta_{\rR_{d-1}}=\pm1$ or $\Z_2^{\rF}$ symmetry or trivial symmetry if $\eta_{\rR_{d-1}}=\pm\ii$)}
under the projection 
${\rm P}_{\pm}= \frac{1 \pm \ii \Gamma^{d-1}}{2}$.

If $\text{S}_b=\rR_i$ ($i\ne d-1$) is unitary, then it becomes anti-unitary $\text{S}_b\cdot\rC\rR_{d-1}\rT=\rR_i\rC\rR_{d-1}\rT$ 
on the domain wall. Since the matrix of $\rR_i\rC\rR_{d-1}\rT$ is $\eta\Gamma^i\Gamma^0\Gamma^{d-1}$ which is block diagonal in the chiral representation (see \App{app:Weyl}), $\rR_i$ becomes an unbroken anti-unitary symmetry $\rC\rR_i\rT\cdot X$ up to $(-1)^{\rF}$ on the domain wall where $\rC\rR_i\rT$ is induced from the $\rC\rR_i\rT$ in the bulk and it is unbroken, and $X$ is induced from $\rR_{d-1}$ in the bulk.

Therefore,
for odd $d$, the domain wall dimensional reduction of the symmetries C, R, and T obeys the following rule (up to $(-1)^{\rF}$): 
\begin{enumerate}
\item
For Dirac fermions and $d=1\mod4$, 
\bea
d\text{-}\dim&\xrightarrow{\rm DW}&(d-1)\text{-}\dim\nn\\
\rR_{d-1}&\xrightarrow{{\rm P}_{\pm}}&X\nn\\
\rR_i&\xrightarrow{\cdot \rC\rR_{d-1}\rT}&\rC\rR_i\rT\cdot X\;\;\;(i\ne d-1)\nn\\
\rC&\xrightarrow{\cdot \rC\rR_{d-1}\rT}&\rT'=X\rT \nn\\
\rT&\xrightarrow{{\rm P}_{\pm}}&\rT.
\eea
\cred{The $(d-1)$d domain wall (DW) 
Weyl fermion has the corresponding symmetry
such that $\rC_{\rm DW}$ and $\rR_{\rm DW}$ are broken,
$\rT_{\rm DW} = \rT$ or $\rT'=X\rT$, and $\rC_{\rm DW}\rR_{\rm DW}\rT_{\rm DW}$ is unbroken,
while there is also an extra internal unitary $X$ symmetry (which generates $\Z_{4,X}$ or $\Z_2^{\rF}$ or 1).} 
By \Table{tab:CPRT-trans-dd}, the matrices of T for Dirac fermions in $d=1\mod4$ and $d-1=0\mod4$ are the same and they are block diagonal in the chiral representation (see \App{app:Weyl}). So the matrix of T for Dirac fermions in $d=1\mod4$ becomes the matrix of T for Weyl fermions in $d-1=0\mod4$ under the projection 
${\rm P}_{\pm}= \frac{1 \pm \ii \Gamma^{d-1}}{2}$.
\item
For Dirac fermions and $d=3\mod4$,
\bea
d\text{-}\dim&\xrightarrow{\rm DW}&(d-1)\text{-}\dim\nn\\
\rR_{d-1}&\xrightarrow{{\rm P}_{\pm}}&X\nn\\
\rR_i&\xrightarrow{\cdot \rC\rR_{d-1}\rT}&\rC\rR_i\rT\cdot X\;\;\;(i\ne d-1)\nn\\
\rC&\xrightarrow{{\rm P}_{\pm}}&\rC \nn\\
\rT&\xrightarrow{\cdot \rC\rR_{d-1}\rT}&\rC'=\rC X.
\eea
\cred{The $(d-1)$d domain wall (DW) 
Weyl fermion has the corresponding symmetry
such that $\rT_{\rm DW}$ and $\rR_{\rm DW}$ are broken,
$\rC_{\rm DW} = \rC$ or $\rC'=\rC X$, and $\rC_{\rm DW}\rR_{\rm DW}\rT_{\rm DW}$ is unbroken,
while there is also an extra internal unitary $X$ symmetry (which generates $\Z_{4,X}$ or $\Z_2^{\rF}$ or 1).}
By \Table{tab:CPRT-trans-dd}, the matrices of C for Dirac fermions in $d=3\mod4$ and $d-1=2\mod4$ are the same and they are block diagonal in the chiral representation (see \App{app:Weyl}). So the matrix of C for Dirac fermions in $d=3\mod4$ becomes the matrix of C for Weyl fermions in $d-1=2\mod4$ under the projection 
${\rm P}_{\pm}= \frac{1 \pm \ii \Gamma^{d-1}}{2}$.
\item
For Majorana fermions and $d=3\mod8$, the matrix of $\rR_{d-1}$ is $\ii\Gamma^{d-1}$ where $\Gamma^{d-1}=\ii\sigma^3\otimes I$ in the Majorana representation (see \App{app:Majorana}) and $I$ is an identity matrix. The matrix of $X$ is $\pm I$ 
\cred{(which generates a unitary $\Z_2^{\rF}$ symmetry or trivial symmetry)}
under the projection 
${\rm P}_{\pm}= \frac{1 \pm \ii \Gamma^{d-1}}{2}$. So $X=1$ up to $(-1)^{\rF}$.
\bea
d\text{-}\dim&\xrightarrow{\rm DW}&(d-1)\text{-}\dim\nn\\
\rR_{d-1}&\xrightarrow{{\rm P}_{\pm}}&X=1\nn\\
\rR_i&\xrightarrow{\cdot \rC\rR_{d-1}\rT}&\rC\rR_i\rT\cdot X\;\;\;(i\ne d-1)\nn\\
\rC=1&\xrightarrow{{\rm P}_{\pm}}&\rC=1 \nn\\
\rT&\xrightarrow{\cdot \rC\rR_{d-1}\rT}&\rC'=\rC X=1.
\eea
The $(d-1)$d domain wall (DW) 
Majorana-Weyl fermion has the corresponding symmetry
such that $\rT_{\rm DW}$ and $\rR_{\rm DW}$ are broken,
$\rC_{\rm DW}$ is trivial, and $\rC_{\rm DW}\rR_{\rm DW}\rT_{\rm DW}$ is unbroken.
Here $\rC=1$ because the C symmetry of a Majorana fermion is trivial.
\end{enumerate}

By comparing the results in \Sec{sec:Dirac} and \Sec{sec:Majorana}, we conclude that for odd $d$,
\bea\label{eq:reduction-odd}
d\text{-}\dim&\xrightarrow{\rm DW}&(d-1)\text{-}\dim\nn\\
G_{{\rm D},i}(d\mod8)&\to&G_{\rm W}(d-1\mod8)\nn\\
G_{\rm M}(3\mod8)&\to&G_{\rm MW}(2\mod8).
\eea
Here $G_{{\rm D},i}(d\mod8)$ $(i=1,2)$ are the two possibilities for the C-R-T fractionalization with the canonical CRT for Dirac fermions in spacetime dimensions $d$, $G_{\rm W}(d\mod8)$ is the C or T fractionalization for Weyl fermions in even spacetime dimensions $d$, $G_{\rm M}(d\mod8)$ is the R-T fractionalization for Majorana fermions in spacetime dimensions $d=0,1,2,3,4\mod8$, and $G_{\rm MW}(2\mod8)=\Z_2^{\rF}$ is the trivial C fractionalization for Majorana-Weyl fermions in spacetime dimensions $d=2\mod8$.  
There is no domain wall dimensional reduction from $G_{\rm M}(1\mod8)$ because a Majorana fermion in $d=1\mod8$ has to be massless while the fermion in the bulk has to be massive.

If the spacetime dimension $d$ is even, the conventional Dirac and Majorana mass terms do not break C, R, and T symmetries. 

First, we consider Dirac fermions.
Since $\Gamma^{d-1}$ is not block diagonal in the chiral representation (see \App{app:Weyl}), we replace $\Gamma^{d-1}$ with ${\Gamma^{d-1}}'=\ii^{\frac{d}{2}}\Gamma^0\Gamma^1\cdots\Gamma^{d-1}=\ii\diag(-I,I)$ which is block diagonal to do the projection. Note that ${\Gamma^{d-1}}'$ is $\Gamma^d$ in $(d+1)$d chiral representation (see \App{app:Weyl}). However, ${\Gamma^{d-1}}'$ is imaginary while $\Gamma^{d-1}$ is real. So the conditions \eqref{eq:CPRT-condition} on the matrices of C, R, and T imply that 
\bea
M_{\rC}&=&\eta_{\rC}\left\{\begin{array}{ll}\Gamma^0\Gamma^1\Gamma^3\cdots\Gamma^{d-3}&d=0\mod4\\\Gamma^2\Gamma^4\cdots\Gamma^{d-2}{\Gamma^{d-1}}'&d=2\mod4\end{array}\right.\nn\\
M_{\rT}&=&\eta_{\rT}\left\{\begin{array}{ll}\Gamma^0\Gamma^2\Gamma^4\cdots\Gamma^{d-2}{\Gamma^{d-1}}'&d=0\mod4\\\Gamma^1\Gamma^3\cdots\Gamma^{d-3}&d=2\mod4\end{array}\right.\nn\\
M_{\rR_i}&=&\eta_{\rR_i}\Gamma^0\Gamma^1\cdots\hat{\Gamma}^i\cdots{\Gamma^{d-1}}'.
\eea
These matrices (except $M_{\rR_{d-1}}$) are block diagonal in the chiral representation (see \App{app:Weyl}).
Under the projection ${\rm P}_{\pm}= \frac{1 \pm \ii {\Gamma^{d-1}}'}{2}$, these matrices become
\bea
M_{\rC}^{d-1,\pm}&=&\eta_{\rC}\left\{\begin{array}{ll}\pm(-1)^{\frac{d}{4}}\gamma^1\gamma^3\cdots\gamma^{d-3}&d=0\mod4\\\mp\ii(-1)^{\frac{d-2}{4}}\gamma^2\gamma^4\cdots\gamma^{d-2}&d=2\mod4\end{array}\right.\nn\\
M_{\rT}^{d-1,\pm}&=&\eta_{\rT}\left\{\begin{array}{ll}-\ii (-1)^{\frac{d}{4}}\gamma^2\gamma^4\cdots\gamma^{d-2}&d=0\mod4\\(-1)^{\frac{d-2}{4}}\gamma^1\gamma^3\cdots\gamma^{d-3}&d=2\mod4\end{array}\right.\nn\\
M_{\rR_i}^{d-1,\pm}&=&\eta_{\rR_i}\ii(-1)^{\frac{d}{2}}\gamma^1\cdots\hat{\gamma}^i\cdots\gamma^{d-2}\;\;\;(i\ne d-1).
\eea

The symmetry $\rR_{d-1}$ flips the spatial coordinate $x_{d-1}$, it is possible to be reduced to an internal unitary symmetry $X$ on the domain wall.
The matrix of $\rR_{d-1}$ is $\eta_{\rR_{d-1}}\Gamma^0\Gamma^1\cdots\Gamma^{d-2}$ which is not block diagonal in the chiral representation (see \App{app:Weyl}), so $X$ is broken on the domain wall.
Therefore,
for Dirac fermions and even $d$, the domain wall dimensional reduction of the symmetries C, R, and T obeys the following rule:
\bea
d\text{-}\dim&\xrightarrow{\rm DW}&(d-1)\text{-}\dim\nn\\
\rR_{d-1}&\xrightarrow{{\rm P}_{\pm}}&1\nn\\
\rR_i&\xrightarrow{{\rm P}_{\pm}}&\rR_i\;\;\;(i\ne d-1)\nn\\
\rC&\xrightarrow{{\rm P}_{\pm}}&\rC\nn\\
\rT&\xrightarrow{{\rm P}_{\pm}}&\rT.
\eea
The $(d-1)$d domain wall (DW) 
Dirac fermion has the corresponding symmetry
such that $\rT_{\rm DW}=\rT$, $\rR_{\rm DW}=\rR$,
and $\rC_{\rm DW} = \rC$. Namely, the C, R, and T symmetries on the domain wall are induced from the C, R, and T symmetries in the bulk.

Based on the composition rules \eqref{eq:composition-matrix1} and $\eqref{eq:composition-matrix2}$, we do the following computation.
\begin{enumerate}
    \item 
    For Dirac fermions and $d=0\mod4$, 
    \bea
\rC_{\rm DW}^2\psi_{\pm}^{d-1}\rC_{\rm DW}^{-2}&=&(\gamma^1\gamma^3\cdots\gamma^{d-3})^2\psi_{\pm}^{d-1}\nn\\
&=&(-1)^{\frac{d}{4}-1}\psi_{\pm}^{d-1}.
    \eea
    \bea
\rT_{\rm DW}^2\psi_{\pm}^{d-1}\rT_{\rm DW}^{-2}&=&(\gamma^2\gamma^4\cdots\gamma^{d-2})^2(-1)^{\frac{d}{2}-1}\psi_{\pm}^{d-1}\nn\\
&=&(-1)^{\frac{d}{4}}\psi_{\pm}^{d-1}.
    \eea
    \bea
    \rR_{{\rm DW},i}^2\psi_{\pm}^{d-1}\rR_{{\rm DW},i}^{-2}
    &=&\eta_{\rR_i}^2(-1)(\gamma^1\cdots\hat{\gamma}^i\cdots\gamma^{d-2})^2\psi_{\pm}^{d-1}\nn\\
    &=&\eta_{\rR_i}^2(-1)(-1)^{\frac{(d-4)(d-3)}{2}}\psi_{\pm}^{d-1}\nn\\
    &=&-\eta_{\rR_i}^2\psi_{\pm}^{d-1}.
    \eea

\bea
(\rC_{\rm DW}\rR_{{\rm DW},i})^2\psi_{\pm}^{d-1}(\rC_{\rm DW}\rR_{{\rm DW},i})^{-2}&=&(\gamma^1\cdots\hat{\gamma}^i\cdots\gamma^{d-2}\gamma^1\gamma^3\cdots\gamma^{d-3})(\gamma^1\cdots\hat{\gamma}^i\cdots\gamma^{d-2}\gamma^1\gamma^3\cdots\gamma^{d-3})^*\psi_{\pm}^{d-1}\nn\\
    &=&(-1)^{\frac{d}{4}-1}\psi_{\pm}^{d-1}.
\eea

    \bea
(\rR_{{\rm DW},i}\rT_{\rm DW})^2\psi_{\pm}^{d-1}(\rR_{{\rm DW},i}\rT_{\rm DW})^{-2}&=&(\gamma^2\gamma^4\cdots\gamma^{d-2}\gamma^1\cdots\hat{\gamma}^i\cdots\gamma^{d-2})^*(\gamma^2\gamma^4\cdots\gamma^{d-2}\gamma^1\cdots\hat{\gamma}^i\cdots\gamma^{d-2})\psi_{\pm}^{d-1}\nn\\
    &=&(-1)^{\frac{d}{4}-1}\psi_{\pm}^{d-1}.
\eea

    \bea
(\rC_{\rm DW}\rT_{\rm DW})^2\psi_{\pm}^{d-1}(\rC_{\rm DW}\rT_{\rm DW})^{-2}&=&(\eta_{\rC}^*\eta_{\rT}^*)^2(-1)(\gamma^1\gamma^2\cdots\gamma^{d-2})^2\psi_{\pm}^{d-1}\nn\\
&=&(\eta_{\rC}^*\eta_{\rT}^*)^2(-1)(-1)^{\frac{(d-3)(d-2)}{2}}\psi_{\pm}^{d-1}\nn\\
&=&(\eta_{\rC}^*\eta_{\rT}^*)^2\psi_{\pm}^{d-1}.
    \eea

Recall that $\eta_{\rR_i}=\pm1$ and $\eta_{\rC}\eta_{\rT}=\pm1$ in the first possibility of the C-R-T fractionalization with the canonical CRT for Dirac fermions, while $\eta_{\rR_i}=\pm\ii$ and $\eta_{\rC}\eta_{\rT}=\pm\ii$ in the second possibility of the C-R-T fractionalization with the canonical CRT for Dirac fermions.
Therefore, the first (second) possibility of the C-R-T fractionalization with the canonical CRT for Dirac fermions in $d=0\mod4$ is reduced to the first (second) possibility of the C-R-T fractionalization with the canonical CRT for Dirac fermions in $d-1=3\mod4$, see \Table{tab:CPRT-group-dd-1} and \Table{tab:CPRT-group-dd-2}.
    \item 
       For Dirac fermions and $d=2\mod4$, 

\bea
\rC_{\rm DW}^2\psi_{\pm}^{d-1}\rC_{\rm DW}^{-2}&=&(\gamma^2\gamma^4\cdots\gamma^{d-2})^2(-1)^{\frac{d}{2}-1}\psi_{\pm}^{d-1}\nn\\
&=&(-1)^{\frac{d-2}{4}}\psi_{\pm}^{d-1}.
    \eea
    \bea
\rT_{\rm DW}^2\psi_{\pm}^{d-1}\rT_{\rm DW}^{-2}&=&(\gamma^1\gamma^3\cdots\gamma^{d-3})^2\psi_{\pm}^{d-1}\nn\\
&=&(-1)^{\frac{d-2}{4}}\psi_{\pm}^{d-1}.
    \eea
    \bea
    \rR_{{\rm DW},i}^2\psi_{\pm}^{d-1}\rR_{{\rm DW},i}^{-2}
    &=&\eta_{\rR_i}^2(-1)(\gamma^1\cdots\hat{\gamma}^i\cdots\gamma^{d-2})^2\psi_{\pm}^{d-1}\nn\\
    &=&\eta_{\rR_i}^2(-1)(-1)^{\frac{(d-4)(d-3)}{2}}\psi_{\pm}^{d-1}\nn\\
    &=&\eta_{\rR_i}^2\psi_{\pm}^{d-1}.
    \eea

\bea
(\rC_{\rm DW}\rR_{{\rm DW},i})^2\psi_{\pm}^{d-1}(\rC_{\rm DW}\rR_{{\rm DW},i})^{-2}&=&(\gamma^1\cdots\hat{\gamma}^i\cdots\gamma^{d-2}\gamma^2\gamma^4\cdots\gamma^{d-2})(\gamma^1\cdots\hat{\gamma}^i\cdots\gamma^{d-2}\gamma^2\gamma^4\cdots\gamma^{d-2})^*\psi_{\pm}^{d-1}\nn\\
    &=&(-1)^{\frac{d-2}{4}-1}\psi_{\pm}^{d-1}.
\eea

    \bea
(\rR_{{\rm DW},i}\rT_{\rm DW})^2\psi_{\pm}^{d-1}(\rR_{{\rm DW},i}\rT_{\rm DW})^{-2}&=&(\gamma^1\gamma^3\cdots\gamma^{d-3}\gamma^1\cdots\hat{\gamma}^i\cdots\gamma^{d-2})^*(\gamma^1\gamma^3\cdots\gamma^{d-3}\gamma^1\cdots\hat{\gamma}^i\cdots\gamma^{d-2})\psi_{\pm}^{d-1}\nn\\
    &=&(-1)^{\frac{d-2}{4}}\psi_{\pm}^{d-1}.
\eea
    
    \bea
(\rC_{\rm DW}\rT_{\rm DW})^2\psi_{\pm}^{d-1}(\rC_{\rm DW}\rT_{\rm DW})^{-2}&=&(\eta_{\rC}^*\eta_{\rT}^*)^2(-1)(\gamma^1\gamma^2\cdots\gamma^{d-2})^2\psi_{\pm}^{d-1}\nn\\
&=&(\eta_{\rC}^*\eta_{\rT}^*)^2(-1)(-1)^{\frac{(d-3)(d-2)}{2}}\psi_{\pm}^{d-1}\nn\\
&=&-(\eta_{\rC}^*\eta_{\rT}^*)^2\psi_{\pm}^{d-1}.
    \eea
    Recall that $\eta_{\rR_i}=\pm1$ and $\eta_{\rC}\eta_{\rT}=\pm1$ in the first possibility of the C-R-T fractionalization with the canonical CRT for Dirac fermions, while $\eta_{\rR_i}=\pm\ii$ and $\eta_{\rC}\eta_{\rT}=\pm\ii$ in the second possibility of the C-R-T fractionalization with the canonical CRT for Dirac fermions.
    Therefore, the first (second) possibility of the C-R-T fractionalization with the canonical CRT for Dirac fermions in $d=2\mod4$ is reduced to the second (first) possibility of the C-R-T fractionalization with the canonical CRT for Dirac fermions in $d-1=1\mod4$, see \Table{tab:CPRT-group-dd-1} and \Table{tab:CPRT-group-dd-2}.
\end{enumerate}

Next, we consider Majorana fermions and $d=2,4\mod8$.
    \begin{enumerate}
      \item 
      For Majorana fermions and $d=4\mod8$, we choose imaginary Gamma matrices. Note that $\Gamma^{d-1}$ is block diagonal in the Majorana representation (see \App{app:Majorana}). 
      By the results in \Sec{sec:Majorana}, we have $M_{\rT}=\ii\Gamma^1\Gamma^2\cdots\Gamma^{d-1}$ and $M_{\rR_i}=\ii\Gamma^0\Gamma^1\cdots\hat{\Gamma}^i\cdots\Gamma^{d-1}$. 
These matrices (except $M_{\rR_{d-1}}$) are block diagonal in the Majorana representation (see \App{app:Majorana}).
For example, we consider the case $d=4$.
Under the projection ${\rm P}_{\pm}= \frac{1 \pm \ii \Gamma^{d-1}}{2}$, these matrices become (for $d=4$)
\bea
M_{\rT}^{d-1,\pm}&=&-\ii\sigma^2,\nn\\
M_{\rR_1}^{d-1,\pm}&=&\sigma^1,\nn\\
M_{\rR_2}^{d-1,\pm}&=&\mp\sigma^3.
\eea
By comparison with the Majorana representation for $d-1=3$, we find that (up to $\pm1$)
\bea
M_{\rT}^{d-1,\pm}&=&\ii\Gamma^0_{d-1},\nn\\
M_{\rR_i}^{d-1,\pm}&=&\ii\Gamma^i_{d-1}\;\;\;(i\ne d-1).
\eea
These are exactly the matrices of R and T for $d-1=3$. Therefore, 
the R-T fractionalization for Majorana fermions in $d=4\mod8$ is reduced to the R-T fractionalization for Majorana fermions in $d-1=3\mod8$.
      \item
       For Majorana fermions and $d=2\mod8$, we can choose imaginary or real Gamma matrices. However, we choose imaginary Gamma matrices because the Majorana fermion in the bulk has to be massive. Since $\Gamma^{d-1}$ is not block diagonal in the Majorana representation (see \App{app:Majorana}), we replace $\Gamma^{d-1}$ with ${\Gamma^{d-1}}'= \ii\sigma^3\otimes I$ which is block diagonal to do the projection where $I$ is an identity matrix. Note that ${\Gamma^{d-1}}'$ is $\Gamma^d$ in $(d+1)$d Majorana representation (see \App{app:Majorana}).
       By the results in \Sec{sec:Majorana}, we have $M_{\rT}=\ii\Gamma^1\Gamma^2\cdots{\Gamma^{d-1}}'$ and $M_{\rR_i}=\ii\Gamma^0\Gamma^1\cdots\hat{\Gamma}^i\cdots{\Gamma^{d-1}}'$ if the Gamma matrices are imaginary. 
These matrices (except $M_{\rR_{d-1}}$) are block diagonal in the Majorana representation (see \App{app:Majorana}).
For example, we consider the case $d=10$.
Under the projection ${\rm P}_{\pm}= \frac{1 \pm \ii {\Gamma^{d-1}}'}{2}$, these matrices become (for $d=10$)
\bea
M_{\rT}^{d-1,\pm}&=&\mp\sigma^2\otimes\sigma^1\otimes I \otimes\sigma^2,\nn\\
M_{\rR_1}^{d-1,\pm}&=&\mp\sigma^1\otimes I\otimes I\otimes\ii\sigma^2,\nn\\
M_{\rR_2}^{d-1,\pm}&=&\pm\sigma^3\otimes I\otimes I\otimes\ii\sigma^2,\nn\\
M_{\rR_3}^{d-1,\pm}&=&\mp\ii\sigma^2\otimes\sigma^2\otimes I\otimes \sigma^2,\nn\\
M_{\rR_4}^{d-1,\pm}&=&\pm\ii\sigma^2\otimes\sigma^3\otimes\sigma^2\otimes\sigma^2,\nn\\
M_{\rR_5}^{d-1,\pm}&=&\mp\ii\sigma^2\otimes\sigma^3\otimes\sigma^1\otimes I,\nn\\
M_{\rR_6}^{d-1,\pm}&=&\pm\ii\sigma^2\otimes\sigma^3\otimes\sigma^3\otimes I,\nn\\
M_{\rR_7}^{d-1,\pm}&=&\mp I\otimes I\otimes\ii\sigma^2\otimes\sigma^1,\nn\\
M_{\rR_8}^{d-1,\pm}&=&\pm I\otimes I\otimes\ii\sigma^2\otimes\sigma^3.
\eea
By comparison with the Majorana representation for $d-1=9$, we find that (up to $\pm1$)
\bea
M_{\rT}^{d-1,\pm}&=&\Gamma^0_{d-1},\nn\\
M_{\rR_i}^{d-1,\pm}&=&\Gamma^i_{d-1}\;\;\;(i\ne d-1).
\eea
These are exactly the matrices of R and T for $d-1=9$. Therefore, 
the first choice of the R-T fractionalization for Majorana fermions in $d=2\mod8$ is reduced to the R-T fractionalization for Majorana fermions in $d-1=1\mod8$.
\end{enumerate}

For Majorana fermions and $d=2,4\mod8$, the symmetry $\rR_{d-1}$ flips the spatial coordinate $x_{d-1}$, it is possible to be reduced to an internal unitary symmetry $X$ on the domain wall.
The matrix of $\rR_{d-1}$ is $\ii\Gamma^0\Gamma^1\cdots\Gamma^{d-2}$ which is not block diagonal in the Majorana representation (see \App{app:Majorana}), so $X$ is broken on the domain wall.

Therefore,
for Majorana fermions and $d=2,4\mod8$, the domain wall dimensional reduction of the symmetries R and T obeys the following rule:
\bea
d\text{-}\dim&\xrightarrow{\rm DW}&(d-1)\text{-}\dim\nn\\
\rR_{d-1}&\xrightarrow{{\rm P}_{\pm}}&1\nn\\
\rR_i&\xrightarrow{{\rm P}_{\pm}}&\rR_i\;\;\;(i\ne d-1)\nn\\
\rC=1&\xrightarrow{{\rm P}_{\pm}}&\rC=1\nn\\
\rT&\xrightarrow{{\rm P}_{\pm}}&\rT.
\eea
The $(d-1)$d domain wall (DW) 
Majorana fermion has the corresponding symmetry
such that $\rT_{\rm DW}=\rT$ and $\rR_{\rm DW}=\rR$. Namely, the R and T symmetries on the domain wall are induced from the R and T symmetries in the bulk. Here $\rC=1$ because the C symmetry of a Majorana fermion is trivial.

By the above discussion,
we conclude that 
\bea\label{eq:reduction-even}
d\text{-}\dim&\xrightarrow{\rm DW}&(d-1)\text{-}\dim\nn\\
G_{{\rm D},1}(0\mod8)&\to&G_{{\rm D},1}(7\mod8)\nn\\
G_{{\rm D},1}(2\mod8)&\to&G_{{\rm D},2}(1\mod8)\nn\\
G_{{\rm D},1}(4\mod8)&\to&G_{{\rm D},1}(3\mod8)\nn\\
G_{{\rm D},1}(6\mod8)&\to&G_{{\rm D},2}(5\mod8)\nn\\
G_{{\rm D},2}(0\mod8)&\to&G_{{\rm D},2}(7\mod8)\nn\\
G_{{\rm D},2}(2\mod8)&\to&G_{{\rm D},1}(1\mod8)\nn\\
G_{{\rm D},2}(4\mod8)&\to&G_{{\rm D},2}(3\mod8)\nn\\
G_{{\rm D},2}(6\mod8)&\to&G_{{\rm D},1}(5\mod8)\nn\\
G_{\rm M,1}(2\mod8)&\to&G_{\rm M}(1\mod8)\nn\\
G_{\rm M}(4\mod8)&\to&G_{\rm M}(3\mod8).
\eea
Here $G_{{\rm D},i}(d\mod8)$ $(i=1,2)$ are the two possibilities for the C-R-T fractionalization with the canonical CRT for Dirac fermions in spacetime dimensions $d$, and $G_{\rm M}(d\mod8)$ is the R-T fractionalization for Majorana fermions in spacetime dimensions $d=0,1,2,3,4\mod8$.  
The C-R-T fractionalization with the canonical CRT for Dirac fermions in even spacetime dimensions $d$ and its domain wall dimensional reduction are the same.
Here for $d=2\mod8$, we choose the first one of the two choices of the R-T fractionalization because the first one is for the massive case and the fermion in the bulk has to be massive. There is no domain wall dimensional reduction from $G_{\rm M}(0\mod8)$ because a Majorana fermion in $d=0\mod8$ has to be massless while the fermion in the bulk has to be massive.

Therefore, we have proved Theorem \ref{thm-DW}.

\section{Conclusion}\label{sec:conclusion}

In this section, we summarize the results in this article.
Our main results are the following:
\begin{itemize}
    \item 
    In \Sec{sec:Dirac}, we study the group (called the C-R-T fractionalization) generated by C, R, and T for Dirac fermions which is also a group extension \eqref{eq:extension}.
    We determine the C-R-T fractionalization with the canonical CRT and its extension class for Dirac fermions in any spacetime dimension, and we find an 8-periodicity. See \Table{tab:CPRT-group-dd-1} and \Table{tab:CPRT-group-dd-2}. We find that there are two possibilities for the C-R-T fractionalization with the canonical CRT for Dirac fermions in any spacetime dimension. The two possibilities correspond to the two choices of the squares of R and CT. The condition for canonical CRT implies that the square of R determines the square of CT.
    
\item 
In \Sec{sec:Majorana}, we study the group (called the R-T fractionalization) generated by R and T for Majorana fermions.
We determine the R-T fractionalization for Majorana fermions in any spacetime dimension $d=0,1,2,3,4\mod8$. See \Table{tab:CPRT-group-dd-Majorana}. We find that the R-T fractionalization for Majorana fermions is a subgroup of the first possibility for the C-R-T fractionalization with the canonical CRT for Dirac fermions for $d=0,1\mod8$. The R-T fractionalization for Majorana fermions is a subgroup of the second possibility for the C-R-T fractionalization with the canonical CRT for Dirac fermions for $d=3,4\mod8$.
There are two choices for the R-T fractionalization for Majorana fermions for $d=2\mod8$. The first choice is for the massive case and the second choice is for the massless case. For $d=2\mod8$, the first choice of the R-T fractionalization for Majorana fermions is a subgroup of the second possibility for the C-R-T fractionalization with the canonical CRT for Dirac fermions, the second choice of the R-T fractionalization for Majorana fermions is a subgroup of the first possibility for the C-R-T fractionalization with the canonical CRT for Dirac fermions.

\item
In \Sec{sec:Majorana}, we also study the group (called the C or T fractionalization) generated by C or T for Weyl fermions.
We determine the C or T fractionalization for Weyl fermions in any even spacetime dimension. See \Table{table:CT-group-dd-Weyl}. 
We find that the C or T fractionalization for Weyl fermions is a subgroup of the two possibilities for the C-R-T fractionalization with the canonical CRT for Dirac fermions in any even spacetime dimension.
For spacetime dimensions $d=2\mod4$, Weyl fermions only have C but no T or R symmetries, while for spacetime dimensions $d=0\mod4$, Weyl fermions only have T but no C or R symmetries. Majorana-Weyl fermions only exist for spacetime dimensions $d=2\mod8$ and Majorana-Weyl fermions only have trivial C but no T or R symmetries. See \Table{table:CPT}.

\item 
In \Sec{sec:mass}, we determine the maximal number of linearly independent Dirac and Majorana mass terms and construct them explicitly. We also study how the conventional Dirac and Majorana mass terms break the symmetries C, R, or T. See \Table{table:CRT-break-Dirac} and \Table{table:CRT-break-Majorana}. The conventional Dirac mass term breaks the R symmetry for all odd spacetime dimensions $d$, the conventional Dirac mass term breaks the C symmetry for spacetime dimensions $d=1\mod4$, and the conventional Dirac mass term breaks the T symmetry for spacetime dimensions $d=3\mod4$. 
A single Majorana fermion has to be massless for spacetime dimensions $d=0,1\mod8$.
The conventional Majorana mass term breaks the R and T symmetries for spacetime dimensions $d=3\mod8$.

\item
In \Sec{sec:domain-wall},
we study the domain wall dimensional reduction of the C-R-T fractionalization and R-T fractionalization. See \eqref{eq:reduction-odd} and \eqref{eq:reduction-even}. We also study the domain wall dimensional reduction of fermions. These results are consistent. See Theorem \ref{thm-DW}, \Fig{fig:domain-wall-fermion}, and \Table{table:domain-wall-fermion}. 

Dirac fermions become Weyl fermions in the domain wall dimensional reduction from odd $d$-dimensional spacetime to even $(d-1)$-dimensional spacetime. 
Majorana fermions become Majorana-Weyl fermions in the domain wall dimensional reduction from odd $d$-dimensional spacetime to even $(d-1)$-dimensional spacetime for $d=3\mod8$. There is no domain wall dimensional reduction from Majorana fermions in spacetime dimensions $d=1\mod8$ because a single Majorana fermion in spacetime dimensions $d=1\mod8$ has to be massless while the fermion in the bulk has to be massive.
The two possibilities for the C-R-T fractionalization with the canonical CRT for Dirac fermions in odd spacetime dimensions $d$ are reduced to the C or T fractionalization for Weyl fermions in even spacetime dimensions $d-1$.
The R-T fractionalization for Majorana fermions in spacetime dimensions $d=3\mod8$ is reduced to the trivial C fractionalization for Majorana-Weyl fermions in spacetime dimensions $d-1=2\mod8$. 

Dirac fermions become Dirac fermions in the domain wall dimensional reduction from even $d$-dimensional spacetime to odd $(d-1)$-dimensional spacetime. 
Majorana fermions become Majorana fermions in the domain wall dimensional reduction from even $d$-dimensional spacetime to odd $(d-1)$-dimensional spacetime for $d=2,4\mod8$. There is no domain wall dimensional reduction from Majorana fermions in spacetime dimensions $d=0\mod8$ because a single Majorana fermion in spacetime dimensions $d=0\mod8$ has to be massless while the fermion in the bulk has to be massive.
The two possibilities of the C-R-T fractionalization with the canonical CRT for Dirac fermions in even spacetime dimensions $d$ are reduced to the two possibilities of the C-R-T fractionalization with the canonical CRT for Dirac fermions in odd spacetime dimensions $d-1$ such that the C-R-T fractionalization with the canonical CRT for Dirac fermions in even spacetime dimensions $d$ and its domain wall dimensional reduction are the same.
The R-T fractionalization for Majorana fermions in spacetime dimensions $d=4\mod8$ is reduced to the R-T fractionalization for Majorana fermions in spacetime dimensions $d-1=3\mod8$. 
The first choice of the R-T fractionalization for Majorana fermions in spacetime dimensions $d=2\mod8$ is reduced to the R-T fractionalization for Majorana fermions in spacetime dimensions $d-1=1\mod8$. Only the first choice works because the first choice is for the massive case and the fermion in the bulk has to be massive.

We discuss possible further domain wall dimensional reduction of fermions in \Fig{fig:domain-wall-fermion-further}. We may combine the left and right Majorana-Weyl or Weyl fermions into a Majorana or Dirac fermion. We may also identify the Majorana and Weyl fermions in spacetime dimensions $d=4\mod8$.
For $d=0\mod4$, the matrix of the charge conjugation is block-off-diagonal in the Weyl representation (see \Table{tab:CPRT-trans-dd} and \App{app:Weyl}), namely $M_{\rC}=\begin{pmatrix}&M\\-M&\end{pmatrix}$ for some matrix $M$.
So a Majorana fermion is of the form $\psi=\begin{pmatrix}
    \psi_L\\\psi_R
\end{pmatrix}$
where $\psi_R=-M\psi_L^*$ and $\psi_L=M\psi_R^*$.
Therefore, a Majorana fermion can be identified with a Weyl fermion for $d=0\mod4$ if and only if $MM^*=-1$. In fact, for $d=2k+2$ ($k$ odd), $M=\eta_{\rC}\gamma^2\gamma^4\cdots\gamma^{2k}$ (see \Table{tab:CPRT-trans-dd} and \App{app:Weyl}), so 
\bea
MM^*=\gamma^2\gamma^4\cdots\gamma^{2k}(\gamma^2\gamma^4\cdots\gamma^{2k})^*=(-1)^{\frac{k+1}{2}}.
\eea
Therefore, $MM^*=-1$ if and only if $d=4\mod8$. That is, a Majorana fermion can be identified with a Weyl fermion for $d=4\mod8$.
Therefore, the domain wall dimensional reduction of fermions may be performed continuously.
\begin{figure}[!h]
\begin{center}
\begin{tikzpicture}
\matrix[matrix of math nodes,inner sep=1pt,row sep=1em,column sep=1em] (M)
{
d & \text{Majorana-Weyl} & \text{Majorana} & \text{Weyl} & \text{Dirac}\\
1 & & 1^{\R} & & 1^{\C}\\
2 & 1_L^{\R} & 1_L^{\R}\oplus 1_R^{\R} & 1_L^{\C} & 1_L^{\C}\oplus 1_R^{\C}\\
3 & & 2^{\R} &  & 2^{\C}\\
4 & & 4^{\R} & 2_L^{\C} &  2_L^{\C}\oplus 2_R^{\C}\\
5 & & & & 4^{\C} \\
6 & & & 4_L^{\C} & 4_L^{\C}\oplus 4_R^{\C}\\
7 & & & & 8^{\C} \\
8 & &  16^{\R} & 8_L^{\C} & 8_L^{\C}\oplus 8_R^{\C}\\
9 & & 16^{\R} &  & 16^{\C}\\
10 &  16_L^{\R} & 16_L^{\R}\oplus 16_R^{\R} & 16_L^{\C} & 16_L^{\C}\oplus 16_R^{\C}\\
11 & & 32^{\R} &  & 32^{\C}\\
12 & & 64^{\R} & 32_L^{\C} & 32_L^{\C}\oplus 32_R^{\C}\\
}
;
\draw[dashed] (M-3-2.east) -- (M-3-3.west);
\draw[dashed] (M-3-4.east) -- (M-3-5.west);

\draw (M-5-3.east) -- (M-5-4.west);
\draw[dashed] (M-5-4.east) -- (M-5-5.west);

\draw[dashed] (M-7-4.east) -- (M-7-5.west);

\draw[dashed] (M-9-4.east) -- (M-9-5.west);

\draw[dashed] (M-11-2.east) -- (M-11-3.west);
\draw[dashed] (M-11-4.east) -- (M-11-5.west);

\draw (M-13-3.east) -- (M-13-4.west);
\draw[dashed] (M-13-4.east) -- (M-13-5.west);

\draw[->] (M-3-3.north) -- (M-2-3.south);
\draw[->] (M-3-5.north) -- (M-2-5.south);
\draw[->] (M-4-3.north west) -- (M-3-2.south east);
\draw[->] (M-4-5.north west) -- (M-3-4.south east);
\draw[->] (M-5-3.north) -- (M-4-3.south);
\draw[->] (M-5-5.north) -- (M-4-5.south);
\draw[->] (M-6-5.north west) -- (M-5-4.south east);
\draw[->] (M-7-5.north) -- (M-6-5.south);
\draw[->] (M-8-5.north west) -- (M-7-4.south east);
\draw[->] (M-9-5.north) -- (M-8-5.south);
\draw[->] (M-10-5.north west) -- (M-9-4.south east);
\draw[->] (M-11-3.north) -- (M-10-3.south);
\draw[->] (M-11-5.north) -- (M-10-5.south);
\draw[->] (M-12-3.north west) -- (M-11-2.south east);
\draw[->] (M-12-5.north west) -- (M-11-4.south east);
\draw[->] (M-13-3.north) -- (M-12-3.south);
\draw[->] (M-13-5.north) -- (M-12-5.south);

\end{tikzpicture}
\end{center}
\caption{Further domain wall dimensional reduction of fermions. The target of the arrow is the domain wall of the source of the arrow. The number indicates the dimension of the representation of the fermion, the upper index indicates that the representation is real or complex, and the lower index indicates the chirality of the fermion.
The dashed lines (- - -) connecting Majorana-Weyl and Majorana fermions or 
connecting Weyl and Dirac fermions mean that the left and right Majorana-Weyl or Weyl fermions can be combined into a Majorana or Dirac fermion. The solid lines (---) connecting Majorana and Weyl fermions in spacetime dimensions $d=4\mod8$ mean that Majorana and Weyl fermions in spacetime dimensions $d=4\mod8$ can be identified. 
}
\label{fig:domain-wall-fermion-further}
\end{figure}
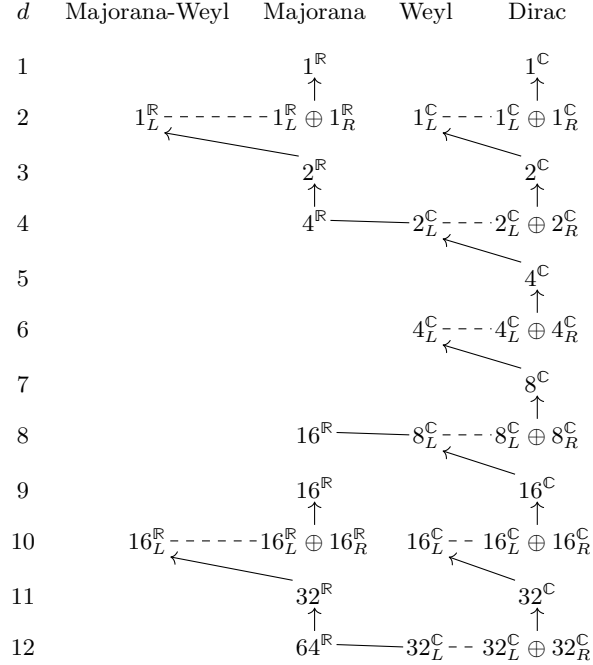

\item Although we study C-R-T fractionalization of fermions,
we only focus on a single-particle fermion theory in the Lagrangian formulation.
In fact, when there are multi-fermions in the many-body quantum systems 
\cite{Gu1308.2488, PrakashJW2011.12320, PrakashJW2011.13921, TurzilloYou2012.04621},
the C-R-T fractionalization can show further richer structures such that the 
supersymmetry comes in naturally. It will be interesting to explore the 
C-R-T fractionalization in many-body quantum systems in any dimensions systematically.

\end{itemize}

\section{Acknowledgments}

\noindent
\emph{Acknowledgments} ---
We thank Sergio Cecotti, Song Cheng, Wei Cui, Pierre Deligne, Dan Freed, Xin Fu, Jun Hou Fung, Yang-Yang Li, Pavel Putrov, Ryan Thorngren, Jie Wang, and Martin Zirnbauer for their helpful comments. 
JW also thanks Abhishodh Prakash for the past collaborations on the fractionalization of time-reversal T symmetry in \cite{PrakashJW2011.12320, PrakashJW2011.13921}.
ZW is supported by the NSFC Grant No. 12405001.
JW is supported by Harvard University CMSA.
YZY is supported by the National Science
Foundation Grant No. DMR-2238360.

\appendix

\section{Chiral representation in any $d$ (Weyl representation in even $d$)}\label{app:Weyl}

\cred{Gamma matrices represent the real Clifford algebra $\Cl_{d-1,1}$ satisfying the
anticommutator $\{\Gamma^{\mu},\Gamma^{\nu}\}=2\eta^{\mu\nu}$.
We consider the Lorentz signature $\eta^{\mu\nu}=\diag(+,-,-,\dots,-)$.
Let $d$ denote the total spacetime dimension. 
We shall use $\mu=0,1,2,\dots, d-1$ to denote the spacetime index and $j=1,2,\dots,d-1$ to denote the space index.}

For $d=2k+2$, we write 
\bea\label{eq:Gamma}
\Gamma^0=\left(\begin{array}{cc}0&I\\I&0\end{array}\right),\quad
\Gamma^j=\left(\begin{array}{cc}0&\gamma^j\\-\gamma^j&0\end{array}\right).
\eea
These Gamma matrices are called the chiral representation (also called the Weyl basis for even $d$), we will construct the Majorana representation for the Majorana fermions later in \App{app:Majorana}.
Here $I$ is an identity matrix and the matrices $\gamma^j$ are to be determined.
Then \eqref{eq:anticommutator} becomes
\bea
\{\gamma^i,\gamma^j\}=2\delta^{ij}.
\eea
The matrices $\gamma^j$ can be constructed in terms of the Pauli matrices explicitly.

The Pauli matrices are
\bea
\sigma^1=\left(\begin{array}{cc}0&1\\1&0\end{array}\right),\quad
\sigma^2=\left(\begin{array}{cc}0&-\ii\\\ii&0\end{array}\right),\quad
\sigma^3=\left(\begin{array}{cc}1&0\\0&-1\end{array}\right).
\eea

Then the matrices $\gamma^j$ are \cite{Murayama2007NotesOC}
\bea\label{eq:gamma}
\gamma^1&=&\sigma^1\otimes I\otimes\cdots\otimes I\otimes I,\nn\\
\gamma^2&=&\sigma^2\otimes I\otimes\cdots\otimes I\otimes I,\nn\\
\gamma^3&=&\sigma^3\otimes\sigma^1\otimes\cdots\otimes I\otimes I,\nn\\
&\vdots&\nn\\
\gamma^{2k-1}&=&\sigma^3\otimes\sigma^3\otimes\cdots\otimes\sigma^3\otimes\sigma^1,\nn\\
\gamma^{2k}&=&\sigma^3\otimes\sigma^3\otimes\cdots\otimes\sigma^3\otimes\sigma^2,\nn\\
\gamma^{2k+1}&=&\sigma^3\otimes\sigma^3\otimes\cdots\otimes\sigma^3\otimes\sigma^3=\ii^{-k}\gamma^1\gamma^2\cdots\gamma^{2k}.
\eea
Here $I$ is the $2\times2$ identity matrix.

For $d=2k+3$, we add $\Gamma^{2k+2}=\ii\diag(-I,I)$ to the Gamma matrices for $d=2k+2$, then we get the Gamma matrices for $d=2k+3$.
Here $\diag(-I,I)=\ii^{k}\Gamma^0\Gamma^1\cdots\Gamma^{2k+1}$ and $I$ is an identity matrix.
 In this case, the image of $\omega$ in the representation is $\ii^{k+1}\Gamma^0\Gamma^1\cdots\Gamma^{2k+2}=(\Gamma^{2k+2})^2=-1$. This is consistent with Definition \ref{def-Dirac'}.

It is useful to note that in the chiral representation, $\Gamma^0$ and $\Gamma^j$ ($j$ odd) are real, while $\Gamma^j$ ($j$ even) are imaginary; $\Gamma^0$ and $\Gamma^j$ ($j$ even) are symmetric, while $\Gamma^j$ ($j$ odd) are anti-symmetric. Therefore, $(\Gamma^0)^{\dagger}=\Gamma^0$, while $(\Gamma^j)^{\dagger}=-\Gamma^j$.

\section{Majorana representation}\label{app:Majorana}

Now we construct the Majorana representation (the Gamma matrices that are all imaginary or real) explicitly. Namely, we construct imaginary Gamma matrices for the massive case (for $d=2,3,4\mod8$) and real Gamma matrices for the massless case (for $d=0,1,2\mod8$). In the following construction, $I$ is the $2\times2$ identity matrix.

For $d=2$, we can choose imaginary Gamma matrices
\bea
\Gamma^0&=&\sigma^2,\nn\\
\Gamma^1&=&\ii\sigma^1.
\eea
We can also choose real Gamma matrices
\bea
\Gamma^0&=&\sigma^1,\nn\\
\Gamma^1&=&\ii\sigma^2.
\eea

For $d=3$, we can choose imaginary Gamma matrices
\bea
\Gamma^0&=&\sigma^2,\nn\\
\Gamma^1&=&\ii\sigma^1,\nn\\
\Gamma^2&=&\ii\sigma^3.
\eea
For $d=4$, we can choose imaginary Gamma matrices
\bea
\Gamma^0&=&\sigma^2\otimes\sigma^1,\nn\\
\Gamma^1&=&\ii\sigma^2\otimes\sigma^2,\nn\\
\Gamma^2&=&\ii\sigma^1\otimes I,\nn\\
\Gamma^3&=&\ii\sigma^3\otimes I.
\eea
For $d=8$, we can choose real Gamma matrices
\bea
\Gamma^0&=&\sigma^2\otimes\sigma^1\otimes I\otimes\sigma^2,\nn\\
\Gamma^1&=&\sigma^1\otimes I\otimes I\otimes\ii\sigma^2,\nn\\
\Gamma^2&=&\sigma^3\otimes I\otimes I\otimes\ii\sigma^2,\nn\\
\Gamma^3&=&\ii\sigma^2\otimes\sigma^2\otimes I\otimes\sigma^2,\nn\\
\Gamma^4&=&\ii\sigma^2\otimes\sigma^3\otimes\sigma^2\otimes\sigma^2,\nn\\
\Gamma^5&=&\ii\sigma^2\otimes\sigma^3\otimes\sigma^1\otimes I,\nn\\
\Gamma^6&=&\ii\sigma^2\otimes\sigma^3\otimes\sigma^3\otimes I,\nn\\
\Gamma^7&=&I\otimes I\otimes\ii\sigma^2\otimes\sigma^1.
\eea
For $d=9$, we can choose real Gamma matrices
\bea
\Gamma^0&=&\sigma^2\otimes\sigma^1\otimes I\otimes\sigma^2,\nn\\
\Gamma^1&=&\sigma^1\otimes I\otimes I\otimes\ii\sigma^2,\nn\\
\Gamma^2&=&\sigma^3\otimes I\otimes I\otimes\ii\sigma^2,\nn\\
\Gamma^3&=&\ii\sigma^2\otimes\sigma^2\otimes I\otimes\sigma^2,\nn\\
\Gamma^4&=&\ii\sigma^2\otimes\sigma^3\otimes\sigma^2\otimes\sigma^2,\nn\\
\Gamma^5&=&\ii\sigma^2\otimes\sigma^3\otimes\sigma^1\otimes I,\nn\\
\Gamma^6&=&\ii\sigma^2\otimes\sigma^3\otimes\sigma^3\otimes I,\nn\\
\Gamma^7&=&I\otimes I\otimes\ii\sigma^2\otimes\sigma^1,\nn\\
\Gamma^8&=&I\otimes I\otimes\ii\sigma^2\otimes\sigma^3.
\eea

For $d=10$, we can choose imaginary Gamma matrices
\bea
\Gamma^0&=&\sigma^2\otimes\sigma^2\otimes\sigma^1\otimes I\otimes\sigma^2,\nn\\
\Gamma^1&=&\sigma^2\otimes\sigma^1\otimes I\otimes I\otimes\ii\sigma^2,\nn\\
\Gamma^2&=&\sigma^2\otimes\sigma^3\otimes I\otimes I\otimes\ii\sigma^2,\nn\\
\Gamma^3&=&\sigma^2\otimes\ii\sigma^2\otimes\sigma^2\otimes I\otimes\sigma^2,\nn\\
\Gamma^4&=&\sigma^2\otimes\ii\sigma^2\otimes\sigma^3\otimes\sigma^2\otimes\sigma^2,\nn\\
\Gamma^5&=&\sigma^2\otimes\ii\sigma^2\otimes\sigma^3\otimes\sigma^1\otimes I,\nn\\
\Gamma^6&=&\sigma^2\otimes\ii\sigma^2\otimes\sigma^3\otimes\sigma^3\otimes I,\nn\\
\Gamma^7&=&\sigma^2\otimes I\otimes I\otimes\ii\sigma^2\otimes\sigma^1,\nn\\
\Gamma^8&=&\sigma^2\otimes I\otimes I\otimes\ii\sigma^2\otimes\sigma^3,\nn\\
\Gamma^9&=&\ii\sigma^1\otimes I\otimes I\otimes I\otimes I.
\eea
We can also choose real Gamma matrices
\bea
\Gamma^0&=&\sigma^1\otimes\sigma^2\otimes\sigma^1\otimes I\otimes\sigma^2,\nn\\
\Gamma^1&=&\sigma^1\otimes\sigma^1\otimes I\otimes I\otimes\ii\sigma^2,\nn\\
\Gamma^2&=&\sigma^1\otimes\sigma^3\otimes I\otimes I\otimes\ii\sigma^2,\nn\\
\Gamma^3&=&\sigma^1\otimes\ii\sigma^2\otimes\sigma^2\otimes I\otimes\sigma^2,\nn\\
\Gamma^4&=&\sigma^1\otimes\ii\sigma^2\otimes\sigma^3\otimes\sigma^2\otimes\sigma^2,\nn\\
\Gamma^5&=&\sigma^1\otimes\ii\sigma^2\otimes\sigma^3\otimes\sigma^1\otimes I,\nn\\
\Gamma^6&=&\sigma^1\otimes\ii\sigma^2\otimes\sigma^3\otimes\sigma^3\otimes I,\nn\\
\Gamma^7&=&\sigma^1\otimes I\otimes I\otimes\ii\sigma^2\otimes\sigma^1,\nn\\
\Gamma^8&=&\sigma^1\otimes I\otimes I\otimes\ii\sigma^2\otimes\sigma^3,\nn\\
\Gamma^9&=&\ii\sigma^2\otimes I\otimes I\otimes I\otimes I.
\eea

For $d=11$, we can choose imaginary Gamma matrices
\bea
\Gamma^0&=&\sigma^2\otimes\sigma^2\otimes\sigma^1\otimes I\otimes\sigma^2,\nn\\
\Gamma^1&=&\sigma^2\otimes\sigma^1\otimes I\otimes I\otimes\ii\sigma^2,\nn\\
\Gamma^2&=&\sigma^2\otimes\sigma^3\otimes I\otimes I\otimes\ii\sigma^2,\nn\\
\Gamma^3&=&\sigma^2\otimes\ii\sigma^2\otimes\sigma^2\otimes I\otimes\sigma^2,\nn\\
\Gamma^4&=&\sigma^2\otimes\ii\sigma^2\otimes\sigma^3\otimes\sigma^2\otimes\sigma^2,\nn\\
\Gamma^5&=&\sigma^2\otimes\ii\sigma^2\otimes\sigma^3\otimes\sigma^1\otimes I,\nn\\
\Gamma^6&=&\sigma^2\otimes\ii\sigma^2\otimes\sigma^3\otimes\sigma^3\otimes I,\nn\\
\Gamma^7&=&\sigma^2\otimes I\otimes I\otimes\ii\sigma^2\otimes\sigma^1,\nn\\
\Gamma^8&=&\sigma^2\otimes I\otimes I\otimes\ii\sigma^2\otimes\sigma^3,\nn\\
\Gamma^9&=&\ii\sigma^1\otimes I\otimes I\otimes I\otimes I,\nn\\
\Gamma^{10}&=&\ii\sigma^3\otimes I\otimes I\otimes I\otimes I.
\eea

For $d=12$, we can choose imaginary Gamma matrices
\bea
\Gamma^0&=&\sigma^2\otimes\sigma^1\otimes\sigma^2\otimes\sigma^1\otimes I\otimes\sigma^2,\nn\\
\Gamma^1&=&\sigma^2\otimes\sigma^1\otimes\sigma^1\otimes I\otimes I\otimes\ii\sigma^2,\nn\\
\Gamma^2&=&\sigma^2\otimes\sigma^1\otimes\sigma^3\otimes I\otimes I\otimes\ii\sigma^2,\nn\\
\Gamma^3&=&\sigma^2\otimes\sigma^1\otimes\ii\sigma^2\otimes\sigma^2\otimes I\otimes\sigma^2,\nn\\
\Gamma^4&=&\sigma^2\otimes\sigma^1\otimes\ii\sigma^2\otimes\sigma^3\otimes\sigma^2\otimes\sigma^2,\nn\\
\Gamma^5&=&\sigma^2\otimes\sigma^1\otimes\ii\sigma^2\otimes\sigma^3\otimes\sigma^1\otimes I,\nn\\
\Gamma^6&=&\sigma^2\otimes\sigma^1\otimes\ii\sigma^2\otimes\sigma^3\otimes\sigma^3\otimes I,\nn\\
\Gamma^7&=&\sigma^2\otimes\sigma^1\otimes I\otimes I\otimes\ii\sigma^2\otimes\sigma^1,\nn\\
\Gamma^8&=&\sigma^2\otimes\sigma^1\otimes I\otimes I\otimes\ii\sigma^2\otimes\sigma^3,\nn\\
\Gamma^9&=&\ii\sigma^2\otimes\sigma^2\otimes I\otimes I\otimes I\otimes I,\nn\\
\Gamma^{10}&=&\ii\sigma^1\otimes I\otimes I\otimes I\otimes I\otimes I,\nn\\
\Gamma^{11}&=&\ii\sigma^3\otimes I\otimes I\otimes I\otimes I\otimes I.
\eea

We can construct the Majorana representation for all $d=0,1,2,3,4\mod8$ but we omit the explicit form for higher $d$ here. The readers should already have found the laws.

\section{Charge-like and isospin-like U(1) symmetries}\label{app:further}

Now we take the internal charge-like $\U(1)_{\hat{q}}$ symmetry (known also as the symmetry of topological insulators) $\U_{\theta}:\psi\mapsto\e^{\ii \hat{q}\theta}\psi$ into account where $\hat{q}$ is the charge operator of $\psi$. This $\U(1)_{\hat{q}}$ symmetry contains the fermion parity $\Z_2^{\rF}$: $(-1)^{\rF}=\U_{\frac{\pi}{\hat{q}}}$.
Recall that $\rC:\psi\mapsto M_{\rC}\psi^*$, 
$\rR_j:\psi\mapsto M_{\rR_j}\psi(t,x_1,\dots,x_{j-1},-x_j,x_{j+1},\dots,x_{d-1})$,
and $\rT:\psi\mapsto M_{\rT}\psi(-t,x)$.

Since C is unitary and linear, C acts on $\U(1)_{\hat{q}}$ non-trivially,
\bea
\rC\U_{\theta}\rC^{-1}=\U_{-\theta}. 
\eea
\bea
\xymatrix{\psi\ar[r]^-{\rC}\ar[d]_{\U_{\theta}}&M_{\rC}\psi^*\ar[d]^{\U_{-\theta}}\\\e^{\ii \hat{q}\theta}\psi\ar[r]^-{\rC}&\e^{\ii \hat{q}\theta}M_{\rC}\psi^*}
\eea
Note that the charge operator of $\psi^*$ is $-\hat{q}$.

Since $\rR_j$ is unitary and linear,
$\rR_j$ acts on $\U(1)_{\hat{q}}$ trivially,
\bea
\rR_j\U_{\theta}\rR_j^{-1}=\U_{\theta}. 
\eea
\bea
\xymatrix{\psi\ar[r]^-{\rR_j}\ar[d]_{\U_{\theta}}&M_{\rR_j}\psi(t,x_1,\dots,x_{j-1},-x_j,x_{j+1},\dots,x_{d-1})\ar[d]^{\U_{\theta}}\\\e^{\ii \hat{q}\theta}\psi\ar[r]^-{\rR_j}&\e^{\ii \hat{q}\theta}M_{\rR_j}\psi(t,x_1,\dots,x_{j-1},-x_j,x_{j+1},\dots,x_{d-1}).}
\eea

Since T is anti-unitary and anti-linear,
T acts on $\U(1)_{\hat{q}}$ non-trivially,
\bea
\rT\U_{\theta}\rT^{-1}=\U_{-\theta}. 
\eea
\bea
\xymatrix{\psi\ar[r]^-{\rT}\ar[d]_{\U_{\theta}}&M_{\rT}\psi(-t,x)\ar[d]^{\U_{-\theta}}\\\e^{\ii \hat{q}\theta}\psi\ar[r]^-{\rT}&\e^{-\ii \hat{q}\theta}M_{\rT}\psi(-t,x).}
\eea

 For each $d$, the full group structure generated by C, R$_j$, T, and $\U_{\theta}$ which is the total group in the group extension
 \bea
1\to \U(1)_{\hat{q}}\to G\to \Z_2^{\rC}\times\Z_2^{\rR_j}\times\Z_2^{\rT}\to1
 \eea
 is $\tilde{G}_{\psi}\ltimes_{\Z_2^{\rF}} \U(1)_{\hat{q}}$ where $\tilde{G}_{\psi}$ is the group in \Table{tab:CPRT-group-dd-1} and \Table{tab:CPRT-group-dd-2}, R$_j$ acts on $\U(1)_{\hat{q}}$ trivially, while C and T act on $\U(1)_{\hat{q}}$ non-trivially.

Now we take the internal isospin-like $\U(1)_{\hat{s}}$ symmetry (known also as the symmetry of topological superconductors) $\U_{\theta}':\psi\mapsto\e^{\ii \hat{s}\theta}\psi$ into account where $\hat{s}$ is the spin operator of $\psi$. This $\U(1)_{\hat{s}}$ symmetry contains the fermion parity $\Z_2^{\rF}$: $(-1)^{\rF}=\U'_{\frac{\pi}{\hat{s}}}$.
Recall that $\rC:\psi\mapsto M_{\rC}\psi^*$, 
$\rR_j:\psi\mapsto M_{\rR_j}\psi(t,x_1,\dots,x_{j-1},-x_j,x_{j+1},\dots,x_{d-1})$,
and $\rT:\psi\mapsto M_{\rT}\psi(-t,x)$.

Since C is unitary and linear, C acts on $\U(1)_{\hat{s}}$ trivially,
\bea
\rC\U_{\theta}'\rC^{-1}=\U_{\theta}'. 
\eea
\bea
\xymatrix{\psi\ar[r]^-{\rC}\ar[d]_{\U_{\theta}'}&M_{\rC}\psi^*\ar[d]^{\U_{\theta}'}\\\e^{\ii \hat{s}\theta}\psi\ar[r]^-{\rC}&\e^{\ii \hat{s}\theta}M_{\rC}\psi^*}
\eea

Since $\rR_j$ is unitary and linear,
$\rR_j$ acts on $\U(1)_{\hat{s}}$ trivially,
\bea
\rR_j\U_{\theta}'\rR_j^{-1}=\U_{\theta}'. 
\eea
\bea
\xymatrix{\psi\ar[r]^-{\rR_j}\ar[d]_{\U_{\theta}'}&M_{\rR_j}\psi(t,x_1,\dots,x_{j-1},-x_j,x_{j+1},\dots,x_{d-1})\ar[d]^{\U_{\theta}'}\\\e^{\ii \hat{s}\theta}\psi\ar[r]^-{\rR_j}&\e^{\ii \hat{s}\theta}M_{\rR_j}\psi(t,x_1,\dots,x_{j-1},-x_j,x_{j+1},\dots,x_{d-1}).}
\eea

Since T is anti-unitary and anti-linear,
T acts on $\U(1)_{\hat{s}}$ trivially,
\bea
\rT\U_{\theta}'\rT^{-1}=\U_{\theta}'. 
\eea
\bea
\xymatrix{\psi\ar[r]^-{\rT}\ar[d]_{\U_{\theta}'}&M_{\rT}\psi(-t,x)\ar[d]^{\U_{\theta}'}\\\e^{\ii \hat{s}\theta}\psi\ar[r]^-{\rT}&\e^{-\ii \hat{s}\theta}M_{\rT}\psi(-t,x).}
\eea
Note that the spin operator of $\psi(-t,x)$ is $-\hat{s}$.

 For each $d$, the full group structure generated by C, R$_j$, T, and $\U_{\theta}'$ which is the total group in the group extension
 \bea
1\to \U(1)_{\hat{s}}\to G'\to \Z_2^{\rC}\times\Z_2^{\rR_j}\times\Z_2^{\rT}\to1
 \eea
 is $\tilde{G}_{\psi}\times_{\Z_2^{\rF}} \U(1)_{\hat{s}}$ where $\tilde{G}_{\psi}$ is the group in \Table{tab:CPRT-group-dd-1} and \Table{tab:CPRT-group-dd-2}, C, R$_j$, and T act on $\U(1)_{\hat{s}}$ trivially.

 Note that both the internal charge-like $\U(1)_{\hat{q}}$ symmetry and the internal isospin-like $\U(1)_{\hat{s}}$ symmetry commute with the symmetry CR$_j$T. This is consistent with \eqref{eq:U-CPT}.

\bibliography{CPTbib.bib}

\end{document}